\documentclass{lmcs}
\pdfoutput=1

\usepackage{lastpage}
\lmcsdoi{15}{2}{13}
\lmcsheading{}{\pageref{LastPage}}{}{}%
{Mar.~13,~2018}{May~10,~2019}{}


\usepackage[T1]{fontenc}

\usepackage[hypcap=true]{caption}

\usepackage{algorithm}
\usepackage[noend]{algpseudocode}

\usepackage{tikz}
\usetikzlibrary{arrows}
\usetikzlibrary{automata}
\usetikzlibrary{shadows}
\usetikzlibrary{positioning}

\newcommand*\circled[1]{\tikz[baseline=(char.base)]{ % chktex 36
            \node[shape=circle,draw,solid,fill=white,inner sep=2pt] (char) {#1};}}

\usepackage{ltl}
\usepackage{centernot}

\newcommand*{\eventually}{\LTLdiamond}
\newcommand*{\globally}{\LTLsquare}
\newcommand*{\once}{\LTLdiamondminus}
\newcommand*{\henceforth}{\LTLsquareminus}
\newcommand*{\nextx}{\LTLcircle}

\newcommand*{\until}{\mathbin{\mathcal{U}}}
\newcommand*{\since}{\mathbin{\mathcal{S}}}

\newcommand*{\modelsfs}{\mathbin{\models_{\mkern-6mu f}^{\mkern-6mu +}}}
\newcommand*{\modelsfn}{\mathbin{\models_{\mkern-6mu f}^{\mkern-6mu \phantom{+}}}}
\newcommand*{\modelsfw}{\mathbin{\models_{\mkern-6mu f}^{\mkern-6mu -}}}

\newcommand*{\first}{\mathcal{B}^\rightarrow}
\newcommand*{\pfirst}{\mathcal{B}^\leftarrow}

\newcommand*{\guntil}{\mathbin{\mathfrak{U}}}
\newcommand*{\gsince}{\mathbin{\mathfrak{S}}}

\newcommand*{\ltl}{\textup{\textmd{$\textsf{LTL}_\textsf{fut}$}}}
\newcommand*{\ltlpast}{\textup{\textmd{\textsf{LTL}}}}

\newcommand*{\mtl}{\textup{\textmd{$\textsf{MTL}_\textsf{fut}$}}}
\newcommand*{\mtlunit}{\textup{\textmd{$\textsf{MTL}_\textsf{fut}$[$\Phi_{\textit{int}}$]}}}
\newcommand*{\mtlpast}{\textup{\textmd{\textsf{MTL}}}}
\newcommand*{\mtlpastcnt}{\textup{\textmd{\textsf{MTL[$\{C_n, \overset{\leftarrow}{C}_n\}_{n=2}^\infty$]}}}} % chktex 3
\newcommand*{\mtlpastg}{\textup{\textmd{\textsf{MTL[\texorpdfstring{$\guntil, \gsince$}{G, U}]}}}}

\newcommand*{\mtlbpast}{\textup{\textmd{\textsf{MTL[$\mathcal{B}^\leftrightarrows$]}}}}

\newcommand*{\fo}{\textup{\textmd{\textsf{FO[$<, +\mathbb{Q}$]}}}}
\newcommand*{\foone}{\textup{\textmd{\textsf{FO[\texorpdfstring{$<, +1$}{<, +1}]}}}}

\newcommand*{\mitl}{\textup{\textmd{$\textsf{MITL}_\textsf{fut}$}}}
\newcommand*{\mitlpast}{\textup{\textmd{\textsf{MITL}}}}

\newcommand{\ropen}[1]{[#1)} % chktex 9
\newcommand{\lopen}[1]{(#1]} % chktex 9
\newcommand{\maxit}{\mathit{max}} % chktex 35

\theoremstyle{defC}
\newtheorem{exaC}[thm]{Example}

\begin{document}

\title{On the Expressiveness and Monitoring of Metric Temporal Logic}
\titlecomment{Part of this work appeared in the \emph{Proceedings of the 8th International Workshop on Reachability Problems} (2014), and in the
\emph{Proceedings of the 5th International Conference on Runtime Verification} (2014). \\
\indent
Jo\"el Ouaknine was supported by ERC grant AVS-ISS (648701) and by DFG grant 389792660 as part of TRR 248 (see \url{https://perspicuous-computing.science}). James Worrell was supported by EPSRC Fellowship EP/N008197/1.}

\author[H.-M.~Ho]{Hsi-Ming~Ho\rsuper{a}}

\author[J.~Ouaknine]{Jo\"el~Ouaknine\rsuper{b}}

\author[J.~Worrell]{James~Worrell\rsuper{a}}

\address{\lsuper{a}Department of Computer Science, University of Oxford, Oxford, UK}
\email{hsimho@gmail.com, jbw@cs.ox.ac.uk}

\address{\lsuper{b}Max Planck Institute for Software Systems, Saarland Informatics Campus, Germany}
\email{joel@mpi-sws.org}

\begin{abstract}
It is known that \emph{Metric Temporal Logic} (\mtlpast{}) is strictly
less expressive than the \emph{Monadic First-Order Logic of Order and Metric}
(\foone{}) when interpreted over timed words; this remains true even when the time domain is
bounded \emph{a priori}.
In this work, we present an extension of \mtlpast{} with
the same expressive power as \foone{} over bounded timed words (and also,
trivially, over time-bounded signals).
We then show that expressive completeness also holds
in the general (time-unbounded) case if we allow the use of rational constants $q \in \mathbb{Q}$ in formulas.
This extended version of \mtlpast{} therefore yields a definitive real-time
analogue of Kamp's theorem.
As an application, we propose a \emph{trace-length independent} monitoring
procedure for our extension of \mtlpast{},
the first such procedure in a dense real-time setting.
\end{abstract}

\maketitle

\section{Introduction}

\paragraph{\emph{Expressiveness of metric temporal logics}}
One of the most prominent specification formalisms used in
verification is \emph{Linear Temporal Logic} (\ltlpast{}), which is
typically interpreted over the non-negative integers or reals. A
celebrated result of Kamp~\cite{Kamp1968} states that, in either case,
\ltlpast{} has precisely the same expressive power as the \emph{Monadic
   First-Order Logic of Order} ($\mathsf{FO[<]}$). These logics,
however, are inadequate to express specifications for systems whose
correct behaviour depends on quantitative timing requirements. Over
the last three decades, much work has therefore gone into lifting
classical verification formalisms and results to the real-time
setting. \emph{Metric Temporal Logic} (\mtlpast{}),\footnote{In this article, we refer to the logic with
constrained `Until' and `Since' modalities exclusively as `\mtlpast{}',
and use the term `metric temporal logics'
in a broader sense to refer to temporal logics with modalities definable in \foone{} (see below).}
which extends \ltlpast{} by constraining the modalities by time intervals, was introduced
by Koymans~\cite{Koymans1990} in 1990 and has emerged as a central
real-time specification formalism.
\mtlpast{} enjoys two main semantics, depending intuitively on whether
atomic formulas are interpreted as \emph{state predicates} or as
(instantaneous) \emph{events}. In the former, the system is assumed to
be under observation at every instant in time, leading to a
`continuous' semantics based on \emph{signals},
whereas in the latter, observations of the system are taken to be
(finite or infinite) sequences of timestamped snapshots, leading to a
`pointwise' semantics based on \emph{timed words}---this is the prevalent
interpretation for systems modelled as timed automata~\cite{Alur1994}.
In both cases, the time domain is usually taken to be the non-negative
real numbers. Both semantics have been extensively studied; see,
e.g.,~\cite{Ouaknine2008} for a historical account.

Alongside these developments, researchers proposed the \emph{Monadic
   First-Order Logic of Order and Metric \emph{($\foone{}$)}} as a natural quantitative extension of $\mathsf{FO[<]}$.
Like \mtlpast{}, \foone{} can be interpreted over signals~\cite{Hirshfeld2004} or timed words~\cite{Wilke1994}.
An obvious question to ask is whether
\mtlpast{} has the same expressive power as \foone{}, i.e.,~an analogue of Kamp's theorem holds in the real-time setting.
Unfortunately, Hirshfeld
and Rabinovich~\cite{Hirshfeld2007} showed that no `finitary'
extension of \mtlpast{}---and \emph{a fortiori} \mtlpast{} itself---could have
the same expressive power as $\foone{}$ over the reals.\footnote{Hirshfeld and Rabinovich's result was only stated and
proved for the continuous semantics, but we believe that their
approach would also carry through for the pointwise semantics.
In any case, using different techniques Prabhakar and D'Souza~\cite{Prabhakar2006}
and Pandya and Shah~\cite{Pandya2011}
independently showed that \mtlpast{} is
strictly weaker than \foone{} in the pointwise semantics.}
Still, in the continuous semantics, \mtlpast{} can be made expressively complete for \foone{}
by extending the logic with an infinite family of `\emph{counting modalities}'~\cite{Hunter2013}
or considering only \emph{bounded} time domains~\cite{Ouaknine2009,Ouaknine2010}.
Nonetheless, and rather surprisingly, \mtlpast{} with counting modalities
remains strictly less expressive than $\foone{}$ over bounded time domains in the pointwise semantics, i.e.,~over timed words of
bounded duration.

\paragraph{\emph{Monitoring of real-time specifications}}
In recent years, \emph{runtime verification} (see~\cite{Leucker2009, Sokolsky2011} for surveys) has emerged as a light-weight complementary technique
to \emph{model checking}~\cite{Clarke1981, Queille1982}.
It is particularly useful for systems whose
internal details are either too complex to be modelled faithfully
or simply not accessible.
Roughly speaking, while in model checking one considers all behaviours of the model,
in runtime verification one focusses on one particular behaviour---the current one.
Given a specification $\varphi$ and a finite timed word $\rho$ (which we call a finite \emph{trace} in this context),
the \emph{prefix} problem asks whether all infinite traces
extending $\rho$ satisfy $\varphi$.
The \emph{monitoring} problem, as far as we are concerned here,
can be seen as an \emph{online}
version of the prefix problem where $\rho$ grows incrementally (i.e.,~one event at a time):
the monitoring procedure (\emph{monitor}) for $\varphi$ is executed in parallel with the system under scrutiny,
and it is required to output an answer when either (i)~all infinite extensions of the current trace satisfy $\varphi$, or (ii)~no infinite extension of the
current trace can satisfy $\varphi$.

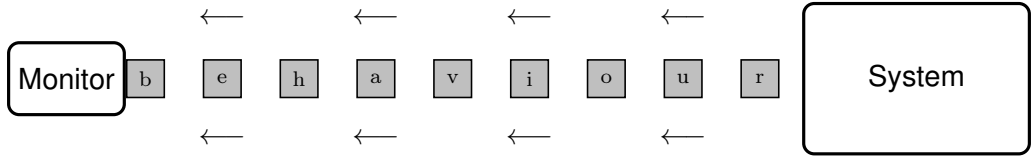
\begin{figure}[h]
\centering
\begin{tikzpicture}[->,>=stealth', auto, node distance=0.8cm,
                    semithick, bend angle=25, every state/.style={fill=none,draw=black,text=black,shape=rectangle}]

\node[state]					(v2) [very thick, minimum width=3cm, minimum height=2cm, rounded corners]	{\textsf{System}};

\node[state]					(vr) [fill=lightgray, left=0.3cm of v2, minimum width=0.5cm, minimum height=0.5cm] {\scriptsize r};
\node[state]					(vu) [fill=lightgray, left=0.5cm of vr, minimum width=0.5cm, minimum height=0.5cm] {\scriptsize u};
\node[state]					(vu') [draw=none, above of=vu] {$\longleftarrow$};
\node[state]					(vu') [draw=none, below of=vu] {$\longleftarrow$};
\node[state]					(vo) [fill=lightgray, left=0.5cm of vu, minimum width=0.5cm, minimum height=0.5cm] {\scriptsize o};
\node[state]					(vi) [fill=lightgray, left=0.5cm of vo, minimum width=0.5cm, minimum height=0.5cm] {\scriptsize i};
\node[state]					(vi') [draw=none, above of=vi] {$\longleftarrow$};
\node[state]					(vi') [draw=none, below of=vi] {$\longleftarrow$};
\node[state]					(vv) [fill=lightgray, left=0.5cm of vi, minimum width=0.5cm, minimum height=0.5cm] {\scriptsize v};
\node[state]					(va) [fill=lightgray, left=0.5cm of vv, minimum width=0.5cm, minimum height=0.5cm] {\scriptsize a};
\node[state]					(va') [draw=none, above of=va] {$\longleftarrow$};
\node[state]					(va') [draw=none, below of=va] {$\longleftarrow$};
\node[state]					(vh) [fill=lightgray, left=0.5cm of va, minimum width=0.5cm, minimum height=0.5cm] {\scriptsize h};
\node[state]					(ve) [fill=lightgray, left=0.5cm of vh, minimum width=0.5cm, minimum height=0.5cm] {\scriptsize e};
\node[state]					(ve') [draw=none, above of=ve] {$\longleftarrow$};
\node[state]					(ve') [draw=none, below of=ve] {$\longleftarrow$};
\node[state]					(vb) [fill=lightgray, left=0.5cm of ve, minimum width=0.5cm, minimum height=0.5cm] {\scriptsize b};

\node[state]					(v1) [very thick, left=0cm of vb, rounded corners]	{\textsf{Monitor}};

\end{tikzpicture}
\caption{A monitor receives the trace incrementally (one event at a time).}%
\label{fig:mon-example}
\end{figure}

Ideally, we would also like to require a monitoring procedure to be \emph{trace-length independent}~\cite{Rosu2012, Bauer2013} in the sense
that its time and space requirements should not depend on the length of
the input trace (this is important since input traces in practical applications tend to be very long; cf., e.g.,~\cite{Basin2014}).
In the untimed case, this is not difficult to achieve: one can translate
\ltlpast{} formulas into B\"uchi automata~\cite{Gastin2003, Dax2010a} and
turn them into efficient trace-length independent monitors~\cite{Arafat2005}.
Unfortunately, a number of obstacles hinder the application of this methodology
to the real-time setting:
it is known that \mtlpast{} is expressively incomparable with timed automata;
even though certain fragments of \mtlpast{} can be translated into timed automata,
the latter are not always determinisable as required for the purpose of monitoring~\cite{Baier2009}.
For this reason, researchers purposed automata-free monitoring procedures that work
directly with metric temporal logic formulas (e.g.,~\cite{Maler2004, Basin2011}).
However, it proved difficult to maintain trace-length independence while allowing
\mtlpast{} in its full generality, i.e.,~with unbounded intervals and nesting of future and past modalities.
Almost all monitoring procedures for metric temporal logics in the
literature have certain syntactic or semantic limitations, e.g.,
only allowing bounded future modalities or assuming integer time.
A notable exception is~\cite{Baldor2012} which handles full \mtlpast{} over
signals, but which unfortunately fails to be trace-length independent.

\paragraph{\emph{Contributions}}
We study the expressiveness of various fragments and extensions of \mtlpast{}
over timed words.
In particular, we highlight a fundamental deficiency in the pointwise interpretation of \mtlpast{}.
To amend this, we propose new (first-order definable) modalities \emph{generalised `Until'} ($\guntil$) and \emph{generalised `Since'} ($\gsince$). With these new modalities and the techniques developed in~\cite{Pandya2011, Ouaknine2009, Hunter2012},
we establish the following results:
\begin{enumerate}[label=(\roman*).] % chktex 36
\item There is a strict hierarchy of metric temporal logics based on their expressiveness over bounded timed words (see Figure~\ref{fig:expsummary}
where the arrows indicate `strictly more expressive than'
and the edges indicate `equally expressive'). Note that this hierarchy collapses in the continuous semantics.
\item The metric temporal logic with the new modalities $\guntil$ and $\gsince$ (denoted \mtlpastg{}) is expressively complete for \foone{}
over bounded timed words.
\item The time-bounded satisfiability and model-checking problems for \mtlpastg{} are $\mathrm{EXPSPACE}$-complete, the same as that of \mtlpast{}.
\item Any \mtlpastg{} formula is equivalent to a \emph{syntactically separated} one.
\item \mtlpastg{} is expressively complete for \fo{} (the rational variant of \foone{}) over unbounded (i.e.,~infinite non-Zeno) timed words
if we allow the use of rational constants in modalities.
\end{enumerate}

\begin{figure}[ht]
\begin{minipage}[b]{0.45\linewidth}
\centering

\begin{tikzpicture}[>=to, scale=.8, transform shape]

\node		(v0) at (0, 0)	{\foone{}};

\node		(v1c) at (0.5, -2) {\raisebox{5.5pt}{\mtlpastcnt{}}};

\node		(v1) at (-2, -2) {\mtlpast{}};

\node		(v2) at (-4, -4) {\mtl{}};

\path		(v0) edge [->] (v1)
			(v0) edge [->] (v1c)
			(v1) edge [->] (v2);

\end{tikzpicture}
\caption*{known results}
\end{minipage}
\begin{minipage}[b]{0.45\linewidth}
\centering
\begin{tikzpicture}[>=to, scale=.8, transform shape]

\node		(v) at (2.5, 0) {\mtlpastg{}};

\node		(v0) at (0, 0)	{\foone{}};

\node		(v0') at (-1, -1) {\mtlbpast{}};

\node		(v1) at (-2, -2) {\mtlpast{}};

\node		(v1c) at (0.8, -2) {\raisebox{5.5pt}{\mtlpastcnt{}}};

\node		(v1') at (-3, -3) {\mtlunit{}};

\node		(v2) at (-4, -4) {\mtl{}};

\path		(v0) edge [->] (v0')
			(v0') edge [->] (v1)
			(v1) edge [->] (v1')
			(v1) edge [-] (v1c)
			(v1') edge [->] (v2)
			(v0) edge [-] (v);

\end{tikzpicture}
\caption*{our results}
\end{minipage}
\caption{Expressiveness results over bounded timed words.}%
\label{fig:expsummary}
\end{figure}

\noindent
For monitoring, we focus on a restricted version of the monitoring problem of \mtlpastg{}, based on the
notion of \emph{informative prefixes}~\cite{Kupferman2001a}.
The main idea of our approach is to work with \mtlpastg{} formulas of a special form:
\ltlpast{} formulas over atomic formulas comprised of bounded \mtlpastg{} formulas.\footnote{It follows
from the syntactic separation result that
no expressiveness is sacrificed in restricting to this fragment.}
The truth values of bounded \mtlpastg{} formulas can be computed and stored efficiently with a dynamic programming algorithm; these values are then used as input to
deterministic finite automata obtained from `backbone' \ltlpast{} formulas.
As a result, we obtain the first trace-length independent
monitoring procedure for a metric temporal logic that subsumes \mtlpast{}.
The procedure is free of dynamic memory allocations, linked lists, etc., and hence can be
implemented efficiently (the \emph{amortised} running time per event is linear
in the number of subformulas in all bounded formulas). To be more precise:
\begin{enumerate}[label=(\roman*)., start=6] % chktex 36
\item We give a trace-length independent monitoring procedure (which detects informative good/bad prefixes) for \ltlpast{} formulas
over atomic formulas comprised of bounded \mtlpastg{} formulas.
\item For an arbitrary \mtlpastg{} formula, we show that its informative good/bad prefixes are preserved
by the syntactic rewriting rules (and thus can be monitored in a
trace-length independent manner).
\end{enumerate}

\paragraph{\emph{Related work}}
Bouyer, Chevalier, and Markey~\cite{Bouyer2005} showed that \mtl{} (the future-only fragment of \mtlpast{})
is strictly less expressive than \mtlpast{} in both the continuous and pointwise semantics.
This, together with the aforementioned results~\cite{Hirshfeld2007, Pandya2011}, form a strict hierarchy
of expressiveness that holds in the both semantics:
\[
\mtl{} \subsetneq \mtlpast{} \subsetneq \foone{} \,.
\]
Ouaknine, Rabinovich, and Worrell~\cite{Ouaknine2009} showed that the hierarchy collapses in the
continuous semantics when one considers bounded time domains of the form $\ropen{0, N}$.
Our results show that this is not the case in the pointwise semantics.

Another route to expressive completeness is by allowing rational constants.
In particular, counting modalities become expressible in \mtlpast{}~\cite{Bouyer2005}.
Exploiting this observation, Hunter, Ouaknine, and Worrell~\cite{Hunter2012} showed that \mtlpast{} with rational constants
is expressively complete for \fo{} in the continuous semantics.
However, as can be immediately derived from a result of Prabhakar and D'Souza~\cite{Prabhakar2006},
this pleasant result does not hold in the pointwise semantics.
On the other hand, D'Souza and Tabareau~\cite{DSouza2004} showed that
\mtlpast{} with rational constants is expressively complete for \textmd{\textsf{rec-TFO[$\eventually, \once$]}} (an `input-determined' fragment of \fo{}) in the pointwise semantics.
We complement these results by extending \mtlpast{} with the new modalities
$\guntil$ and $\gsince$ to make it expressively complete for \fo{} in the pointwise semantics.

In a pioneering work, Thati and Ro{\c{s}}u~\cite{Thati2005} proposed a rewriting-based
monitoring procedure for \mtlpast{} over integer-timed traces.
Their procedure is trace-length independent and amenable to efficient implementations.
However, trace-length independent monitoring of \mtlpast{} is not possible in dense real-time settings: a monitor would have to `remember' an
infinite number of timestamps. For this reason, researchers often impose a
\emph{bounded-variability} assumption on input traces, i.e.,~only a bounded number of events may occur in any time unit.
Under such an assumption, Nickovic and Piterman~\cite{Nickovic2010a} showed that \mtl{} formulas can be
translated into deterministic timed automata.
Unfortunately, their approach does not easily extend to full \mtlpast{}.

It is known that the non-punctual fragment of \mtlpast{}, called \mitlpast{}, can be
translated into timed automata. Since the standard constructions~\cite{Alur1992, Alur1996}
are notoriously complicated, there have been some proposals for simplified or
improved constructions~\cite{Maler2006, Kini2011, DSouza2013, Brihaye2014, Brihaye2017}.
The difficulty in using these constructions
for monitoring, again, lies in the fact that timed automata cannot be
determinised in general.  In principle one can carry out on-the-fly
determinisation for input traces of bounded variability (cf., e.g.,~\cite{Tripakis2002, Baier2009});
however, it is not clear that this approach can yield an efficient procedure.

\section{Preliminaries}

\subsection{Automata and logics for real-time}

\paragraph{\emph{Timed words}}

Let the \emph{time domain} $\mathbb{T}$ be a subinterval of $\mathbb{R}_{\geq 0}$ that contains $0$.
A \emph{time sequence} $\tau = \tau_0\tau_1\dots$ is a non-empty finite or infinite sequence over
$\mathbb{T}$ (\emph{timestamps}) that satisfies the requirements below (we
denote the length of $\tau$ by $|\tau|$):
\begin{itemize}
\item \emph{Initialisation}: $\tau_0 = 0$

\item \emph{Strict monotonicity}: For all $i$, $0 \leq i < |\tau| - 1$, we
have $\tau_i < \tau_{i+1}$.\footnote{This requirement is chosen to simplify
the presentation; all the results
still hold (with some minor modifications) in the case of weakly-monotonic time,
i.e.,~requiring instead $\tau_i \leq \tau_{i+1}$ for all $i$, $0 \leq i < |\tau| - 1$.}
\end{itemize}
If $\tau$ is infinite we require it to be unbounded, i.e.,~we disallow
so-called Zeno sequences.
Given a finite alphabet $\Sigma$, a \emph{$\mathbb{T}$-timed word} over
$\Sigma$ is a pair $\rho = (\sigma, \tau)$ where $\sigma =
\sigma_0\sigma_1\dots$ is a non-empty finite or infinite word over
$\Sigma$ and $\tau$ is a time sequence over $\mathbb{T}$ of the same length.
We refer to a $\mathbb{T}$-timed word simply as a \emph{timed word} when $\mathbb{T} = \mathbb{R}_{\geq 0}$.\footnote{By the non-Zeno requirement, if $\mathbb{T}$ is bounded then a $\mathbb{T}$-timed word must be a finite timed word.}
We refer the pair $(\sigma_i, \tau_i)$ as the \emph{$i^{th}$ event} in $\rho$,
and define the \emph{distance} between $i^{th}$ and $j^{th}$ ($i \leq j$) events to
be $\tau_j - \tau_i$.
In this way, a timed word can be equivalently regarded as a sequence of events.
We denote by $|\rho|$ the number of events in $\rho$.
A \emph{position} in $\rho$ is a number $i$ such that $0 \leq i < |\rho|$.
The \emph{duration} of $\rho$ is defined as $\tau_{|\rho| - 1}$ if $\rho$ is finite.
We write $t \in \rho$ if $t$ is equal to one of the timestamps in $\rho$.
For a finite alphabet $\Sigma$, we write $T\Sigma^*$ and $T\Sigma^\omega$ for the respective sets of
finite and infinite timed words over $\Sigma$.
A \emph{timed (finite-word) language} over $\Sigma$ is a subset of $T\Sigma^\omega$ ($T\Sigma^*$).

\paragraph{\emph{Timed automata}}
The most popular model for real-time systems are \emph{timed automata}~\cite{Alur1994}, introduced by Alur and Dill in the early 1990s. Timed automata extends finite automata by real-valued variables (called \emph{clocks}).

\begin{defi}
Given a set of clocks $X$, the set $\mathcal{G}(X)$ of clock constraints $g$ is
defined inductively by
\[
g ::= \mathbf{true} \,\mid\, x \bowtie c \,\mid\, g_1 \wedge g_2
\]
where $x \in X$, $c \in \mathbb{N}$ and $\bowtie \; \in \{<, \leq, >, \geq\}$.
\end{defi}
\begin{defi}
A (non-deterministic) timed automaton $\mathcal{A}$ is a tuple $\langle \Sigma, S, S_0, X, I, E, F \rangle$ where
\begin{itemize}
\item $\Sigma$ is a finite alphabet
\item $S$ is a finite set of locations
\item $S_0 \subseteq S$ is the set of initial locations
\item $X$ is a finite set of clocks
\item $I: S \mapsto \mathcal{G}(X)$ is a mapping that labels each location in $S$ with a clock constraint in $\mathcal{G}(X)$ (an `invariant')
\item $E \subseteq S \times \Sigma \times 2^X \times \mathcal{G}(X) \times S$ is the set of edges.
	An edge $\langle s, a, \lambda, g, s' \rangle$
	denotes an $a$-labelled edge from location $s$ to location $s'$ where $g$ (a `guard') specifies
	when the edge is enabled and $\lambda \subseteq X$ is the set of clocks to be reset with this edge
\item $F$ is the set of accepting locations.
\end{itemize}
\end{defi}

\noindent
We say that $\mathcal{A}$ is \emph{deterministic} if it (i) has only one initial location and
(ii) for each $s \in S$, $a \in \Sigma$ and every pair of edges $\langle s, a, \lambda_1, g_1, s_1 \rangle$,
$\langle s, a, \lambda_2, g_2, s_2 \rangle$, $g_1$ and $g_2$ are mutually exclusive (i.e.,~$g_1 \wedge g_2$ is unsatisfiable).
We say that $\mathcal{A}$ is \emph{complete} if for each $s \in S$ and $a \in \Sigma$,
the disjunction of the clock constraints of the $a$-labelled edges starting
at $s$ is a valid formula.

Assume that $\mathcal{A}$ has $n$ clocks. We define its set of clock values as
$\textsf{Val} = [0, c_{\maxit}] \cup \{\top\}$ where $c_{\maxit}$ is the maximum constant
appearing in $\mathcal{A}$. A \emph{state} of $\mathcal{A}$ as a pair $(s, \mathbf{v})$ where
$s \in S$ is a location and $\mathbf{v} \in \textsf{Val}^n$ is a \emph{clock valuation}. Write $\mathbf{v}(x)$ for
the value of clock $x$ in $\mathbf{v}$. We denote by $Q = S \times \textsf{Val}^n$ the set of all states of $\mathcal{A}$.
A \emph{run} of $\mathcal{A}$ on a timed word can be seen as follows: the automaton takes some edge when an event arrives, otherwise it stays in the same location as time elapses.
More precisely, $\mathcal{A}$ induces a labelled transition system $\mathcal{T_A} = \langle Q, \leadsto, \rightarrow \rangle$
where $\leadsto \; \subseteq Q \times \mathbb{R}_{> 0} \times Q$ is the \emph{delay-step relation} and $\rightarrow \; \subseteq
Q \times \Sigma \times Q$ is the \emph{discrete-step relation}. In these steps, corresponding invariants and guards must be met (define $\top > c$ for all constants $c$):
\begin{itemize}
\item For $(s, \mathbf{v}) \overset{t}{\leadsto} (s', \mathbf{v'})$, $s' = s$, $\mathbf{v'} = \mathbf{v} + t$
		and $\mathbf{v} + t' \models I(s)$ for all $0 \leq t' \leq t$.
\item For $(s, \mathbf{v}) \overset{a}{\rightarrow} (s', \mathbf{v'})$, there is an edge $\langle s, a, \lambda, g, s' \rangle \in E$
		such that $\mathbf{v'} = \mathbf{v}[\lambda := 0]$ and $\mathbf{v} \models g$.
\end{itemize}
The clock valuation $\mathbf{v} + t$ maps each clock $x$ to $\mathbf{v}(x) + t$ if $\mathbf{v}(x) + t \leq c_{\maxit}$, otherwise $\top$.
$\mathbf{v}[\lambda := 0]$ maps $x$ to $\mathbf{v}(x)$ if $x \notin \lambda$, otherwise $0$.
Formally, a run of $\mathcal{A}$ on $\rho = (\sigma, \tau)$ is an alternating
sequence of delay steps and discrete steps
\[
(s_0, \mathbf{v}_0) \overset{\sigma_0}{\rightarrow} (s_1, \mathbf{v}_1) \overset{d_0}{\leadsto} (s_2, \mathbf{v}_2)
\overset{\sigma_1}{\rightarrow} (s_3, \mathbf{v}_3) \overset{d_1}{\leadsto} (s_4, \mathbf{v}_4) \overset{\sigma_2}{\rightarrow} \dots
\]
where $d_i = \tau_{i+1} - \tau_{i}$ for $i \geq 0$, $s_0 \in S_0$ and $\mathbf{v}_0 = 0^n$.
A finite timed word $\rho'$ is \emph{accepted} by $\mathcal{A}$ if there is an \emph{accepting} run (i.e.,~ending in an
accepting location) of $\mathcal{A}$ on $u'$.
We can also equip $\mathcal{A}$ with a B\"uchi acceptance condition; in this case,
a run is \emph{accepting} if it visits an accepting location infinitely often,
and an infinite timed word $\rho$ is \emph{accepted} by $\mathcal{A}$ if there is
such a run of $\mathcal{A}$ on $\rho$.
The \emph{timed (finite-word) language} defined by $\mathcal{A}$ is the set of (finite) timed words
accepted by $\mathcal{A}$.
Note that timed automata are not closed under complementation; for example,
the complement of the timed language accepted by the timed automaton in the
example below cannot be recognised by any timed automaton.

\begin{exaC}[\cite{Alur2004}]
Consider the timed automaton with $\Sigma = \{a, b\}$ in Figure~\ref{fig:taexample}.
The automaton accepts timed words containing an $a$ event at some time $t$ such that no
event occurs at time $t + 1$.

\begin{figure}[ht]
\centering
\begin{tikzpicture}[->, >=stealth', shorten >=1pt, auto, node distance=6cm, transform shape, scale=0.8,
							semithick, bend angle=25, every state/.style={fill=none,draw=black,text=black,shape=circle}]

\node[state]					(l0) [initial left, initial text={}]	{\large $l_0$};
\node[state, accepting]					(l1) [right=4cm of l0]	{\large $l_1$};

\path
			(l0)	edge node [align=center] {\large $a$ \\ $x := 0$} (l1)
			(l1)	edge [loop above] node [align=center] {\large $a$, $b$ \\ $x \neq 1$} (l1)
			(l0)	edge [loop above] node [align=center] {\large $a$, $b$} (l0);
\end{tikzpicture}
\caption{A timed automaton.}%
\label{fig:taexample}
\end{figure}
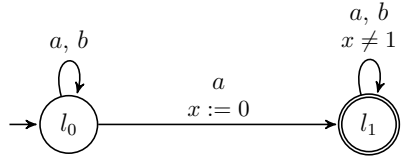

\end{exaC}

\paragraph{\emph{Monadic First-Order Logic of Order and Metric}}
We now define the \emph{Monadic First-Order Logic of Order and Metric} (\foone{})~\cite{Wilke1994}
which subsumes all the other logics discussed in this article.

\begin{defi}
Given a set of monadic predicates $\mathbf{P}$,
the set of \emph{\foone{}} formulas is generated by the grammar
\[
    \begin{array}{rcl}
    \vartheta & ::= & \mathbf{true} \,\mid\, P(x) \,\mid\, x < x' \,\mid\, d(x, x') \sim c \,\mid\, \vartheta_1 \wedge \vartheta_2 \,\mid\, \neg \vartheta \,\mid\,
    \exists x \, \vartheta \,,
    \end{array}
\]
where $P \in \mathbf{P}$, $x, x'$ are variables, $\sim \; \in \{ <, > \}$ and
$c \in \mathbb{N}$.\footnote{Note that whilst we refer to the logic as \foone{}, we adopt an equivalent definition where binary distance predicates $d(x, x') \sim c$ (as in~\cite{Wilke1994}) are used in place of the usual $+1$ function symbol.}
\end{defi}
The fragment where $d(x, x') \sim c$ is absent is called the
\emph{Monadic First-Order Logic of Order} (\textup{\textmd{\textsf{FO[$<$]}}}).

\paragraph{\emph{Metric temporal logics}}
Formulas of metric temporal logics are
\foone{} formulas (with a single free variable)
built from monadic predicates using Boolean connectives and \emph{modalities} (or \emph{operators}).
A $k$-ary modality is defined by an \mbox{\foone{}} formula $\varphi(x, X_1, \dots, X_k)$
with a single free variable $x$ and $k$ free monadic predicates $X_1, \dots, X_k$.
For example, the \mtlpast{}~\cite{Koymans1990} modality $\until_{(0, 5)}$
is defined by the \foone{} formula
\[
    \arraycolsep=0.5ex
    \begin{array}{rcll}
    \until_{(0, 5)}(x, X_1, X_2) & = & \exists x' \, \Bigl( & x < x' \wedge d(x, x') < 5 \wedge X_2(x') \\
    & & & {} \wedge \forall x'' \, \bigl( x < x'' \wedge x'' < x' \implies X_1(x'') \bigr) \Bigr) \,.
    \end{array}
\]
The \mtlpast{} formula $\varphi_1 \until_{(0, 5)} \varphi_2$ (usually written
in infix notation) is obtained by substituting \mtlpast{} formulas $\varphi_1, \varphi_2$ for $X_1, X_2$, respectively.

\begin{defi}
Given a set of monadic predicates $\mathbf{P}$, the set of \emph{\mtlpast{}} formulas is generated by the grammar
\[
    \begin{array}{rcl}
    \varphi & ::= & \mathbf{true} \,\mid\, P \,\mid\, \varphi_1
    \wedge \varphi_2 \,\mid\, \neg \varphi \,\mid\, \varphi_1 \until_I
    \varphi_2 \,\mid\, \varphi_1 \since_I \varphi_2 \,,
    \end{array}
\]
where $P \in \mathbf{P}$ and $I \subseteq (0, \infty)$ is an interval
with endpoints in $\mathbb{N} \cup \{\infty\}$.
\end{defi}
The (future-only) fragment \mtl{} is obtained by disallowing subformulas
of the form $\varphi_1 \since_I \varphi_2$.
We write $|I|$ for $\sup(I) - \inf(I)$.
If $I$ is not present as a subscript then it is assumed to be $(0, \infty)$.
We sometimes use pseudo-arithmetic expressions to denote intervals, e.g.,~`$\geq 1$' denotes $\ropen{1, \infty}$
and `$= 1$' denotes the singleton $\{1\}$.
We also employ the usual syntactic sugar,
e.g.,~$\mathbf{false} \equiv \neg \mathbf{true}$, $\eventually_I \varphi \equiv \mathbf{true} \, \until_I \varphi$,
$\once_I \varphi \equiv \mathbf{true} \since_I \varphi$,
$\globally_I \varphi \equiv \neg \eventually_I \neg
\varphi$ and $\nextx_I \varphi \equiv \mathbf{false} \until_I \varphi$, etc.
For convenience, we also use `weak' temporal operators as syntactic sugar, e.g.,~$\varphi_1 \until^w_I \varphi_2 \equiv \varphi_1
\wedge (\varphi_1 \until_I \varphi_2)$ if $0 \notin I$ and $\varphi_1
\until^w_I \varphi_2 \equiv \varphi_2 \vee \left( \varphi_1 \wedge
(\varphi_1 \until_I \varphi_2) \right)$ if $0 \in I$ (we allow $0 \in I$ in the case of weak temporal operators).
We denote by $|\varphi|$ the number of subformulas in $\varphi$.

\paragraph{\emph{The pointwise semantics}}

With each $\mathbb{T}$-timed word $\rho = (\sigma, \tau)$ over $\Sigma_{\mathbf{P}} = 2^\mathbf{P}$ we associate a structure $M_\rho$.
Its universe $U_\rho$ is the subset $\{ \tau_i \mid 0 \leq i < |\rho|\}$ of $\mathbb{T}$.
The order relation $<$ and monadic predicates in $\mathbf{P}$ are interpreted in the expected way, e.g.,~$P(\tau_i)$ holds in $M_\rho$ iff $P \in \sigma_i$.
The binary \emph{distance predicate} $d(x, x') \sim c$ holds iff $|x - x'| \sim c$.
The satisfaction relation is defined inductively as usual.
We write $M_\rho, t_0, \dots, t_{n-1} \models \vartheta(x_0, \dots, x_{n-1})$
(or $\rho, t_0, \dots, t_{n-1} \models \vartheta(x_0, \dots, x_{n-1})$)
if $t_0, \dots, t_{n-1} \in U_\rho$
and $\vartheta(t_0, \dots, t_{n-1})$ holds in $M_\rho$.
We say that two \foone{} formulas $\vartheta_1(x)$ and $\vartheta_2(x)$ are \emph{equivalent} over
$\mathbb{T}$-timed words if for all $\mathbb{T}$-timed words $\rho$ and $t \in U_\rho$,
\[
\rho, t \models \vartheta_1(x) \iff \rho, t \models \vartheta_2(x) \,.
\]
We say that a metric logic $L'$ is \emph{expressively complete} for metric logic $L$
over $\mathbb{T}$-timed words iff for any formula $\vartheta(x) \in L$,
there is an equivalent formula $\varphi(x) \in L'$ over $\mathbb{T}$-timed words.
We say that $L'$ is \emph{at least as expressive as}
(or \emph{more expressive than}) $L$ over $\mathbb{T}$-timed words (written $L \subseteq L'$)
iff for any formula $\vartheta \in L$, there is an \emph{initially equivalent} formula $\varphi \in L'$ over $\mathbb{T}$-timed words
(i.e.,~$\vartheta$ and $\varphi$ evaluates to the same truth value at the beginning of any $\mathbb{T}$-timed word).
If $L \subseteq L'$ but $L' \nsubseteq L$ then we say that $L'$ is \emph{strictly more expressive than} $L$
(or $L$ is \emph{strictly less expressive than} $L'$) over $\mathbb{T}$-timed words.

As we have seen earlier, each \mtlpast{} formula can be defined as an \foone{}
formula with a single free variable. Here, for the sake of completeness we
give an (equivalent) traditional inductive definition of the satisfaction relation for \mtlpast{}
over timed words. We write $\rho \models \varphi$ if $\rho, 0 \models \varphi$.
\begin{defi}
The satisfaction relation $\rho, i \models \varphi$ for an \emph{\mtlpast{}}
formula $\varphi$, a timed word $\rho = (\sigma, \tau)$ and a position $i$ in $\rho$ is defined as follows:
\begin{itemize}
\item $\rho, i \models \mathbf{true}$
\item $\rho, i \models P$ iff $P(\tau_i)$ holds in $M_\rho$
\item $\rho, i \models \varphi_1 \wedge \varphi_2$ iff $\rho, i
\models \varphi_1$ and $\rho, i \models \varphi_2$
\item $\rho, i \models \neg \varphi$ iff $\rho, i \not \models
\varphi$
\item $\rho, i \models \varphi_1 \until_I \varphi_2$ iff there exists $j$, $i < j < |\rho|$
such that $\rho, j \models \varphi_2$, $\tau_j - \tau_i \in I$
and $\rho, k \models \varphi_1$ for all $k$ with $i < k < j$
\item $\rho, i \models \varphi_1 \since_I \varphi_2$ iff there exists
$j$, $0 \leq j < i$ such that $\rho, j \models \varphi_2$, $\tau_i - \tau_j \in I$
and $\rho, k \models \varphi_1$ for all $k$ with $j < k < i$.
\end{itemize}
\end{defi}

\begin{exa}
The \mtl{} formula
\begin{equation} \label{for:example}
\varphi = \globally (P \implies \eventually_{< 3} Q)
\end{equation}
is satisfied by a timed word $\rho$ if and only if there is a $P$-event in $\rho$ (say at time $t$), and
there is a $Q$-event in $\rho$ with timestamp in $(t, t+3)$.

\end{exa}
\paragraph{\emph{Safety relative to the divergence of time}}

Recall that we require the timestamps of any infinite timed word to
be a strictly-increasing divergent sequence. Based upon this assumption, we
define \emph{safety properties} in exactly the same way as in the qualitative case~\cite{Alpern1987}; for example, (\ref{for:example}) is a safety property as any infinite timed word $u'$ violating $\varphi$ must have a prefix $u$
such that there is a $P$-event in $u$ with no $Q$-event in the following three time units.
On the other hand, had we allowed Zeno timed words, $\varphi$ would not be safety as
\[
(\{P\}, 0)(\{P\}, 1)(\{P\}, 1 + \frac{1}{2})(\{P\}, 1 + \frac{1}{2} + \frac{1}{4}) \dots
\]
violates $\varphi$ without having a prefix that cannot be extended into an infinite timed word
satisfying $\varphi$.
The notion we adopt here is called
\emph{safety relative to the divergence of time} in the literature~\cite{Henzinger1992}.

\paragraph{\emph{The continuous semantics}}

Another way to interpret metric logics is to regard time as a continuous entity;
a behaviour of a system can thus be viewed as a continuous function.
Formally, a $\mathbb{T}$-\emph{signal} over finite alphabet $\Sigma$ is a function $f: \mathbb{T} \mapsto \Sigma$ that
is \emph{finitely variable}, i.e.,~the restriction of $f$ to a subinterval of $\mathbb{T}$ of finite length
has only a finite number of discontinuities.
We refer to a $\mathbb{T}$-signal simply as a \emph{signal} when $\mathbb{T} = \mathbb{R}_{\geq 0}$.
With each signal $f$ over $\Sigma_{\mathbf{P}}$ we associate a structure $M_f$.
Its universe $U_f$ is $\mathbb{T}$.
The order relation $<$ and monadic predicates in $\mathbf{P}$ are interpreted in the expected way,
e.g., $P(x)$ holds in $M_f$ iff $P \in f(x)$.
The binary \emph{distance predicate} $d(x, x') \sim c$ holds iff $|x - x'| \sim c$.
We write $M_f, t_0, \ldots, t_{n-1} \models \vartheta(x_0, \ldots, x_{n-1})$
(or $f, t_0, \ldots, t_{n-1} \models \vartheta(x_0, \ldots, x_{n-1})$)
if $t_0, \ldots, t_{n-1} \in U_f$
and $\vartheta(t_0, \ldots, t_{n-1})$ holds in $M_f$.
The notions of equivalence of formulas, expressiveness of
metric logics, etc.\ are defined as in the case of timed words.

The satisfaction relation for \mtlpast{} over signals is defined as follows.
We write $f \models \varphi$ if $f, 0 \models \varphi$.
\begin{defi}
The satisfaction relation $f, t \models \varphi$ for an \emph{\mtlpast{}}
formula $\varphi$, a signal $f$ and $t \in U_f$ is defined as follows:
\begin{itemize}
\item $f, t \models P$ iff $P(t)$ holds in $M_f$
\item $f, t \models \mathbf{true}$
\item $f, t \models \varphi_1 \wedge \varphi_2$ iff $f, t \models \varphi_1$ and $f, t \models \varphi_2$
\item $f, t \models \neg \varphi$ iff $f, t \not \models \varphi$
\item $f, t \models \varphi_1 \until_I \varphi_2$ iff there exists $t' > t$, $t' \in \mathbb{T}$
such that $f, t' \models \varphi_2$, $t' - t \in I$
and $f, t'' \models \varphi_1$ for all $t''$ with $t < t'' < t'$
\item $f, t \models \varphi_1 \since_I \varphi_2$ iff there exists
$t' < t$, $t' \in \mathbb{T}$ such that $f, t' \models \varphi_2$, $t - t' \in I$
and $f, t'' \models \varphi_1$ for all $t''$ with $t' < t'' < t$.
\end{itemize}
\end{defi}

\paragraph{\emph{Relating the two semantics}}

Note that timed words can be regarded as a particular kind of signal:
for a given $\mathbb{T}$-timed word $\rho$ over $\Sigma_P$,
we can introduce a `silent' monadic predicate $P_\epsilon$ and construct the corresponding $\mathbb{T}$-signal $f^{\rho}$ over $\Sigma_{\mathbf{P}'}$,
where $\mathbf{P}' = \mathbf{P} \cup \{ P_\epsilon \}$, as follows:
\begin{itemize}
\item $f^{\rho}(\tau_i) = \sigma_i$ for all $i$, $0 \leq i < |\rho|$
\item $f^{\rho}(\tau_i) = \{P_\epsilon\}$.
\end{itemize}
This enables us to interpret metric logics over timed words `continuously'.
We can thus compare the expressiveness of metric logics in both semantics by restricting the models of the continuous interpretations of metric logics to signals of this form (i.e.,~$f^\rho$ for some timed word $\rho$).
For example, we say that continuous \foone{} is at least as
expressive as pointwise \foone{} since for each \foone{} formula $\vartheta_{\mathit{pw}}(x)$,
there is an `equivalent' \foone{} formula $\vartheta_{\mathit{cont}}(x)$ such that
$\rho, t \models \vartheta_{\mathit{pw}}(x)$ iff $f^{\rho}, t \models \vartheta_{\mathit{cont}}(x)$.

\begin{exa}
Consider the timed word $\rho$ illustrated in Figure~\ref{fig:semanticsexample}
where the red boxes denote $P$-events.
The \mtl{} formula
\[
\varphi = \eventually (\eventually_{=1} P)
\]
does not hold at the beginning of $\rho$ in the pointwise semantics (i.e.,~$\rho \not \models \varphi$) since
there is no event at exactly one time unit before the second event in $\rho$.
On the other hand, $\varphi$ holds
at the beginning of $\rho$ in the continuous semantics
(i.e.,~$f^\rho \models \varphi$) since
there is a point (at which $P_\epsilon$ holds) at exactly one time unit before
the second event in $f^\rho$.
We can, however, simulate the pointwise semantics with
\[
\varphi' = \eventually \Bigl(\neg P_\epsilon \wedge \bigl(\eventually_{=1} (\neg P_\epsilon \wedge P) \bigr)\Bigr)\,,
\]
for which we have $\pi \models \varphi$ iff $f^{\pi} \models \varphi'$
for all timed words $\pi$.

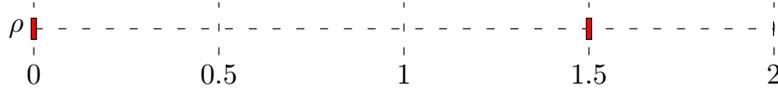
\begin{figure}[h]
\centering
\begin{tikzpicture}[scale=1.0]
\begin{scope}[>=latex]

\draw[|-|, loosely dashed] (0pt,0pt) -- (280pt,0pt) node[at start, left] {$\rho$};

\draw[loosely dashed] (0pt,-10pt) -- (0pt,10pt) node[at start,below] {$0$};
\draw[loosely dashed] (70pt,-10pt) -- (70pt,10pt) node[at start,below] {$0.5$};
\draw[loosely dashed] (140pt,-10pt) -- (140pt,10pt) node[at start,below] {$1$};
\draw[loosely dashed] (210pt,-10pt) -- (210pt,10pt) node[at start,below] {$1.5$};
\draw[loosely dashed] (280pt,-10pt) -- (280pt,10pt) node[at start,below] {$2$};

\draw[draw=black, fill=red] (1pt, -4pt) rectangle (-1pt, 4pt);

\draw[draw=black, fill=red] (209pt, -4pt) rectangle (211pt, 4pt);

\end{scope}

\end{tikzpicture}
\caption{The timed word $\rho$.}%
\label{fig:semanticsexample}
\end{figure}

\end{exa}
As we see in the example above, the pointwise and continuous interpretations
of metric logics differs in the range of first-order quantifiers.
While the ability to quantify over time points \emph{between} events appears to increase the expressiveness
of metric logics, this is not the case for \foone{} as both interpretations
are indeed equally expressive (when one considers only signals of the form $f^\rho$)~\cite{DSouza2007}.\footnote{The translation in~\cite{DSouza2007} also holds in a time-bounded setting with trivial modifications.}
By contrast, \mtlpast{} is strictly more expressive in the continuous semantics
than in the pointwise semantics~\cite{Prabhakar2006}.

\subsection{Model checking}

A key advantage in using \ltl{} (or \ltlpast{}) in verification is that its \emph{model-checking}
problem is $\mathrm{PSPACE}$-complete~\cite{Sistla1985}, much better than the complexity of the same problem for \textmd{\textsf{FO[$<$]}} (non-elementary~\cite{Stockmeyer1974}).
Given a B\"uchi automaton $\mathcal{A}$ that models the system and a specification written as an \ltl{} formula $\Phi$, the corresponding model-checking problem
asks whether the language defined by $\mathcal{A}$ is included in
the language defined by $\Phi$. By a fundamental result in verification---\ltl{} formulas
can be translated into B\"uchi automata~\cite{Wolper1983}---this reduces to the
\emph{emptiness} problem on the product B\"uchi automaton
of $\mathcal{A}$ and the B\"uchi automaton $\mathcal{B}_{\neg \Phi}$ translated from $\neg \Phi$.
The latter problem can be solved, e.g., by a standard fixed-point algorithm~\cite{Emerson1986}.
This is sometimes called the \emph{automata-theoretic approach} to \ltl{} model checking.

In the real-time setting, given a timed (B\"uchi) automaton $\mathcal{A}$ and a specification $\varphi$
(e.g., a formula of some metric logic), the  corresponding model-checking problem asks
whether the timed (finite-word) language defined by $\mathcal{A}$ is included in the timed (finite-word) language
defined by $\varphi$. By analogy with the untimed case, one may solve this problem
by first translating $\neg \varphi$ into a timed automaton $\mathcal{A}_{\neg \varphi}$
and then checking the emptiness of the product of $\mathcal{A}$ and $\mathcal{A}_{\neg \varphi}$.
This methodology works for certain metric logics;
for example, each formula of \mitl{} (the non-punctual fragment of \mtl{})
can be translated into a timed automaton, and the model-checking problem for timed automata
against \mitl{} is $\mathrm{EXPSPACE}$-complete~\cite{Alur1996}.
However, this does not apply to \mtl{} as
\mtl{} formulas, in general, cannot be translated into timed automata.

\subsection{Monitoring}

The \emph{prefix} problem~\cite{Bauer2013} asks the following:
given a specification $\Phi$ and a finite word $u$, do all infinite extensions
of $u$ satisfy $\Phi$? If the answer is `yes', then we say that $u'$ is a \emph{good prefix} for $\Phi$.
Similarly, $u$ is a \emph{bad prefix} for $\Phi$ if the answer to the dual problem is `yes', i.e.,~none of its infinite extensions satisfies $\Phi$.
The \emph{monitoring} problem takes instead a specification $\Phi$ and
an infinite word $u'$ as inputs. In contrast to standard decision problems,
the latter input is given \emph{incrementally}, i.e.,~one symbol at a time;
a monitor (a procedure that `solves' the monitoring problem) is required to
continuously check whether the currently accumulated finite word $u$
(a prefix of $u'$) is a good/bad prefix for $\Phi$ and report as necessary.

\section{Expressive completeness of \mtlpastg{} over bounded timed words}

In this section, we study the expressiveness of \mtlpast{} (and its various fragments and extensions) in
a time-bounded pointwise setting, i.e.,~all timed words are assumed to have durations less than a positive integer $N$.
We first recall \mtlpast{} EF games~\cite{Pandya2011},
which serves as our main tool in proving expressiveness results.
Then we demonstrate a strict hierarchy of metric temporal logics (based on their expressiveness
over bounded timed words) as we extend \mtl{} incrementally towards \foone{}.
Finally, we show that \mtlpast{}, equipped with both
the forwards and backwards temporal modalities generalised `Until' ($\guntil_I^c$) and
generalised `Since' ($\gsince_I^c$),
has precisely the same expressive
power as \foone{} over bounded time domains in the pointwise
semantics.
For the time-bounded satisfiability and model-checking problems,
we show that the relevant constructions (and hence the complexity bounds) for \mtlpast{} in~\cite{Ouaknine2009}
carry over to our new logic \mtlpastg{}.

\subsection{\mtlpast{} EF games}

Ehrenfeucht-Fra\"{\i}ss\'{e} games are handy tools in proving
the inexpressibility of certain properties in first-order logics.
In many proofs in this section,
we resort to (extended versions of) Pandya and Shah's \mtlpast{} \emph{EF games} on timed words~\cite{Pandya2011}, which
itself is a timed generalisation of Etessami and Wilke's \ltlpast{} EF games~\cite{Etessami1996}.

An $m$-round \mtlpast{} EF game starts with round $0$ and ends with round $m$.
The game is played by two players (\emph{Spoiler} and \emph{Duplicator}) on a pair of timed words $\rho$ and~$\rho'$.\footnote{We follow the convention
that \emph{Spoiler} is male and \emph{Duplicator} is female.}
A \emph{configuration} is a pair of positions $(i, j)$, respectively in $\rho$ and~$\rho'$.
In each round $r$ ($0 \leq r \leq m$), the game proceeds as follows.
\emph{Spoiler} first checks whether the two events that correspond to the current configuration $(i_r, j_r)$
in $\rho$ and $\rho'$ satisfy the same set of monadic predicates.
If this is not the case then he wins the game. Otherwise if $r < m$, \emph{Spoiler}
chooses an interval $I \subseteq (0, \infty)$ with endpoints in $\mathbb{N} \cup \{ \infty \}$
and plays either of the following moves:
\begin{itemize}
\item \emph{$\until_{I}$-move}: \emph{Spoiler} chooses one of the two timed words (say $\rho$).
He then picks $i_r'$ such that $i_r < i_r'$ and $\tau_{i_r'} - \tau_{i_r} \in I$ where $\tau_{i_r'}$ and $\tau_{i_r}$ are the corresponding timestamps in $\rho$ (if there is no such $i_r'$ then \emph{Duplicator} wins the game).
\emph{Duplicator} must choose a position $j_r'$ in $\rho'$ such that the difference of
the corresponding timestamps in $\rho'$ is in $I$. If she cannot find such a position then \emph{Spoiler}
wins the game. Otherwise, \emph{Spoiler} plays either of the following `parts':
	\begin{itemize}
	\item \emph{$\eventually$-part}: The game proceeds to the next round with $(i_{r+1}, j_{r+1}) = (i_r', j_r')$.
	\item \emph{$\until$-part}: If $j_r' = j_r + 1$ the game proceeds to the next round
	with $(i_{r+1}, j_{r+1}) = (i_r', j_r')$. If $i_r' = i_r + 1$ but $j_r' \neq j_r + 1$ then \emph{Spoiler} wins the game.
	Otherwise \emph{Spoiler} picks another position $j_r''$ in $\rho'$ such that $j_r < j_r'' < j_r'$.
	\emph{Duplicator} have to choose a position $i_r''$ in $\rho$ such that $i_r < i_r'' < i_r'$ in response.
	If she cannot find such a position then \emph{Spoiler} wins the game;
   otherwise the game proceeds to the next round with $(i_{r+1}, j_{r+1}) = (i_r'', j_r'')$.
	\end{itemize}
\item \emph{$\since_{I}$-move}: Defined symmetrically.
\end{itemize}
We say that \emph{Duplicator}
has a \emph{winning strategy} for the $m$-round \mtlpast{} EF game on $\rho$ and $\rho'$ that starts from configuration $(i, j)$
if and only if, no matter how \emph{Spoiler} plays, he cannot win the $m$-round \mtlpast{} EF game on $\rho$ and $\rho'$ with $(i_0, i_0) = (i, j)$.
If this is not the case then we say that \emph{Spoiler} has a winning strategy.

It is obvious that the moves in \mtlpast{} EF games
are closely related to the semantics of modalities in \mtlpast{} formulas.
For example, the $\until_I$-move can be seen as \emph{Spoiler}'s attempt
to verify that a formula of the form $\varphi_1 \until_I \varphi_2$
holds at $i_r$ in $\rho$ if and only if it holds at $j_r$ in $\rho'$: the $\eventually$-part
and the remaining rounds verify that $\varphi_2$ holds at $i_r'$ in $\rho$
iff it holds at $j_r'$ in $\rho'$, whereas the $\until$-part and the remaining rounds
verify that $\varphi_1$ holds at all $i_r''$, $i_r < i_r'' < i_r'$ in $\rho$
iff it holds at all $j_r''$, $j_r < j_r'' < j_r'$ in $\rho'$.
Formally, the following theorem relates the number of rounds of \mtlpast{} EF games to the
\emph{modal depth} (i.e.,~the maximal depth of nesting of modalities) of
\mtlpast{} formulas.

\begin{thmC}[\cite{Pandya2011}]\label{thm:mtlef}
For (finite) timed words $\rho, \rho'$ and an \emph{\mtlpast{}} formula $\varphi$ of modal depth $\leq m$,
if \emph{Duplicator} has a winning strategy for the $m$-round \emph{\mtlpast{}} EF game on
$\rho, \rho'$ with $(i_0, j_0) = (0, 0)$, then
\[
\rho \models \varphi \iff \rho' \models \varphi \,.
\]
\end{thmC}
In other words, $\rho, \rho'$ can be distinguished by an \mtlpast{} formula
of modal depth $\leq m$ if and only if \emph{Spoiler} has
a winning strategy for the $m$-round \mtlpast{} EF game on $\rho, \rho'$
with $(i_0, j_0) = (0, 0)$.
Note that specialised versions of Theorem~\ref{thm:mtlef} also hold for sublogics of \mtlpast{};
for example, the corresponding theorem for \mtl{} is obtained by banning the $\since_I$-move.
\begin{exa}
Consider the timed words $\rho$ and $\rho'$ illustrated in Figure~\ref{fig:mtlefexample} where
the white, red and blue boxes represent events at which no monadic predicate holds,
$P$-events, and $Q$-events, respectively. The positions are labelled above
the events.
\begin{figure}[h]
\centering
\begin{tikzpicture}[scale=1.0]
\begin{scope}[>=latex]

\draw[|-|, loosely dashed] (0pt,40pt) -- (140pt,40pt) node[at start, left] {$\rho$};
\draw[|-|, loosely dashed] (0pt,0pt) -- (140pt,0pt) node[at start, left] {$\rho'$};

\draw[loosely dashed] (0pt,-10pt) -- (0pt,0pt) node[at start,below] {$0$};
\draw[loosely dashed] (0pt, 15pt) -- (0pt,40pt);
\draw[loosely dashed] (140pt,-10pt) -- (140pt,50pt) node[at start,below] {$1$};

\draw[draw=black, fill=white] (1pt, -4pt) rectangle (-1pt, 4pt);
\draw[draw=black, fill=white] (1pt, 36pt) rectangle (-1pt, 44pt);
\node[above] at (0pt, 5pt) {{\tiny $0$}};
\node[above] at (0pt, 45pt) {{\tiny $0$}};

\draw[draw=black, fill=red] (19pt, -4pt) rectangle (21pt, 4pt);
\draw[draw=black, fill=red] (19pt, 36pt) rectangle (21pt, 44pt);
\node[above] at (20pt, 5pt) {{\tiny $1$}};
\node[above] at (20pt, 45pt) {{\tiny $1$}};

\draw[draw=black, fill=red] (39pt, -4pt) rectangle (41pt, 4pt);
\draw[draw=black, fill=red] (39pt, 36pt) rectangle (41pt, 44pt);
\node[above] at (40pt, 5pt) {{\tiny $2$}};
\node[above] at (40pt, 45pt) {{\tiny $2$}};

\draw[draw=black, fill=white] (59pt, -4pt) rectangle (61pt, 4pt);
\draw[draw=black, fill=red] (59pt, 36pt) rectangle (61pt, 44pt);
\node[above] at (60pt, 5pt) {{\tiny $3$}};
\node[above] at (60pt, 45pt) {{\tiny $3$}};

\draw[draw=black, fill=red] (79pt, -4pt) rectangle (81pt, 4pt);
\draw[draw=black, fill=red] (79pt, 36pt) rectangle (81pt, 44pt);
\node[above] at (80pt, 5pt) {{\tiny $4$}};
\node[above] at (80pt, 45pt) {{\tiny $4$}};

\draw[draw=black, fill=red] (99pt, -4pt) rectangle (101pt, 4pt);
\draw[draw=black, fill=red] (99pt, 36pt) rectangle (101pt, 44pt);
\node[above] at (100pt, 5pt) {{\tiny $5$}};
\node[above] at (100pt, 45pt) {{\tiny $5$}};

\draw[draw=black, fill=blue] (119pt, -4pt) rectangle (121pt, 4pt);
\draw[draw=black, fill=blue] (119pt, 36pt) rectangle (121pt, 44pt);
\node[above] at (120pt, 5pt) {{\tiny $6$}};
\node[above] at (120pt, 45pt) {{\tiny $6$}};

\end{scope}

\end{tikzpicture}
\caption{$\rho$ and $\rho'$ can be distinguished by $P \until Q$.}%
\label{fig:mtlefexample}
\end{figure}
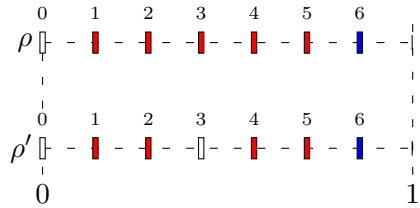

In the $1$-round \mtlpast{} EF game on $\rho$, $\rho'$ with $(i_0, j_0) = (0, 0)$,
a winning strategy for \emph{Spoiler} can be described as follows:
\begin{enumerate}
\item The two events that correspond to $(i_0, j_0) = (0, 0)$ in $\rho$ and $\rho'$ satisfy
the same set of monadic predicates, so \emph{Spoiler} does not win here.
\item \emph{Spoiler} chooses $I = (0, \infty)$ and $i_0' = 6$ in $\rho$.
\item If \emph{Duplicator} chooses $j_0' \neq 6$ in $\rho'$, she will lose at the beginning
of round $1$. So she chooses $j_0' = 6$.
\item \emph{Spoiler} plays the $\until$-part and chooses $j_0'' = 3$ in $\rho'$.
\item \emph{Duplicator} can only choose $i_0''$ in $\rho$ such that $1 \leq i_0'' \leq 5$.
But she will then lose at the beginning of round $1$.
\end{enumerate}
It follows that there is an \mtlpast{} formula of modal depth $1$
that distinguishes $\rho$ and $\rho'$.
One such formula is $P \until Q$, which can be obtained
from \emph{Spoiler}'s winning strategy above.
\end{exa}

\subsection{A hierarchy of expressiveness}\label{subsec:hierarchy}

We now present a sequence of successively more expressive
extensions of \mtl{} over bounded timed words.
The technique we use here is to construct two \emph{families} of
models---parametrised by $m$---such that there is a certain formula
of the more expressive logic telling them apart for all $m$, yet they cannot be distinguished by any
formula of the less expressive logic with modal depth $\leq m$ (i.e.,~\emph{Duplicator} has
a winning strategy in the corresponding $m$-round \mtlpast{} EF game).
Along the way we highlight the key features that give rise to the differences in expressiveness.
The necessity of new modalities
is justified by the fact that no known extension can lead to expressive completeness.

\paragraph{\emph{Definability of the beginning of time}}

Recall that \mtl{} and \foone{} have the same expressiveness over $\ropen{0, N}$-signals~\cite{Ouaknine2009}.
This result fails in the pointwise semantics.
\begin{prop}[Corollary of~{\cite[Section 8]{Prabhakar2006}}]
\emph{\mtlpast{}} is strictly more expressive than \emph{\mtl{}} over $\ropen{0, N}$-timed words.\footnote{The models constructed in~{\cite[Section 8]{Prabhakar2006}} are bounded timed words.}
\end{prop}
To explain this discrepancy between the two semantics, observe that a distinctive feature
of the continuous interpretation of \mtl{} is exploited in~\cite{Ouaknine2009}:
in any $\ropen{0, N}$-signal, the formula $\eventually_{=(N-1)} \mathbf{true}$ holds in $\ropen{0, 1}$
and nowhere else. One can make use of conjunctions of similar formulas to determine the integer part of
the current instant (where the relevant formula is being evaluated).
Unfortunately, since the duration of a given bounded timed word is not known \emph{a priori},
this trick does not work for \mtl{} in the pointwise semantics.
For example, the formula $\eventually_{=1} \mathbf{true}$ does not hold at any position in the
$\ropen{0, 2}$-timed word $\rho = (\sigma_0, 0)(\sigma_1, 0.5)$.
However, the same effect can be achieved in \mtlpast{} by using past modalities.
Let
\[
\varphi_{i, i+1} = \once_{\ropen{i, i+1}} (\neg \once \mathbf{true})\label{def:intformulas}
\]
and
$\Phi_{\mathit{int}} = \{\varphi_{i, i+1} \mid i \in \mathbb{N}\}$.
Note that the subformula $\neg \once \mathbf{true}$ can only hold at the very first event (with timestamp $0$),
thus $\varphi_{i, i+1}$ holds only at events with timestamps in $\ropen{i, i + 1}$.
Denote by \mtlunit{} the extension of \mtl{} obtained by allowing these formulas as atomic formulas.
It turned out that this very restrictive use of past modalities strictly increases the expressiveness of \mtl{}
over bounded timed words. Indeed, the main result of this section (Theorem~\ref{thm:boundedexpcomp})
crucially depends on the use of these formulas.

\begin{prop}\label{prop:unit-strict}
\emph{\mtlunit{}} is strictly more expressive than \emph{\mtl{}} over $\ropen{0, N}$-timed words.
\end{prop}
\begin{proof}
For a given $m \in \mathbb{N}$, we construct the following models:
\[
    \begin{array}{rcl}
    \mathcal{A}_{m} & = & (\emptyset, 0)(\emptyset, 1-\frac{1.5}{2m+5})(\emptyset, 1-\frac{0.5}{2m+5})\ldots(\emptyset, 1+\frac{m+2.5}{2m+5}) \,, \\
    \mathcal{B}_{m} & = & (\emptyset, 0)(\emptyset, 1-\frac{0.5}{2m+5})(\emptyset, 1+\frac{0.5}{2m+5})\ldots(\emptyset, 1+\frac{m+3.5}{2m+5}) \,.
    \end{array}
\]

\begin{figure}[ht]
\centering
\begin{tikzpicture}[scale=1.0]
\begin{scope}[>=latex]

\draw[|-|, loosely dashed] (0pt,40pt) -- (280pt,40pt) node[at start, left] {$\mathcal{A}_{m}$};
\draw[|-|, loosely dashed] (0pt,0pt) -- (280pt,0pt) node[at start, left] {$\mathcal{B}_{m}$};

\draw[loosely dashed] (0pt,-10pt) -- (0pt,50pt) node[at start,below] {$0$};
\draw[loosely dashed] (140pt,-10pt) -- (140pt,50pt) node[at start,below] {$1$};
\draw[loosely dashed] (210pt,-10pt) -- (210pt,50pt) node[at start,below] {$1.5$};
\draw[loosely dashed] (280pt,-10pt) -- (280pt,50pt) node[at start,below] {$2$};

\draw[draw=black, fill=white] (1pt, -4pt) rectangle (-1pt, 4pt);
\draw[draw=black, fill=white] (1pt, 36pt) rectangle (-1pt, 44pt);

\draw[draw=black, fill=white] (126pt, 36pt) rectangle (124pt, 44pt);

\draw[draw=black, fill=white] (136pt, -4pt) rectangle (134pt, 4pt);
\draw[draw=black, fill=white] (136pt, 36pt) rectangle (134pt, 44pt);

\draw[draw=black, fill=white] (146pt, -4pt) rectangle (144pt, 4pt);
\draw[draw=black, fill=white] (146pt, 36pt) rectangle (144pt, 44pt);

\draw[draw=black, fill=white] (156pt, -4pt) rectangle (154pt, 4pt);

\draw[-, very thick, loosely dotted] (165pt,20pt) -- (190pt,20pt);

\draw[draw=black, fill=white] (201pt, 36pt) rectangle (199pt, 44pt);

\draw[draw=black, fill=white] (211pt, -4pt) rectangle (209pt, 4pt);
\draw[draw=black, fill=white] (211pt, 36pt) rectangle (209pt, 44pt);

\draw[draw=black, fill=white] (221pt, -4pt) rectangle (219pt, 4pt);

\end{scope}

\end{tikzpicture}
\caption{Models $\mathcal{A}_{m}$ and $\mathcal{B}_{m}$.}%
\label{fig:unit-expressiveness}
\end{figure}
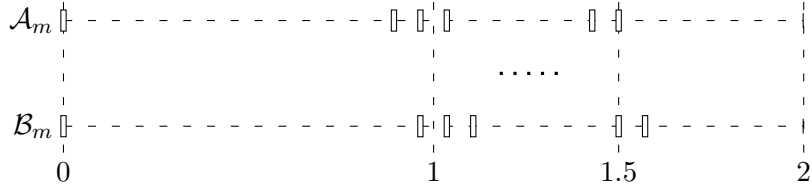
The models are illustrated in Figure~\ref{fig:unit-expressiveness}, where each white box represents an event (at which no monadic predicate holds).
We play an $m$-round \mtlpast{} EF game on $\mathcal{A}_{m}$, $\mathcal{B}_{m}$ and allow only $\until_{I}$-move.
After round $0$, either (i) $i_1 = j_1 \geq 1$ (in which case \emph{Duplicator} can, obviously, win the remaining rounds)
or (ii) $(i_1, j_1) = (2, 1)$ (\emph{Spoiler} chooses position $2$ in $\mathcal{A}_{m}$) or $(i_1, j_1) = (3, 2)$ (\emph{Spoiler} chooses position $2$ in $\mathcal{B}_{m}$). In the latter case, it is easy to verify that
in any remaining round $r$, \emph{Duplicator} can make $i_{r+1} = j_{r+1} \geq 1$ or $(i_{r+1}, j_{r+1}) = (i_r + 1, j_r + 1)$.
It follows from the \mtlpast{} EF Theorem that no \mtl{} formula of modal depth $\leq m$ can distinguish $\mathcal{A}_{m}$
and $\mathcal{B}_{m}$; however, the formula
\[
\eventually_{(0, 1)} ( \varphi_{0, 1} \wedge \nextx \varphi_{0, 1} ) \,,
\]
which says ``in the next time unit there are two events with timestamps in $\ropen{0, 1}$'', distinguishes $\mathcal{A}_{m}$ and $\mathcal{B}_{m}$ for any $m \in \mathbb{N}$
(when evaluated at position $0$).
\end{proof}

\paragraph{\emph{Past modalities}}
The conservative extension above uses past modalities in a very restricted way.
This is not sufficient for obtaining the full expressiveness of \mtlpast{}:
the following proposition says that non-trivial nesting of future modalities and past modalities
gives more expressiveness.
\begin{prop}\label{prop:past-strict}
\emph{\mtlpast{}} is strictly more expressive than \emph{\mtlunit{}} over $\ropen{0, N}$-timed words.
\end{prop}
\begin{proof}
For a given $m \in \mathbb{N}$, we construct
\[
    \begin{array}{rcl}
    \mathcal{C}_{m} & = & (\emptyset, 0)(\emptyset, \frac{0.5}{2m+3})(\emptyset, \frac{1.5}{2m+3})\ldots(\emptyset, 2 - \frac{0.5}{2m+3}) \,.
    \end{array}
\]
$\mathcal{D}_{m}$ is constructed as $\mathcal{C}_{m}$ except that the event at time $\frac{m+1.5}{2m+3} = 0.5$ is missing.
\begin{figure}[h]
\centering
\begin{tikzpicture}[scale=1.0]
\begin{scope}[>=latex]

\draw[|-|, loosely dashed] (0pt,40pt) -- (260pt,40pt) node[at start, left] {$\mathcal{C}_{m}$};
\draw[|-|, loosely dashed] (0pt,0pt) -- (260pt,0pt) node[at start, left] {$\mathcal{D}_{m}$};

\draw[loosely dashed] (0pt,-10pt) -- (0pt,50pt) node[at start,below] {$0$};
\draw[loosely dashed] (130pt,-10pt) -- (130pt,50pt) node[at start,below] {$1$};
\draw[loosely dashed] (260pt,-10pt) -- (260pt,50pt) node[at start,below] {$2$};

\draw[draw=black, fill=white] (1pt, -4pt) rectangle (-1pt, 4pt);
\draw[draw=black, fill=white] (1pt, 36pt) rectangle (-1pt, 44pt);

\draw[draw=black, fill=white] (6pt, -4pt) rectangle (4pt, 4pt);
\draw[draw=black, fill=white] (6pt, 36pt) rectangle (4pt, 44pt);

\draw[draw=black, fill=white] (16pt, -4pt) rectangle (14pt, 4pt);
\draw[draw=black, fill=white] (16pt, 36pt) rectangle (14pt, 44pt);

\draw[-, very thick, loosely dotted] (26pt,20pt) -- (44pt,20pt);

\draw[draw=black, fill=white] (56pt, -4pt) rectangle (54pt, 4pt);
\draw[draw=black, fill=white] (56pt, 36pt) rectangle (54pt, 44pt);

\draw[draw=black, fill=white] (66pt, 36pt) rectangle (64pt, 44pt);

\draw[draw=black, fill=white] (76pt, -4pt) rectangle (74pt, 4pt);
\draw[draw=black, fill=white] (76pt, 36pt) rectangle (74pt, 44pt);

\draw[-, very thick, loosely dotted] (86pt,20pt) -- (104pt,20pt);

\draw[draw=black, fill=white] (114pt, -4pt) rectangle (116pt, 4pt);
\draw[draw=black, fill=white] (114pt, 36pt) rectangle (116pt, 44pt);

\draw[draw=black, fill=white] (124pt, -4pt) rectangle (126pt, 4pt);
\draw[draw=black, fill=white] (124pt, 36pt) rectangle (126pt, 44pt);

\draw[draw=black, fill=white] (136pt, -4pt) rectangle (134pt, 4pt);
\draw[draw=black, fill=white] (136pt, 36pt) rectangle (134pt, 44pt);

\draw[draw=black, fill=white] (146pt, -4pt) rectangle (144pt, 4pt);
\draw[draw=black, fill=white] (146pt, 36pt) rectangle (144pt, 44pt);

\draw[-, very thick, loosely dotted] (156pt,20pt) -- (234pt,20pt);

\draw[draw=black, fill=white] (246pt, -4pt) rectangle (244pt, 4pt);
\draw[draw=black, fill=white] (246pt, 36pt) rectangle (244pt, 44pt);

\draw[draw=black, fill=white] (256pt, -4pt) rectangle (254pt, 4pt);
\draw[draw=black, fill=white] (256pt, 36pt) rectangle (254pt, 44pt);

\end{scope}

\end{tikzpicture}
\caption{Models $\mathcal{C}_{m}$ and $\mathcal{D}_{m}$.}%
\label{fig:past-expressiveness}
\end{figure}
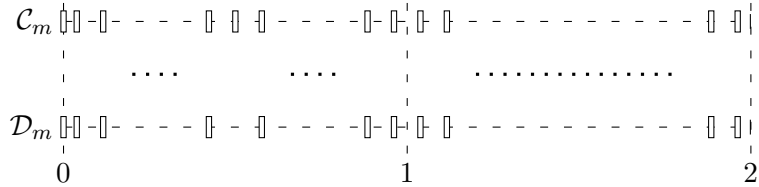

The models are illustrated in Figure~\ref{fig:past-expressiveness}, where each white box represents an event (at which no monadic predicate holds).
We play an $m$-round \mtlpast{} EF game on $\mathcal{C}_{m}$ and $\mathcal{D}_{m}$, allowing only $\until_{I}$-move.
For simplicity, assume that we can use special monadic predicates to refer to formulas in $\Phi_{\mathit{int}}$.
In each round $r$, \emph{Duplicator} can either make (i) $i_{r+1} = j_{r+1} + 1$ and $i_{r+1} \geq m + 3$ (in which case she can win the remaining rounds)
or (ii) $i_{r+1} = j_{r+1}$ and $i_{r+1}$ is not equal to $2m + 2$, $2m + 3$ or $4m + 5$.
It follows from the \mtlpast{} EF Theorem that no \mtlunit{} formula of modal depth $\leq m$ can distinguish $\mathcal{C}_{m}$
and $\mathcal{D}_{m}$; but the formula
\[
\globally_{(1, 2)} (\once_{= 1} \mathbf{true}) \,,
\]
which says ``for each event in $(1, 2)$ from now, there is a corresponding event exactly $1$ time unit earlier'',
distinguishes $\mathcal{C}_{m}$ and $\mathcal{D}_{m}$ for any $m \in \mathbb{N}$ (when evaluated at position $0$).
\end{proof}

\paragraph{\emph{Counting modalities}}
The modality $C_n(x, X)$ asserts that $X$ holds at least at $n$ points in the open interval $(x, x + 1)$.
The modalities $C_n$ for $n \geq 2$ are called \emph{counting modalities}.
It is well-known that these modalities are not expressible in \mtlpast{} over signals~\cite{Hirshfeld2007}.
For this reason, they (and variants thereof) are often used to prove
inexpressiveness results for various metric logics.
For example, the following property
\begin{itemize}
\item $P$ holds at an event at time $y$ in the future
\item $Q$ holds at an event at time $y' > y$
\item $R$ holds at an event at time $y'' > y' > y$
\item Both the $Q$-event and the $R$-event are within $(1, 2)$ from the $P$-event
\end{itemize}
can be expressed as the \foone{} formula
\[
    \arraycolsep=0.3ex
    \begin{array}{rcll}
    \vartheta_{\mathit{pqr}}(x) & = & \exists y \, \biggl( x < y \ \wedge & P(y) \wedge \exists y' \, \Bigl(y < y' \wedge d(y, y') > 1 \wedge d(y, y') < 2  \wedge Q(y') \\
    & & & {} \wedge \exists y'' \, \bigl(y' < y'' \wedge d(y, y'') > 1 \wedge d(y, y'') < 2 \wedge R(y'') \bigr)\Bigr)\biggr) \,,
    \end{array}
\]
yet it has no equivalent in \mtlpast{} over timed words~\cite{Pandya2011}.
The difficulty here is that while we can easily write `there is a $Q$-event within $(1, 2)$
from a $P$-event in the future' as $\eventually(P \wedge \eventually_{(1, 2)} Q)$, it is not possible
to express `there is a $R$-event after the $Q$-event' and `that $R$-event is within $(1, 2)$ from the $P$-event'
simultaneously in \mtlpast{}.
Indeed, it was shown recently that in the continuous semantics, \mtlpast{} extended with counting modalities and their past counterparts
(which we denote by \mtlpastcnt{})
is expressively complete for \foone{}~\cite{Hunter2013}.
In other words, counting modalities are exactly what separates
the expressiveness of \mtlpast{} and \foone{} in the continuous semantics.
In the time-bounded pointwise case, however,
they add no expressiveness to \mtlpast{}.
To see this, observe that the following formula is equivalent to $\vartheta_{\mathit{pqr}}(x)$ over $\ropen{0, N}$-timed words (we make use of the formulas in $\Phi_{\mathit{int}}$ defined earlier):
\[
    \arraycolsep=0.3ex
    \begin{array}{ll}
    \eventually \Biggl(\bigvee_{0 \leq i \leq N - 1} \biggl( P \wedge \varphi_{i, i + 1} \wedge \Bigl(& \underbrace{\eventually_{> 1} \bigl(Q \wedge \eventually(R \wedge \varphi_{i + 1, i + 2}) \bigr)}_\text{Case (i)} \\
    & {} \vee \underbrace{\eventually_{< 2} \bigl(R \wedge \varphi_{i + 2, i + 3} \wedge \once (Q \wedge \varphi_{i + 2, i + 3}) \bigr)}_\text{Case (ii)} \\
    & {} \vee \underbrace{\bigl(\eventually_{> 1} (Q \wedge \varphi_{i + 1, i + 2}) \wedge \eventually_{< 2} (R \wedge \varphi_{i + 2, i + 3}) \bigr)}_\text{Case (iii)} \Bigr) \biggr) \Biggr) \,.
    \end{array}
\]
The three cases that correspond to the subformulas are illustrated in Figure~\ref{fig:countingex}
where time is measured relative to the very first event (with timestamp $0$).
Note how we use the `integer boundaries' as an alternative distance measure and thus ensure that
both the $Q$-event and the $R$-event are within $(1, 2)$ from the $P$-event.

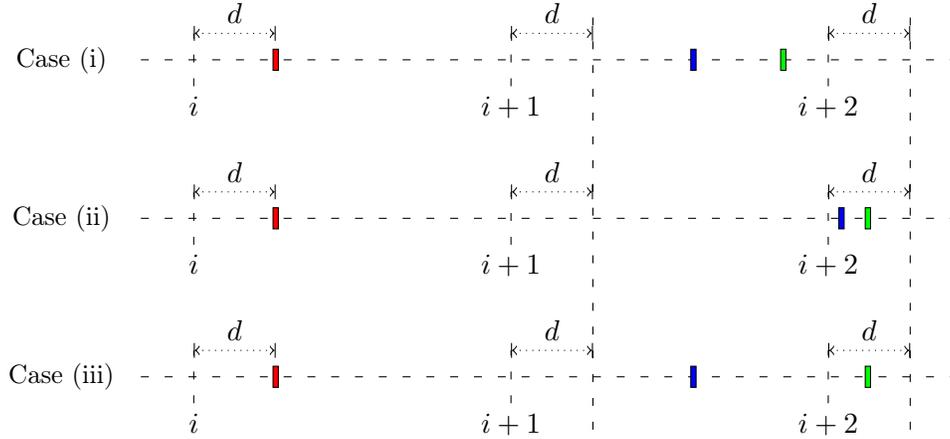
\begin{figure}[ht]
\centering
\begin{tikzpicture}[scale=1.0]
\begin{scope}[>=latex]

\draw[-, loosely dashed] (-20pt,0pt) -- (290pt,0pt);

\draw[-, loosely dashed] (0pt,-10pt) -- (0pt,10pt) node[at start, below] {$i$};

\draw[-, loosely dashed] (120pt,-10pt) -- (120pt,10pt) node[at start, below] {$i + 1$};
\draw[-, loosely dashed] (240pt,-10pt) -- (240pt,10pt) node[at start, below] {$i + 2$};

\draw[draw=black, fill=red] (32pt, -4pt) rectangle (30pt, 4pt);
\draw[draw=black, fill=blue] (190pt, -4pt) rectangle (188pt, 4pt);
\draw[draw=black, fill=green] (224pt, -4pt) rectangle (222pt, 4pt);

\draw[-, loosely dashed] (-20pt,-60pt) -- (290pt,-60pt);

\draw[-, loosely dashed] (0pt,-70pt) -- (0pt,-50pt) node[at start, below] {$i$};

\draw[-, loosely dashed] (120pt,-70pt) -- (120pt,-50pt) node[at start, below] {$i + 1$};
\draw[-, loosely dashed] (240pt,-70pt) -- (240pt,-50pt) node[at start, below] {$i + 2$};

\draw[draw=black, fill=red] (32pt, -64pt) rectangle (30pt, -56pt);
\draw[draw=black, fill=blue] (246pt, -64pt) rectangle (244pt, -56pt);
\draw[draw=black, fill=green] (256pt, -64pt) rectangle (254pt, -56pt);

\draw[-, loosely dashed] (-20pt,-120pt) -- (290pt,-120pt);

\draw[-, loosely dashed] (0pt,-130pt) -- (0pt,-110pt) node[at start, below] {$i$};

\draw[-, loosely dashed] (120pt,-130pt) -- (120pt,-110pt) node[at start, below] {$i + 1$};
\draw[-, loosely dashed] (240pt,-130pt) -- (240pt,-110pt) node[at start, below] {$i + 2$};

\draw[draw=black, fill=red] (32pt, -124pt) rectangle (30pt, -116pt);
\draw[draw=black, fill=blue] (190pt, -124pt) rectangle (188pt, -116pt);
\draw[draw=black, fill=green] (256pt, -124pt) rectangle (254pt, -116pt);

\draw[-, loosely dashed] (151pt,-140pt) -- (151pt,20pt);
\draw[-, loosely dashed] (271pt,-140pt) -- (271pt,20pt);

\node at (-50pt, 0pt) {\small{Case (i)}};
\node at (-50pt, -60pt) {\small{Case (ii)}};
\node at (-50pt, -120pt) {\small{Case (iii)}};

\end{scope}

\begin{scope}[>=to]
\draw[|<->|][dotted] (0pt,10pt)  -- (31pt,10pt) node[midway,above] {{$d$}};
\draw[|<->|][dotted] (120pt,10pt)  -- (151pt,10pt) node[midway,above] {{$d$}};
\draw[|<->|][dotted] (240pt,10pt)  -- (271pt,10pt) node[midway,above] {{$d$}};

\draw[|<->|][dotted] (0pt,-50pt)  -- (31pt,-50pt) node[midway,above] {{$d$}};
\draw[|<->|][dotted] (120pt,-50pt)  -- (151pt,-50pt) node[midway,above] {{$d$}};
\draw[|<->|][dotted] (240pt,-50pt)  -- (271pt,-50pt) node[midway,above] {{$d$}};

\draw[|<->|][dotted] (0pt,-110pt)  -- (31pt,-110pt) node[midway,above] {{$d$}};
\draw[|<->|][dotted] (120pt,-110pt)  -- (151pt,-110pt) node[midway,above] {{$d$}};
\draw[|<->|][dotted] (240pt,-110pt)  -- (271pt,-110pt) node[midway,above] {{$d$}};
\end{scope}

\end{tikzpicture}
\caption{Counting modalities is expressible in \mtlpast{} over $\ropen{0, N}$-timed words. The red, blue, and green boxes represent $P$-events, $Q$-events, and $R$-events respectively.}%
\label{fig:countingex}
\end{figure}
The same idea can readily be generalised to handle counting modalities and their past counterparts.
We therefore have the following proposition.

\begin{prop}\label{prop:countinguseless}
\emph{\mtlpast{}} is expressively complete for \emph{\mtlpastcnt{}} over $\ropen{0, N}$-timed words.
\end{prop}

\paragraph{\emph{Non-local properties (one reference point)}}
Proposition~\ref{prop:countinguseless} shows that a part of the expressiveness hierarchy
of metric logics over ($\mathbb{R}_{\geq 0}$-)timed words collapses in a time-bounded pointwise setting. % chktex 36
Nonetheless, \mtlpast{} is still not expressive enough to capture the whole of \foone{}
in such a setting. Recall that another feature of the continuous interpretation of \mtl{}
used in the proof in~\cite{Ouaknine2009} is that
$\eventually_{= k} \varphi$ holds at $t$ \emph{iff} $\varphi$ holds at $t + k$.
Suppose that we want to specify the following property over $\mathbf{P} = \{P, Q\}$
for some positive integer $a$ (let the current instant be $t_1$):
\begin{itemize}
\item There is an event at time $t_2 > t_1 + a$ where $Q$ holds
\item $P$ holds at all events in $(t_1 + a, t_2)$.
\end{itemize}
In the continuous semantics,  the property can easily be expressed as
the following \mtl{} formula
\[
\varphi_{\mathit{cont1}} = \eventually_{=a} \bigl( (P \vee P_{\epsilon}) \until Q \bigr)
\]
over signals of the form $f^\rho$ (over $\Sigma_{\mathbf{P'}}$ where $\mathbf{P'} = \mathbf{P} \cup \{P_\epsilon\}$); see Figure~\ref{fig:puntilq} for an example where
the formula $\varphi_{\mathit{cont1}}$ holds at $t_1$ in the continuous semantics.

\begin{figure}[h]
\centering
\begin{tikzpicture}[scale=1]
\begin{scope}[>=latex]

\draw[-, loosely dashed] (-5pt,0pt) -- (285pt,0pt);

\draw[-, loosely dashed] (0pt,-10pt) -- (0pt,10pt) node[at start, below] {$t_1$};

\draw[-, loosely dashed] (220pt,-10pt) -- (220pt,10pt) node[at start, below] {$t_1 + a$};
\draw[-, loosely dashed] (260pt,-10pt) -- (260pt,10pt) node[at start, below] {$t_2$};

\draw[draw=black, fill=white] (1pt, -4pt) rectangle (-1pt, 4pt);

\draw[draw=black, fill=red] (231pt, -4pt) rectangle (233pt, 4pt);

\draw[draw=black, fill=red] (240pt, -4pt) rectangle (242pt, 4pt);

\draw[draw=black, fill=red] (254pt, -4pt) rectangle (256pt, 4pt);

\draw[draw=black, fill=blue] (261pt, -4pt) rectangle (259pt, 4pt);

\end{scope}

\end{tikzpicture}
\caption{$\varphi_{\mathit{cont1}}$ holds at $t_1$ in the continuous semantics. The red boxes denote $P$-events and the blue boxes denote $Q$-events.}%
\label{fig:puntilq}
\end{figure}
Essentially, when the current instant is $t_1$, the continuous interpretation of \mtlpast{} allows one to speak of
events `from' $t_1 + a$, regardless of whether there is an actual event at $t_1 + a$.
As we will show, it is not possible to do the same with the pointwise interpretation of \mtlpast{} when there is no event at $t_1 + a$.
To remedy this issue within the pointwise semantic framework, we introduce a simple
family of modalities $\first_I$ (`Beginning')
and their past versions $\pfirst_I$.
They can be used to refer to the \emph{first} (earliest or latest, respectively) event in a given interval.
For example, we define the modality that asserts ``$X$ holds at the first event in $(a, b)$ relative
to the current instant'' as the following \foone{} formula:
%no other event in $(a, b)$ with a smaller timestamp:
\[
    \arraycolsep=0.3ex
    \begin{array}{rcll}
    \first_{(a, b)}(x, X) & = & \exists x' \, \Bigl( & x < x' \wedge d(x, x') > a \wedge d(x, x') < b \wedge X(x') \\
    & & & {} \wedge \nexists x'' \, \bigl(x < x'' \wedge x'' < x' \wedge d(x, x'') > a \bigr) \Bigr) \,.
    \end{array}
\]
The property above can now be written as $\first_{(a, \infty)} \bigl( Q \vee (P \until Q) \bigr)$ in
the pointwise semantics.
We refer to the extension of \mtlpast{} with $\first_I, \pfirst_I$
as \mtlbpast{}.\footnote{Readers may find the modalities $\first_I$ similar to
the modalities $\triangleright_I$ in \emph{Event-Clock Logic}~\cite{Henzinger1998}.
The difference is that the formula $\first_I \varphi$ requires
$\varphi$ to hold at the \emph{first} event in $I$, whereas the formula $\triangleright_I \varphi$
requires (i) $\varphi$ to hold at \emph{some} event in $I$ and that (ii)
$\varphi$ does not hold at any other event between the current instant and the time of that event.}
The following proposition states that this extension is indeed non-trivial.

\begin{prop}\label{prop:beginning-strict}
\emph{\mtlbpast{}} is strictly more expressive than \emph{\mtlpast{}} over $\ropen{0, N}$-timed words.
\end{prop}
\begin{proof}
The proof we give here is inspired by a proof in~{\cite[Section 5]{Pandya2011}}.
Given $m \in \mathbb{N}$, we describe models $\mathcal{E}_{m}$
and $\mathcal{F}_{m}$ that are indistinguishable by \mtlpast{} formulas of modal depth $\leq m$
but distinguishing in \mtlbpast{}.

We first describe $\mathcal{F}_{m}$.
Let $g = \frac{1}{2m + 6}$ and pick positive $\varepsilon < \frac{g}{\frac{1}{g} - 1}$.
The first event (at time $0$) satisfies $\neg P \wedge \neg Q$.
Then, a sequence of overlapping segments (arranged as described below) starts at time $\frac{0.5}{2m + 5}$; see Figure~\ref{fig:segment} for an illustration of a segment.
Each segment consists of an
event satisfying $P \wedge \neg Q$
and an event satisfying $\neg P \wedge Q$
(we refer to them as $P$-events and $Q$-events, respectively).
If the $P$-event in the $i^{\mathit{th}}$ segment is at time $t$, then its $Q$-event is at time $t+2g+ \frac{1}{2} \varepsilon$.
All $P$-events in neighbouring segments are separated by $g - \frac{g}{\frac{1}{g} - 1}$.
We put a total of $4m + 12$ segments.
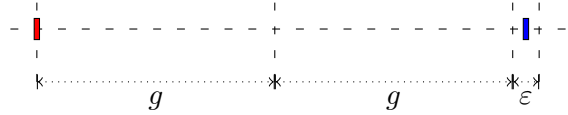
\begin{figure}[h]
\centering
\begin{tikzpicture}[scale=1]
\begin{scope}[>=latex]

\draw[-, loosely dashed] (-10pt,0pt) -- (200pt,0pt);

\draw[loosely dashed] (0pt,-20pt) -- (0pt,10pt);
\draw[loosely dashed] (90pt,-20pt) -- (90pt,10pt);
\draw[loosely dashed] (180pt,-20pt) -- (180pt,10pt);
\draw[loosely dashed] (190pt,-20pt) -- (190pt,10pt);

\draw[draw=black, fill=red] (1pt, -4pt) rectangle (-1pt, 4pt);

\draw[draw=black, fill=blue] (186pt, -4pt) rectangle (184pt, 4pt);

\end{scope}
\begin{scope}[>=to]
\draw[|<->|][dotted] (0pt,-20pt)  -- (90pt,-20pt) node[midway,below] {{$g$}};
\draw[|<->|][dotted] (90pt,-20pt)  -- (180pt,-20pt) node[midway,below] {$g$};
\draw[|<->|][dotted] (180pt,-20pt)  -- (190pt,-20pt) node[midway,below] {$\varepsilon$};
\end{scope}
\end{tikzpicture}
\caption{A single segment in $\mathcal{F}_{m}$. The red box denotes a $P$-event and the blue box denotes a $Q$-event.}%
\label{fig:segment}
\end{figure}

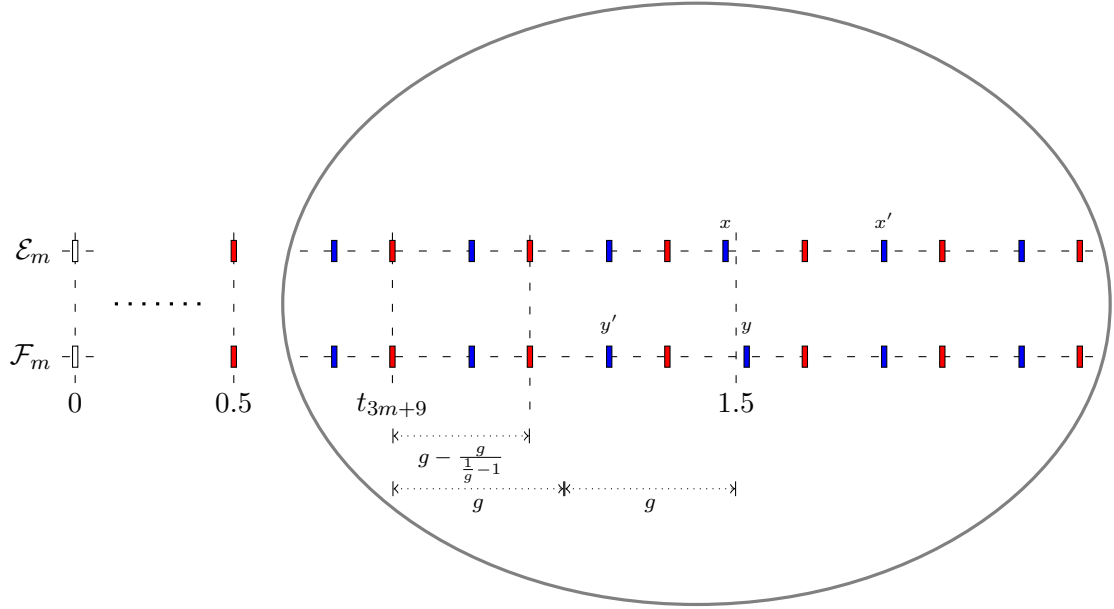
\begin{figure}[t]
\centering
\begin{tikzpicture}[scale=1]
\begin{scope}[>=latex]

\draw[-, loosely dashed] (-35pt,40pt) -- (265pt,40pt);
\draw[-, loosely dashed] (-35pt,0pt) -- (265pt,0pt);
\draw[-, loosely dashed] (-125pt,40pt) -- (-110pt,40pt) node[at start, left] {$\mathcal{E}_{m}$};
\draw[-, loosely dashed] (-125pt,0pt) -- (-110pt,0pt) node[at start, left] {$\mathcal{F}_{m}$};

\draw[-, very thick, loosely dotted] (-105pt,20pt) -- (-70pt,20pt);

\draw[very thick, draw=gray] (115pt, 20pt) ellipse (5.5cm and 4cm);

\draw[loosely dashed] (0pt,-10pt) -- (0pt,50pt) node[at start, below] {$t_{3m+9}$};
\draw[loosely dashed] (130pt,-10pt) -- (130pt,50pt) node[at start, below] {$1.5$};
\draw[loosely dashed] (-60pt,-10pt) -- (-60pt,50pt) node[at start, below] {$0.5$};
\draw[loosely dashed] (-120pt,-10pt) -- (-120pt,50pt) node[at start, below] {$0$};
\draw[loosely dashed] (52pt,-20pt) -- (52pt,50pt);

\draw[draw=black, fill=white] (-121pt, -4pt) rectangle (-119pt, 4pt);
\draw[draw=black, fill=white] (-121pt, 36pt) rectangle (-119pt, 44pt);

\draw[draw=black, fill=red] (-61pt, -4pt) rectangle (-59pt, 4pt);
\draw[draw=black, fill=red] (-61pt, 36pt) rectangle (-59pt, 44pt);

\draw[draw=black, fill=red] (1pt, -4pt) rectangle (-1pt, 4pt);
\draw[draw=black, fill=red] (1pt, 36pt) rectangle (-1pt, 44pt);

\draw[draw=black, fill=red] (53pt, -4pt) rectangle (51pt, 4pt);
\draw[draw=black, fill=red] (53pt, 36pt) rectangle (51pt, 44pt);

\draw[draw=black, fill=red] (105pt, -4pt) rectangle (103pt, 4pt);
\draw[draw=black, fill=red] (105pt, 36pt) rectangle (103pt, 44pt);

\draw[draw=black, fill=red] (157pt, -4pt) rectangle (155pt, 4pt);
\draw[draw=black, fill=red] (157pt, 36pt) rectangle (155pt, 44pt);

\draw[draw=black, fill=red] (209pt, -4pt) rectangle (207pt, 4pt);
\draw[draw=black, fill=red] (209pt, 36pt) rectangle (207pt, 44pt);

\draw[draw=black, fill=red] (261pt, -4pt) rectangle (259pt, 4pt);
\draw[draw=black, fill=red] (261pt, 36pt) rectangle (259pt, 44pt);

\draw[draw=black, fill=blue] (-21pt, -4pt) rectangle (-23pt, 4pt);
\draw[draw=black, fill=blue] (-21pt, 36pt) rectangle (-23pt, 44pt);

\draw[draw=black, fill=blue] (31pt, -4pt) rectangle (29pt, 4pt);
\draw[draw=black, fill=blue] (31pt, 36pt) rectangle (29pt, 44pt);

\draw[draw=black, fill=blue] (83pt, -4pt) rectangle (81pt, 4pt);
\node[above] at (82pt, 5pt) {{\tiny $y'$}};
\draw[draw=black, fill=blue] (83pt, 36pt) rectangle (81pt, 44pt);

\draw[draw=black, fill=blue] (135pt, -4pt) rectangle (133pt, 4pt);
\node[above] at (134pt, 5pt) {{\tiny $y$}};

\draw[draw=black, fill=blue] (127pt, 36pt) rectangle (125pt, 44pt);
\node[above] at (126pt, 45pt) {{\tiny $x$}};

\draw[draw=black, fill=blue] (187pt, -4pt) rectangle (185pt, 4pt);
\draw[draw=black, fill=blue] (187pt, 36pt) rectangle (185pt, 44pt);
\node[above] at (186pt, 45pt) {{\tiny $x'$}};

\draw[draw=black, fill=blue] (239pt, -4pt) rectangle (237pt, 4pt);
\draw[draw=black, fill=blue] (239pt, 36pt) rectangle (237pt, 44pt);

\end{scope}

\begin{scope}[>=to]
\draw[|<->|][dotted] (0pt,-30pt)  -- (52pt,-30pt) node[midway,below] {{\scriptsize $g - \frac{g}{\frac{1}{g} - 1}$}};
\draw[|<->|][dotted] (0pt,-50pt)  -- (65pt,-50pt) node[midway,below] {{\scriptsize $g$}};
\draw[|<->|][dotted] (65pt,-50pt)  -- (130pt,-50pt) node[midway,below] {{\scriptsize $g$}};
\end{scope}
\end{tikzpicture}
\caption{A close-up near the ${(3m+9)}^{\mathit{th}}$-segments in $\mathcal{E}_{m}$ and $\mathcal{F}_{m}$.}%
\label{fig:3mp9segments}
\end{figure}

$\mathcal{E}_{m}$ is almost identical to $\mathcal{F}_{m}$ except the ${(3m + 9)}^{\mathit{th}}$ segment.
Let this segment start at $t_{3m + 9}$. In $\mathcal{E}_{m}$, we move the corresponding $Q$-event to
$t+2g-\frac{1}{2} \varepsilon$ (see Figure~\ref{fig:3mp9segments}).
Note in particular that there are $P$-events at time $0.5$ in both models (in their ${(m+4)}^{\mathit{th}}$ segment).

The only difference in two models is a pair of $Q$-events, which we
denote by $x$ and $y$ respectively and write their corresponding timestamps
as $t_x$ and $t_y$ (see Figure~\ref{fig:3mp9segments}).
It is easy to verify that no two events are separated by an integer distance.
We say a configuration $(i, j)$ is \emph{identical} if $i = j$.
For $i \geq 1$, we denote by $\mathit{seg}(i)$ the segment that the $i^{\mathit{th}}$ event belongs to,
and we write $P(i)$ if the $i^{\mathit{th}}$ event is a $P$-event and $Q(i)$ if its a $Q$-event.

\begin{prop}\label{prop:beginning-strict-induction}
\emph{Duplicator} has a winning strategy for $m$-round \emph{\mtlpast{}} EF game
on $\mathcal{E}_{m}$ and $\mathcal{F}_{m}$ with $(i_0, j_0) = (0, 0)$.
In particular, she has a winning strategy such that for each round  $0 \leq r \leq m$,
the $i_r^{\mathit{th}}$ event in $\mathcal{E}_{m}$ and the $j_r^{\mathit{th}}$ event in $\mathcal{F}_{m}$
satisfy the same set of monadic predicates and
\begin{itemize}
\item if $i_r \neq j_r$, then
	\begin{itemize}
	\item $\mathit{seg}(i_r) - \mathit{seg}(j_r) < r$
	\item $(m + 1 - r) < \mathit{seg}(i_r), \mathit{seg}(j_r) < (m + 5 + r)$ or $(3m + 8 - r) < \mathit{seg}(i_r), \mathit{seg}(j_r) < (3m + 12 + r)$.
	\end{itemize}
\end{itemize}
\end{prop}

\noindent
We prove the proposition by induction on $r$. The idea is to try to make the resulting configurations identical.
When this is not possible \emph{Duplicator} simply imitates what \emph{Spoiler} does.
\begin{itemize}
\item \emph{Base step.} The proposition holds trivially for $(i_0, j_0) = (0, 0)$.
\item \emph{Induction step.} Suppose that the claim holds for $r < m$. We prove it
also holds for $r + 1$.
	\begin{itemize}
	\item $(i_r, j_r) = (0, 0)$:\\
	\emph{Duplicator} can always make $(i_{r+1}, j_{r+1})$ identical.
	\item $(i_r, j_r) \neq (0, 0)$ is identical:\\
	\emph{Duplicator} tries to make $(i_r', j_r')$ identical.
	This may only fail when
	\begin{itemize}
	\item $P(i_r)$, $P(j_r)$ and $\mathit{seg}(i_r) = \mathit{seg}(j_r) = m + 4$.
	\item $Q(i_r)$, $Q(j_r)$ and $\mathit{seg}(i_r) = \mathit{seg}(j_r) = 3m + 9$, i.e.,~$x$ and $y$.
	\end{itemize}
	In these cases, \emph{Duplicator} chooses another event in a neighbouring segment
	that minimises $|\mathit{seg}(i_r') - \mathit{seg}(j_r')|$.
	For example, if $(i_r, j_r)$ corresponds to $x$ and $y$ and \emph{Spoiler} chooses
	$j_r'$ such that $P(j_r')$ and $\mathit{seg}(j_r') = m + 4$ in a $\since_{(1, \infty)}$-move,
	\emph{Duplicator} chooses $i_r'$ with $\mathit{seg}(i_r') = m + 3$.
	If \emph{Spoiler} then plays $\once$-part, the resulting configuration $(i_{r+1}, j_{r+1}) = (i_r', j_r')$
	clearly satisfy the claim. If she plays $\since$-part, \emph{Duplicator} makes $(i_r'', j_r'')$
	identical whenever possible. Otherwise she chooses a suitable event that minimises
	$|\mathit{seg}(i_r'') - \mathit{seg}(j_r'')|$. For instance, if $Q(i_r'')$ and $\mathit{seg}(i_r'') = m + 1$,
	\emph{Duplicator} chooses $j_r''$ such that $Q(j_r'')$ and $\mathit{seg}(j_r'') = m + 2$.
	The resulting configuration $(i_{r+1}, j_{r+1}) = (i_r'', j_r'')$ clearly satisfies the claim.

	\item $(i_r, j_r)$ is not identical:\\
	\emph{Duplicator} tries to make $(i_r', j_r')$ identical. If this is not possible,
	then \emph{Duplicator} chooses an event that minimises $|\mathit{seg}(i_r') - \mathit{seg}(j_r')|$.
	For example, consider $\mathit{seg}(i_r) = m + 4$, $\mathit{seg}(j_r) = m + 3$ such that $P(i_r)$ and $P(j_r)$,
	and \emph{Spoiler} chooses $x$ in an $\until_{(0, 1)}$-move. In this case, \emph{Duplicator}
	cannot choose $y'$, but she may choose the first $Q$-event that happens before $y'$. \emph{Duplicator} responds to
	$\until$-parts and $\since$-parts in similar ways as before. It is easy to see that the claim holds.
	\end{itemize}

\end{itemize}
Proposition~\ref{prop:beginning-strict} now follows from Proposition~\ref{prop:beginning-strict-induction},
the \mtlpast{} EF Theorem, and the fact that
$\mathcal{E}_{m} \models \eventually (P \wedge \first_{(1, 2)} P)$ but
$\mathcal{F}_{m} \not \models \eventually (P \wedge \first_{(1, 2)} P)$.
\end{proof}

\paragraph{\emph{Non-local properties (two reference points)}}

Adding modalities $\first_I, \pfirst_I$ to \mtlpast{} allows one to specify properties
with respect to a distant time point even when there is no event at that point.
However, the following proposition shows that this is still not enough for expressive completeness.

\begin{prop}\label{prop:foone-strict}
\emph{\foone{}} is strictly more expressive than \emph{\mtlbpast} over $\ropen{0, N}$-timed words.
\end{prop}
\begin{proof}
This is similar to a proof in~{\cite[Section 7]{Prabhakar2006}}.
Given $m \in \mathbb{N}$, we construct two models as follows.
Let
\[
    \arraycolsep=0.3ex
    \begin{array}{rcl}
    \mathcal{G}_{m} & = & (\emptyset, 0)(\emptyset, \frac{0.5}{2m+3})(\emptyset, \frac{1.5}{2m+3})\ldots(\emptyset, 1-\frac{0.5}{2m+3}) \\
    & & (\emptyset, 1+\frac{0.5}{2m+2})(\emptyset, 1+\frac{1.5}{2m+2})\ldots\ldots(\emptyset, 2-\frac{0.5}{2m+2}) \,.
    \end{array}
\]
$\mathcal{H}_{m}$ is constructed as $\mathcal{G}_{m}$ except that the event at time $\frac{m+1.5}{2m+3} = 0.5$ is missing.
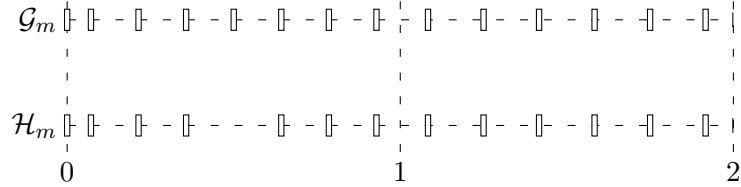
\begin{figure}[h]
\centering
\begin{tikzpicture}[scale=1.0]
\begin{scope}[>=latex]

\draw[|-|, loosely dashed] (0pt,40pt) -- (252pt,40pt) node[at start, left] {$\mathcal{G}_{m}$};
\draw[|-|, loosely dashed] (0pt,0pt) -- (252pt,0pt) node[at start, left] {$\mathcal{H}_{m}$};

\draw[loosely dashed] (0pt,-10pt) -- (0pt,50pt) node[at start,below] {$0$};
\draw[loosely dashed] (126pt,-10pt) -- (126pt,50pt) node[at start,below] {$1$};
\draw[loosely dashed] (252pt,-10pt) -- (252pt,50pt) node[at start,below] {$2$};

\draw[draw=black, fill=white] (1pt, -4pt) rectangle (-1pt, 4pt);
\draw[draw=black, fill=white] (1pt, 36pt) rectangle (-1pt, 44pt);

\draw[draw=black, fill=white] (10pt, -4pt) rectangle (8pt, 4pt);
\draw[draw=black, fill=white] (10pt, 36pt) rectangle (8pt, 44pt);

\draw[draw=black, fill=white] (28pt, -4pt) rectangle (26pt, 4pt);
\draw[draw=black, fill=white] (28pt, 36pt) rectangle (26pt, 44pt);

\draw[draw=black, fill=white] (46pt, -4pt) rectangle (44pt, 4pt);
\draw[draw=black, fill=white] (46pt, 36pt) rectangle (44pt, 44pt);

\draw[draw=black, fill=white] (64pt, 36pt) rectangle (62pt, 44pt);

\draw[draw=black, fill=white] (82pt, -4pt) rectangle (80pt, 4pt);
\draw[draw=black, fill=white] (82pt, 36pt) rectangle (80pt, 44pt);

\draw[draw=black, fill=white] (100pt, -4pt) rectangle (98pt, 4pt);
\draw[draw=black, fill=white] (100pt, 36pt) rectangle (98pt, 44pt);

\draw[draw=black, fill=white] (118pt, -4pt) rectangle (116pt, 4pt);
\draw[draw=black, fill=white] (118pt, 36pt) rectangle (116pt, 44pt);

\draw[draw=black, fill=white] (137.5pt, -4pt) rectangle (135.5pt, 4pt);
\draw[draw=black, fill=white] (137.5pt, 36pt) rectangle (135.5pt, 44pt);

\draw[draw=black, fill=white] (158.5pt, -4pt) rectangle (156.5pt, 4pt);
\draw[draw=black, fill=white] (158.5pt, 36pt) rectangle (156.5pt, 44pt);

\draw[draw=black, fill=white] (179.5pt, -4pt) rectangle (177.5pt, 4pt);
\draw[draw=black, fill=white] (179.5pt, 36pt) rectangle (177.5pt, 44pt);

\draw[draw=black, fill=white] (200.5pt, -4pt) rectangle (198.5pt, 4pt);
\draw[draw=black, fill=white] (200.5pt, 36pt) rectangle (198.5pt, 44pt);

\draw[draw=black, fill=white] (221.5pt, -4pt) rectangle (219.5pt, 4pt);
\draw[draw=black, fill=white] (221.5pt, 36pt) rectangle (219.5pt, 44pt);

\draw[draw=black, fill=white] (242.5pt, -4pt) rectangle (240.5pt, 4pt);
\draw[draw=black, fill=white] (242.5pt, 36pt) rectangle (240.5pt, 44pt);

\end{scope}

\end{tikzpicture}
\caption{Models $\mathcal{G}_{m}$ and $\mathcal{H}_{m}$ for $m = 2$.}%
\label{fig:prabhakarmodels}
\end{figure}

Figure~\ref{fig:prabhakarmodels} illustrates the models for the case $m = 2$ where white boxes
represent events at which no monadic predicate holds.
Observe that no two events are separated by an integer distance. We say that a configuration
$(i, j)$ is \emph{synchronised} if they correspond to events with the same timestamp.
Here we extend \mtlpast{} EF games with the following moves
to obtain \mtlbpast{} EF games:
\begin{itemize}
\item \emph{$\first_{I}$-move}: \emph{Spoiler} chooses one of the two timed words (say $\rho$)
and picks $i_r'$ such that (i) $\tau_{i_r'} - \tau_{i_r} \in I$ in $\rho$
and (ii) there is no position $i' < i_r'$ in $\rho$ such that $\tau_{i'} - \tau_{i_r} \in I$.
\emph{Duplicator} must choose a position $j_r'$ in $\rho'$ such that $j_r'$ is the
first position in $I$ relative to $j_r$ in $\rho'$. If she cannot find such a position
then \emph{Spoiler} wins the game.
\item \emph{$\pfirst_{I}$-move}: Defined symmetrically.
\end{itemize}

\begin{thm}[\mtlbpast{} EF Theorem]\label{thm:mtlbef}
For (finite) timed words $\rho$, $\rho'$ and an \emph{\mtlbpast{}} formula $\varphi$ of modal depth $\leq m$,
if \emph{Duplicator} has a winning strategy for the $m$-round \emph{\mtlbpast{}} EF game on
$\rho$, $\rho'$ with $(i_0, j_0) = (0, 0)$, then
\[
\rho \models \varphi \iff \rho' \models \varphi \,.
\]
\end{thm}

\begin{prop}\label{prop:foone-strict-induction}
\emph{Duplicator} has a winning strategy for $m$-round \emph{\mtlbpast{}} EF game
on $\mathcal{G}_{m}$ and $\mathcal{H}_{m}$ with $(i_0, j_0) = (0, 0)$.
In particular, she has a winning strategy such that for each round  $0 \leq r \leq m$,
the $i_r^{\mathit{th}}$ event in $\mathcal{G}_{m}$ and the $j_r^{\mathit{th}}$ event in $\mathcal{H}_{m}$
satisfy the same set of monadic predicates and
\begin{itemize}
\item if $(i_r, j_r)$ is not synchronised, then
	\begin{itemize}
	\item $|i_r - j_r| = 1$
	\item $(m + 1 - r) < i_r, j_r < (m + 3 + r)$ or $(3m + 4 - r) < i_r, j_r < (3m + 5 + r)$.
	\end{itemize}
\end{itemize}
\end{prop}

\noindent
We prove the proposition by induction on $r$. The idea, again, is to try to make the resulting configurations identical.
\begin{itemize}
\item \emph{Base step.} The proposition holds trivially for $(i_0, j_0) = (0, 0)$.
\item \emph{Induction step.} Suppose that the claim holds for $r < m$. We prove it
also holds for $r + 1$.
	\begin{itemize}
	\item $(i_r, j_r) = (0, 0)$:\\
	\emph{Duplicator} tries to make $(i_r', j_r')$ synchronised.
	If \emph{Spoiler} chooses $i_r' = m + 2$, \emph{Duplicator} chooses either
	$j_r' = m + 1$ or $j_r' = m + 2$.
	\item $(i_r, j_r) \neq (0, 0)$ is synchronised:\\
	\emph{Duplicator} tries to make $(i_r', j_r')$ synchronised.
	If this is not possible then \emph{Duplicator} chooses a suitable event that minimises $|i_r' - j_r'|$.
	It is easy to see that the resulting configuration $(i_{r+1}, j_{r+1})$ satisfies the claim
	regardless of how \emph{Spoiler} plays.
	\item $(i_r, j_r)$ is not synchronised:\\
	The strategy of \emph{Duplicator} is same as the case above.
	\end{itemize}
\end{itemize}
Proposition~\ref{prop:foone-strict} now follows from Proposition~\ref{prop:foone-strict-induction},
Theorem~\ref{thm:mtlbef}, and the fact that the \foone{} formula
\[
    \arraycolsep=0.3ex
    \begin{array}{ll}
    \exists x' \, \biggl(  d(x, x') > 1 \wedge d(x, x') < 2  \wedge \exists x'' \, \Bigl( & x' < x'' \wedge \nexists y' \, (x' < y' \wedge y' < x'') \\
    & {} \wedge d(x, x'') > 1 \wedge d(x, x'') < 2 \\
    & {} \wedge \nexists y'' \, \bigl( d(x', y'') < 1 \wedge d(x'', y'') > 1 \bigr) \Bigr) \biggr)
    \end{array}
\]
distinguishes $\mathcal{G}_{m}$ and $\mathcal{H}_{m}$ for any $m \in \mathbb{N}$ (when evaluated at position $0$).
This formula asserts that there is a pair of neighbouring events in $(1, 2)$ such that
there is no event between them if they are both mapped to exactly one time unit earlier.
\end{proof}

One way to understand why \mtlbpast{} is still less expressive than \foone{}
is to consider the arity of modalities.
Let the current instant be $t_1$, and suppose that we want
to specify the following property for some positive integers $a$ and $c$ ($a > c$):\footnote{We remark that a closely related yet different property
is used in~\cite{Lasota2008} to show that one-clock alternating timed automata and
timed automata are expressively incomparable.}
\begin{itemize}
\item There is an event at $t_2 > t_1 + a$ where $Q$ holds
\item $P$ holds at all events in $\bigl(t_1 + c, t_1 + c + (t_2 - t_1 - a)\bigr)$.
\end{itemize}
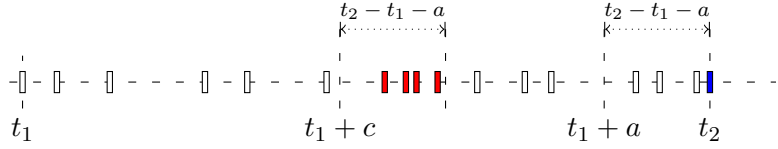
\begin{figure}[ht]
\centering
\begin{tikzpicture}[scale=1]
\begin{scope}[>=latex]

\draw[-, loosely dashed] (-5pt,0pt) -- (285pt,0pt);

\draw[-, loosely dashed] (0pt,-10pt) -- (0pt,10pt) node[at start, below] {$t_1$};

\draw[-, loosely dashed] (120pt,-10pt) -- (120pt,20pt) node[at start, below] {$t_1 + c$};
\draw[-, loosely dashed] (160pt,-10pt) -- (160pt,20pt);
\draw[-, loosely dashed] (220pt,-10pt) -- (220pt,20pt) node[at start, below] {$t_1 + a$};
\draw[-, loosely dashed] (260pt,-10pt) -- (260pt,20pt) node[at start, below] {$t_2$};

\draw[draw=black, fill=white] (1pt, -4pt) rectangle (-1pt, 4pt);
\draw[draw=black, fill=white] (14pt, -4pt) rectangle (12pt, 4pt);
\draw[draw=black, fill=white] (34pt, -4pt) rectangle (32pt, 4pt);

\draw[draw=black, fill=white] (70pt, -4pt) rectangle (68pt, 4pt);
\draw[draw=black, fill=white] (86pt, -4pt) rectangle (84pt, 4pt);
\draw[draw=black, fill=white] (114pt, -4pt) rectangle (116pt, 4pt);

\draw[draw=black, fill=red] (136pt, -4pt) rectangle (138pt, 4pt);

\draw[draw=black, fill=red] (144pt, -4pt) rectangle (146pt, 4pt);

\draw[draw=black, fill=red] (148pt, -4pt) rectangle (150pt, 4pt);

\draw[draw=black, fill=red] (156pt, -4pt) rectangle (158pt, 4pt);

\draw[draw=black, fill=white] (171pt, -4pt) rectangle (173pt, 4pt);

\draw[draw=black, fill=white] (191pt, -4pt) rectangle (189pt, 4pt);

\draw[draw=black, fill=white] (199pt, -4pt) rectangle (201pt, 4pt);

\draw[draw=black, fill=white] (231pt, -4pt) rectangle (233pt, 4pt);

\draw[draw=black, fill=white] (240pt, -4pt) rectangle (242pt, 4pt);

\draw[draw=black, fill=white] (254pt, -4pt) rectangle (256pt, 4pt);

\draw[draw=black, fill=blue] (261pt, -4pt) rectangle (259pt, 4pt);

\end{scope}

\begin{scope}[>=to]
\draw[|<->|][dotted] (120pt,20pt)  -- (160pt,20pt) node[midway,above] {{\scriptsize $t_2 - t_1 - a$}};
\draw[|<->|][dotted] (220pt,20pt)  -- (260pt,20pt) node[midway,above] {{\scriptsize $t_2 - t_1 - a$}};
\end{scope}

\end{tikzpicture}
\caption{$\varphi_{\mathit{cont2}}$ holds at $t_1$ in the continuous semantics. The red boxes denote $P$-events and the blue boxes denote $Q$-events.}%
\label{fig:puntilq2}
\end{figure}
See Figure~\ref{fig:puntilq2} for an example. In the continuous semantics, this property can
be expressed as the following simple formula over signals of the form $f^\rho$:
\[
\varphi_{\mathit{cont2}} = \bigl(\eventually_{=c} (P \vee P_{\epsilon}) \bigr) \until (\eventually_{=a} Q) \,.
\]
Observe how this formula talks about events from two (instead of just one) time points: $t_1 + c$ and $t_1 + a$.
In the same vein,
the following formula can be used to distinguish $\mathcal{G}_{m}$ and $\mathcal{H}_{m}$
(defined in the proof of Proposition~\ref{prop:foone-strict})
in the continuous semantics:
\[
\varphi_{\mathit{cont3}} = \eventually_{(1, 2)} \bigl( \neg P_\epsilon \wedge (\once_{= 1} P_\epsilon) \until (\neg P_\epsilon) \bigr) \label{for:cont3} \,.
\]
Indeed, to express such properties in the pointwise semantics, we need \emph{binary} variants of $\first_I, \pfirst_I$, which are exactly what we propose
in the next section.

\subsection{New modalities}

We define a family of modalities which can be understood as generalisations
of the usual `Until' and `Since' modalities.
Intuitively, these new modalities closely mimic
the meanings of formulas of the form $(\eventually_{={k_1}} \varphi_1) \until_{< k_3} (\eventually_{={k_2}} \varphi_2)$
or $(\eventually_{={k_1}} \varphi_1) \since_{< k_3} (\eventually_{={k_2}} \varphi_2)$
in the continuous semantics.

\paragraph{\emph{Generalised `Until' and `Since'}}

Let $I \subseteq (0, \infty)$ be an interval
with endpoints in $\mathbb{N} \cup \{\infty\}$
and $c \in \mathbb{N}$, $c \leq \inf(I)$.
The formula $\varphi_1 \guntil_{I}^c \varphi_2$, when imposed at $t_1$, asserts that
\begin{itemize}
\item There is an event at $t_2$ where $\varphi_2$ holds and $t_2 - t_1 \in I$
\item $\varphi_1$ holds at all events in the open interval $\biggl(t_1 + c, t_1 + c + \Bigl(t_2 - \bigl(t_1 + \inf(I)\bigr)\Bigr)\biggr)$.
\end{itemize}
For example, the formula $P \guntil_{(a, \infty)}^c Q$ (which is `equivalent' to $\varphi_{\mathit{cont2}}$ when the latter is interpreted over signals of the form $f^\rho$) holds at time $t_1$ in Figure~\ref{fig:puntilq2}.
Formally, for $I = (a, b) \subseteq (0, \infty)$, $a \in \mathbb{N}$, $b \in \mathbb{N} \cup \{ \infty \}$
and $c \in \mathbb{N}$ with $c \leq a$,
we define the \emph{generalised `Until'} modality $\guntil_{(a, b)}^c$ by the following \foone{} formula:
\[
    \arraycolsep=0.3ex
    \begin{array}{rclll}
    \guntil_{(a, b)}^c(x, X_1, X_2) & = & \exists x' \, \Bigl(& x < x' & {} \wedge  d(x, x') > a \wedge d(x, x') < b \wedge X_2(x') \\
    & & & {} \wedge \forall x'' \, & \bigl(x < x'' \wedge d(x, x'') > c \wedge x'' < x' \\
    & & & & \hspace{1mm} {} \wedge d(x', x'') > (a-c) \implies X_1(x'') \bigr) \Bigr) \,.
    \end{array}
\]
Symmetrically, we define the \emph{generalised `Since'} modality $\gsince_{(a, b)}^c$ as
\[
    \arraycolsep=0.3ex
    \begin{array}{rclll}
    \gsince_{(a, b)}^c(x, X_1, X_2) & = & \exists x' \, \Bigl(& x' < x & {} \wedge  d(x, x') > a \wedge d(x, x') < b \wedge X_2(x') \\
    & & & {} \wedge \forall x'' \, & \bigl(x'' < x \wedge d(x, x'') > c \wedge x' < x'' \\
    & & & & \hspace{1mm} {} \wedge d(x', x'') > (a-c) \implies X_1(x'') \bigr) \Bigr) \,.
    \end{array}
\]
We also define the modalities for $I \subseteq (0, \infty)$ being a half-open interval or a closed interval
in the expected way and refer to the logic obtained by adding these modalities to \mtlpast{}
as \mtlpastg{}. Note that the usual `Until' and `Since' modalities can be written in terms of the generalised
modalities. For instance,
\[
\varphi_1 \until_{(a, b)} \varphi_2 = \varphi_1 \guntil_{(a, b)}^{a} \varphi_2 \wedge \neg \left(\mathbf{true} \guntil_{\lopen{0, a}}^0 (\neg \varphi_1)\right) \,.
\]

\paragraph{\emph{More liberal bounds}}

In defining modalities $\guntil_{(a, b)}^c$ and $\gsince_{(a, b)}^c$
we required that $0 \leq c \leq a$. We now show that more liberal uses of bounds (constraining
intervals and superscript `$c$') are indeed syntactic sugars, and we therefore allow them in the rest of this section.
For instance, suppose that we want to
to assert the following property (which translates to $\bigl(\eventually_{=10} (\varphi_1 \vee P_{\epsilon}) \bigr) \until_{<3} (\eventually_{=2} \varphi_2)$ in the continuous semantics) at $t_1$:
\begin{itemize}
\item There is an event at $t_2$ where $\varphi_2$ holds and $t_2 - t_1 \in (2, 5)$
\item $\varphi_1$ holds at all events in $\bigl(t_1 + 10, t_1 + 10 + (t_2 - t_1 - 2)\bigr)$.
\end{itemize}
This can be expressed in \foone{} as
\[
    \arraycolsep=0.3ex
    \begin{array}{lll}
    \exists x' \, \Bigl(& x < x' & {} \wedge  d(x, x') > 2 \wedge d(x, x') < 5 \wedge X_2(x') \\
    & {} \wedge \forall x'' \, & \bigl(x < x'' \wedge d(x, x'') > 10 \wedge d(x', x'') < 8 \implies X_1(x'') \bigr) \Bigr)
    \end{array}
\]
where $X_1, X_2$ are to be substituted with $\varphi_1, \varphi_2$.
While we could define a modality
\[
    \guntil_{(2, 5)}^{10}(x, X_1, X_2)
\]
by this formula, this is not necessary as the formula is indeed equivalent to
\[
\eventually_{(2, 5)} \varphi_2 \wedge \neg \Bigl( (\neg \varphi_2) \guntil_{(10, 13)}^2 \bigl( \neg \varphi_1 \wedge \neg (\once_{= 8} \varphi_2) \bigr) \Bigr) \,.
\]
In the continuous semantics we can, of course, also refer points in the past
in such formulas, e.g.,~$(\once_{={k_1}} \varphi_1) \until_{< k_3} (\eventually_{={k_2}} \varphi_2)$.
We now generalise the idea above to handle these cases.
\begin{prop}
Let the current instant be $t_1$. The property (and its past counterpart):
\begin{itemize}
\item There is an event at $t_2$ where $\varphi_2$ holds and $t_2 - t_1 \in I$
\item $\varphi_1$ holds at all events in $\biggl(t_1 + c, t_1 + c + \Bigl(t_2 - \bigl(t_1 + \inf(I)\bigr)\Bigr)\biggr)$
\end{itemize}
where $I \subseteq (-\infty, \infty)$, $\inf(I) \in \mathbb{Z}$, $\sup(I) \in \mathbb{Z} \cup \{ \infty \}$ and $c \in \mathbb{Z}$
can be expressed with the modalities defined earlier (i.e.,~$\guntil_I^c, \gsince_I^c$ with $I \subseteq (0, \infty)$ and $c \leq \inf(I)$).
\end{prop}
\begin{proof}
Without loss of generality, we shall only focus on expressing the future version of the property for
the case of $I$ being an open interval.
To ease the presentation, we use the following convention in all the illustrations in this proof: the red boxes denote $\varphi_1$-events,
blue boxes denote $\varphi_2$-events, and white boxes denote events where neither $\varphi_1$ nor $\varphi_2$ hold.
We prove the claim in each of the following cases:
\begin{itemize}
\item \emph{$a \geq 0$ and $0 \leq c \leq a$}: This corresponds to the standard
version of $\guntil_I^c$ that we have already defined.
\item \emph{$a \geq 0$ and $c > a$}:
$\varphi_1 \guntil_{(a, b)}^c \varphi_2$ does not hold at $t_1$ if and only if one of the following holds \mbox{at $t_1$}:
	\begin{itemize}
	\item \emph{There is no $\varphi_2$-event in $(t_1 + a, t_1 + b)$}: This can be enforced by
	\[
	\neg (\eventually_{(a, b)} \varphi_2) \,.
	\]
	\item \emph{$\neg \varphi_1$ holds at an event at $t_3 \in \bigl( t_1 + c, t_1 + c + (b - a) \bigr)$ and there is no
			$\varphi_2$-event in $\lopen{t_1 + a, t_1 + a + (t_3 - t_1 - c)}$}: This can be enforced by
	\[
	(\neg \varphi_2) \guntil_{\bigl(c, c + (b - a)\bigr)}^a \bigl(\neg \varphi_1 \wedge \underbrace{\neg(\once_{=(c-a)} \varphi_2)}_{\psi} \bigr) \,.
	\]
We need the subformula $\psi$ to ensure that there is no $\varphi_2$-event at $t_1 + a + (t_3-t_1-c)$.
\begin{figure}[ht]
\centering
\begin{tikzpicture}[scale=1]
\begin{scope}[>=latex]

\draw[-, loosely dashed] (-5pt,0pt) -- (285pt,0pt);

\draw[-, loosely dashed] (0pt,-10pt) -- (0pt,10pt) node[at start, below] {$t_1$};

\draw[-, loosely dashed] (40pt,-10pt) -- (40pt,10pt) node[at start, below] {$t_1 + a$};
\draw[-, loosely dashed] (120pt,-10pt) -- (120pt,10pt) node[at start, below] {$t_1 + b$};
\draw[-, loosely dashed] (160pt,-10pt) -- (160pt,10pt) node[at start, below] {$t_1 + c$};

\draw[-, loosely dashed] (90pt,0pt) -- (90pt,20pt);
\draw[-, loosely dashed] (210pt,-10pt) -- (210pt,20pt) node[at start, below] {$t_3$};

\draw[draw=black, fill=white] (1pt, -4pt) rectangle (-1pt, 4pt);
\draw[draw=black, fill=white] (10pt, -4pt) rectangle (8pt, 4pt);
\draw[draw=black, fill=white] (30pt, -4pt) rectangle (28pt, 4pt);

\draw[draw=black, fill=white] (50pt, -4pt) rectangle (48pt, 4pt);
\draw[draw=black, fill=white] (80pt, -4pt) rectangle (78pt, 4pt);
\draw[draw=black, fill=white] (91pt, -4pt) rectangle (89pt, 4pt);

\draw[draw=black, fill=blue] (99pt, -4pt) rectangle (101pt, 4pt);

\draw[draw=black, fill=white] (130pt, -4pt) rectangle (128pt, 4pt);

\draw[draw=black, fill=white] (144pt, -4pt) rectangle (146pt, 4pt);

\draw[draw=black, fill=white] (150pt, -4pt) rectangle (148pt, 4pt);

\draw[draw=black, fill=white] (154pt, -4pt) rectangle (156pt, 4pt);

\draw[draw=black, fill=red] (165pt, -4pt) rectangle (163pt, 4pt);

\draw[draw=black, fill=red] (181pt, -4pt) rectangle (179pt, 4pt);

\draw[draw=black, fill=white] (211pt, -4pt) rectangle (209pt, 4pt);

\draw[draw=black, fill=white] (221pt, -4pt) rectangle (223pt, 4pt);

\draw[draw=black, fill=white] (240pt, -4pt) rectangle (242pt, 4pt);

\draw[draw=black, fill=white] (248pt, -4pt) rectangle (250pt, 4pt);

\draw[draw=black, fill=white] (260pt, -4pt) rectangle (258pt, 4pt);

\end{scope}

\begin{scope}[>=to]
\draw[|<->|][dotted] (40pt,20pt)  -- (90pt,20pt) node[midway,above] {{\scriptsize $t_3 - t_1 - c$}};
\draw[|<->|][dotted] (160pt,20pt)  -- (210pt,20pt) node[midway,above] {{\scriptsize $t_3 - t_1 - c$}};
\end{scope}

\end{tikzpicture}
\end{figure}
	\end{itemize}
	The desired formula is the conjunction of the negations of these two formulas.
\item \emph{$a \geq 0$ and $c < 0$}:
Let $t_2$ be the first time instant in $(t_1 + a, t_1 + b)$ where there is a $\varphi_2$-event.
Consider the following subcases:
\begin{itemize}
\item \emph{There is no event in $\bigl(t_1, t_1 + (t_2 - t_1 - a)\bigr)$}:
This can be enforced by
\[
\varphi = \mathbf{false} \, \guntil_{(a, b)}^0 \varphi_2 \,.
\]
Then we can enforce that $\varphi_1$ holds at all events in \circled{\scriptsize 1} in the illustration below
by
\[
\varphi' = (\neg \varphi_2) \guntil_{(a, b)}^{a} \bigl( \varphi_2 \wedge \underbrace{(\varphi_1 \gsince_{(a, b)}^{a + |c|} \mathbf{true})}_{\psi'} \bigr) \,.
\]
Observe that the subformula $\psi'$ must hold at $t_2$ if $\varphi_1$ holds at all events in \circled{\scriptsize 1}.
This is because, by assumption, there must be an event at $t_1$.
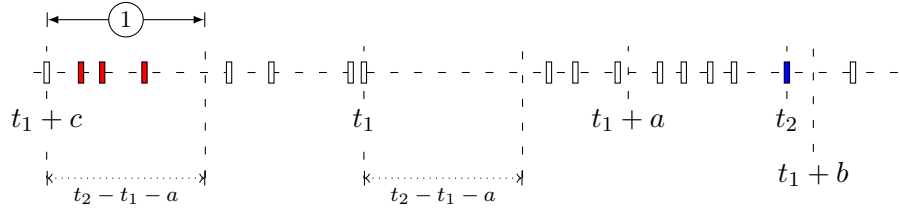
\begin{figure}[h]
\centering
\begin{tikzpicture}[scale=1]
\begin{scope}[>=latex]

\draw[-, loosely dashed] (-5pt,0pt) -- (325pt,0pt);

\draw[-, loosely dashed] (0pt,-10pt) -- (0pt,10pt) node[at start, below] {$t_1 + c$};
\draw[-, loosely dashed] (60pt,-40pt) -- (60pt,10pt);
\draw[-, loosely dashed] (0pt,-30pt) -- (0pt,-40pt);

\draw[-, loosely dashed] (120pt,-10pt) -- (120pt,10pt) node[at start, below] {$t_1$};
\draw[-, loosely dashed] (180pt,-40pt) -- (180pt,10pt);
\draw[-, loosely dashed] (120pt,-30pt) -- (120pt,-40pt);

\draw[-, loosely dashed] (220pt,-10pt) -- (220pt,10pt) node[at start, below] {$t_1 + a$};
\draw[-, loosely dashed] (280pt,-10pt) -- (280pt,10pt) node[at start, below] {$t_2$};
\draw[-, loosely dashed] (290pt,-30pt) -- (290pt,10pt) node[at start, below] {$t_1 + b$};

\draw[draw=black, fill=white] (1pt, -4pt) rectangle (-1pt, 4pt);
\draw[draw=black, fill=red] (14pt, -4pt) rectangle (12pt, 4pt);
\draw[draw=black, fill=red] (20pt, -4pt) rectangle (22pt, 4pt);

\draw[draw=black, fill=red] (36pt, -4pt) rectangle (38pt, 4pt);

\draw[draw=black, fill=white] (70pt, -4pt) rectangle (68pt, 4pt);
\draw[draw=black, fill=white] (86pt, -4pt) rectangle (84pt, 4pt);
\draw[draw=black, fill=white] (114pt, -4pt) rectangle (116pt, 4pt);

\draw[draw=black, fill=white] (119pt, -4pt) rectangle (121pt, 4pt);

\draw[draw=black, fill=white] (191pt, -4pt) rectangle (189pt, 4pt);

\draw[draw=black, fill=white] (199pt, -4pt) rectangle (201pt, 4pt);

\draw[draw=black, fill=white] (215pt, -4pt) rectangle (217pt, 4pt);

\draw[draw=black, fill=white] (231pt, -4pt) rectangle (233pt, 4pt);

\draw[draw=black, fill=white] (240pt, -4pt) rectangle (242pt, 4pt);

\draw[draw=black, fill=white] (250pt, -4pt) rectangle (252pt, 4pt);

\draw[draw=black, fill=white] (261pt, -4pt) rectangle (259pt, 4pt);

\draw[draw=black, fill=blue] (281pt, -4pt) rectangle (279pt, 4pt);

\draw[draw=black, fill=white] (304pt, -4pt) rectangle (306pt, 4pt);

\end{scope}

\begin{scope}[>=latex]
\draw[|<->|] (0pt,20pt)  -- (60pt,20pt) node[midway] {\circled{\footnotesize 1}};
\end{scope}

\begin{scope}[>=to]
\draw[|<->|][dotted] (0pt,-40pt)  -- (60pt,-40pt) node[midway,below] {{\scriptsize $t_2 - t_1 - a$}};
\draw[|<->|][dotted] (120pt,-40pt)  -- (180pt,-40pt) node[midway,below] {{\scriptsize $t_2 - t_1 - a$}};
\end{scope}

\end{tikzpicture}
\end{figure}

\item \emph{There are events in $\bigl( t_1, t_1 + (t_2 - t_1 - a) \bigr)$}:
In this case, $\varphi'$
can only ensure that $\varphi_1$ holds at all events in \circled{\footnotesize 2} (see the illustration below
where $d_1 + d_2 = t_2 - t_1 - a$).
We can enforce that $\varphi_1$ holds at all events in \circled{\footnotesize 1} by
\[
\varphi'' = \psi'' \until ( \varphi \wedge \psi''  )
\]
where
\[
\psi'' = \underbrace{(\varphi_1 \gsince_{(0, b-a)}^{|c|} \mathbf{true})}_{\psi'''} \wedge \neg (\once_{=|c|} \neg \varphi_1) \,.
\]
It is easy to see that $\varphi$ must hold at the last event in $\bigl( t_1, t_1 + (t_2 - t_1 - a) \bigr)$.
The correctness of our use of the subformula $\psi'''$ here again depends on the fact that there is an event at $t_1$.
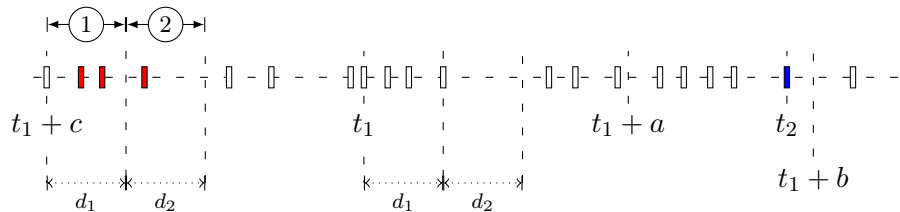
\begin{figure}[h!]
\centering
\begin{tikzpicture}[scale=1]
\begin{scope}[>=latex]

\draw[-, loosely dashed] (-5pt,0pt) -- (325pt,0pt);

\draw[-, loosely dashed] (0pt,-10pt) -- (0pt,10pt) node[at start, below] {$t_1 + c$};
\draw[-, loosely dashed] (60pt,-40pt) -- (60pt,10pt);
\draw[-, loosely dashed] (0pt,-30pt) -- (0pt,-40pt);
\draw[-, loosely dashed] (30pt,10pt) -- (30pt,-40pt);

\draw[-, loosely dashed] (120pt,-10pt) -- (120pt,10pt) node[at start, below] {$t_1$};
\draw[-, loosely dashed] (180pt,-40pt) -- (180pt,10pt);
\draw[-, loosely dashed] (120pt,-30pt) -- (120pt,-40pt);

\draw[-, loosely dashed] (150pt,10pt) -- (150pt,-40pt);

\draw[-, loosely dashed] (220pt,-10pt) -- (220pt,10pt) node[at start, below] {$t_1 + a$};
\draw[-, loosely dashed] (280pt,-10pt) -- (280pt,10pt) node[at start, below] {$t_2$};
\draw[-, loosely dashed] (290pt,-30pt) -- (290pt,10pt) node[at start, below] {$t_1 + b$};

\draw[draw=black, fill=white] (1pt, -4pt) rectangle (-1pt, 4pt);
\draw[draw=black, fill=red] (14pt, -4pt) rectangle (12pt, 4pt);
\draw[draw=black, fill=red] (20pt, -4pt) rectangle (22pt, 4pt);

\draw[draw=black, fill=red] (36pt, -4pt) rectangle (38pt, 4pt);

\draw[draw=black, fill=white] (70pt, -4pt) rectangle (68pt, 4pt);
\draw[draw=black, fill=white] (86pt, -4pt) rectangle (84pt, 4pt);
\draw[draw=black, fill=white] (114pt, -4pt) rectangle (116pt, 4pt);

\draw[draw=black, fill=white] (119pt, -4pt) rectangle (121pt, 4pt);

\draw[draw=black, fill=white] (128pt, -4pt) rectangle (130pt, 4pt);

\draw[draw=black, fill=white] (136pt, -4pt) rectangle (138pt, 4pt);

\draw[draw=black, fill=white] (149pt, -4pt) rectangle (151pt, 4pt);

\draw[draw=black, fill=white] (191pt, -4pt) rectangle (189pt, 4pt);

\draw[draw=black, fill=white] (199pt, -4pt) rectangle (201pt, 4pt);

\draw[draw=black, fill=white] (215pt, -4pt) rectangle (217pt, 4pt);

\draw[draw=black, fill=white] (231pt, -4pt) rectangle (233pt, 4pt);

\draw[draw=black, fill=white] (240pt, -4pt) rectangle (242pt, 4pt);

\draw[draw=black, fill=white] (250pt, -4pt) rectangle (252pt, 4pt);

\draw[draw=black, fill=white] (261pt, -4pt) rectangle (259pt, 4pt);

\draw[draw=black, fill=blue] (281pt, -4pt) rectangle (279pt, 4pt);

\draw[draw=black, fill=white] (304pt, -4pt) rectangle (306pt, 4pt);

\end{scope}

\begin{scope}[>=latex]
\draw[|<->|] (0pt,20pt)  -- (30pt,20pt) node[midway] {\circled{\footnotesize 1}};
\draw[|<->|] (30pt,20pt)  -- (60pt,20pt) node[midway] {\circled{\footnotesize 2}};
\end{scope}

\begin{scope}[>=to]
\draw[|<->|][dotted] (0pt,-40pt)  -- (30pt,-40pt) node[midway,below] {{\scriptsize $d_1$}};
\draw[|<->|][dotted] (30pt,-40pt)  -- (60pt,-40pt) node[midway,below] {{\scriptsize $d_2$}};
\draw[|<->|][dotted] (120pt,-40pt)  -- (150pt,-40pt) node[midway,below] {{\scriptsize $d_1$}};
\draw[|<->|][dotted] (150pt,-40pt)  -- (180pt,-40pt) node[midway,below] {{\scriptsize $d_2$}};
\end{scope}

\end{tikzpicture}
\end{figure}
\end{itemize}
The desired formula is $\varphi' \wedge (\varphi \vee \varphi'')$.

\item \emph{$a < 0$ and $c \geq 0$}:
Without loss of generality we assume $a < b < 0$. Similar to the case \emph{$a \geq 0$ and $c > a$} above, the desired formula is
\[
\once_{(|b|, |a|)} \varphi_2 \wedge \neg \Bigl( (\neg \varphi_2) \guntil_{\bigl(c, c + (b-a)\bigr)}^{a} \bigl( \neg \varphi_1 \wedge \neg (\once_{= (c-a)} \varphi_2) \bigr) \Bigr) \,.
\]

\item \emph{$a < 0$ and $c \leq a$}:
Without loss of generality we assume $a < b < 0$.
Let $t_2$ be the first time instant in $(t_1 + a, t_1 + b)$ where there is a $\varphi_2$-event.
Similar to the case $a \geq 0$ and $c < 0$ above,
consider the following subcases:
\begin{itemize}
\item \emph{There is no event in $\bigl(t_1, t_1 + (t_2 - t_1 - a)\bigr)$}:
We enforce that $\varphi_1$ holds at all events in \circled{\footnotesize 1}
in the illustration below by
\[
\varphi''' = \mathbf{false} \, \guntil_{(a, b)}^0 \bigl( \varphi_2 \wedge (\varphi_1 \gsince_{(a, b)}^{a + |c|} \mathbf{true}) \bigr)\,.
\]
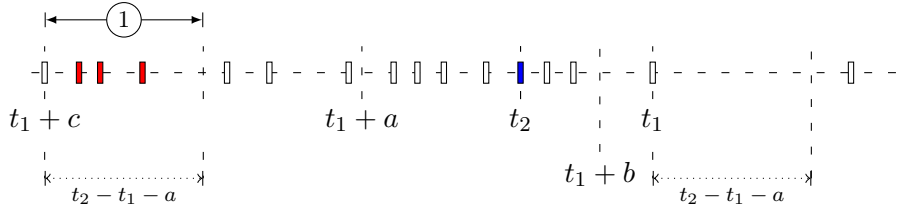
\begin{figure}[h]
\centering
\begin{tikzpicture}[scale=1]
\begin{scope}[>=latex]

\draw[-, loosely dashed] (-5pt,0pt) -- (325pt,0pt);

\draw[-, loosely dashed] (0pt,-10pt) -- (0pt,10pt) node[at start, below] {$t_1 + c$};
\draw[-, loosely dashed] (60pt,-10pt) -- (60pt,10pt);
\draw[-, loosely dashed] (0pt,-30pt) -- (0pt,-40pt);
\draw[-, loosely dashed] (60pt,-30pt) -- (60pt,-40pt);

\draw[-, loosely dashed] (120pt,-10pt) -- (120pt,10pt) node[at start, below] {$t_1 + a$};
\draw[-, loosely dashed] (180pt,-10pt) -- (180pt,10pt);

\draw[-, loosely dashed] (210pt,-30pt) -- (210pt,10pt) node[at start, below] {$t_1 + b$};

\draw[-, loosely dashed] (180pt,-10pt) -- (180pt,10pt) node[at start, below] {$t_2$};
\draw[-, loosely dashed] (230pt,-10pt) -- (230pt,10pt) node[at start, below] {$t_1$};
\draw[-, loosely dashed] (230pt,-30pt) -- (230pt,-40pt);
\draw[-, loosely dashed] (290pt,-40pt) -- (290pt,10pt);

\draw[draw=black, fill=white] (1pt, -4pt) rectangle (-1pt, 4pt);
\draw[draw=black, fill=red] (14pt, -4pt) rectangle (12pt, 4pt);
\draw[draw=black, fill=red] (20pt, -4pt) rectangle (22pt, 4pt);

\draw[draw=black, fill=red] (36pt, -4pt) rectangle (38pt, 4pt);

\draw[draw=black, fill=white] (70pt, -4pt) rectangle (68pt, 4pt);
\draw[draw=black, fill=white] (86pt, -4pt) rectangle (84pt, 4pt);
\draw[draw=black, fill=white] (114pt, -4pt) rectangle (116pt, 4pt);

\draw[draw=black, fill=white] (131pt, -4pt) rectangle (133pt, 4pt);

\draw[draw=black, fill=white] (140pt, -4pt) rectangle (142pt, 4pt);

\draw[draw=black, fill=white] (150pt, -4pt) rectangle (152pt, 4pt);

\draw[draw=black, fill=white] (166pt, -4pt) rectangle (168pt, 4pt);

\draw[draw=black, fill=blue] (181pt, -4pt) rectangle (179pt, 4pt);

\draw[draw=black, fill=white] (191pt, -4pt) rectangle (189pt, 4pt);

\draw[draw=black, fill=white] (199pt, -4pt) rectangle (201pt, 4pt);

\draw[draw=black, fill=white] (231pt, -4pt) rectangle (229pt, 4pt);

\draw[draw=black, fill=white] (304pt, -4pt) rectangle (306pt, 4pt);

\end{scope}

\begin{scope}[>=latex]
\draw[|<->|] (0pt,20pt)  -- (60pt,20pt) node[midway] {\circled{\footnotesize 1}};
\end{scope}

\begin{scope}[>=to]
\draw[|<->|][dotted] (0pt,-40pt)  -- (60pt,-40pt) node[midway,below] {{\scriptsize $t_2 - t_1 - a$}};
\draw[|<->|][dotted] (230pt,-40pt)  -- (290pt,-40pt) node[midway,below] {{\scriptsize $t_2 - t_1 - a$}};
\end{scope}

\end{tikzpicture}
\end{figure}

\item \emph{There are events in $\bigl( t_1, t_1 + (t_2 - t_1 - a) \bigr)$}:
We enforce that $\varphi_1$ holds at all events in \circled{\footnotesize 1} and \circled{\footnotesize 2}
in the illustration below (in which $d_1 + d_2 = t_2 - t_1 - a$) by
\[
\once_{(|b|, |a|)} \varphi_2 \wedge
\bigl( \psi'' \until ( \varphi''' \wedge \psi'' )  \bigr) \,,
\]
where $\psi''$ is defined in the case $a \geq 0$ and $c < 0$ above.
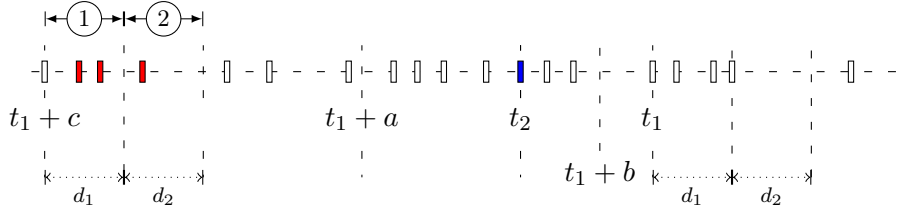
\begin{figure}[h]
\centering
\begin{tikzpicture}[scale=1]
\begin{scope}[>=latex]

\draw[-, loosely dashed] (-5pt,0pt) -- (325pt,0pt);

\draw[-, loosely dashed] (0pt,-10pt) -- (0pt,10pt) node[at start, below] {$t_1 + c$};
\draw[-, loosely dashed] (60pt,-10pt) -- (60pt,10pt);
\draw[-, loosely dashed] (0pt,-30pt) -- (0pt,-40pt);
\draw[-, loosely dashed] (60pt,-30pt) -- (60pt,-40pt);
\draw[-, loosely dashed] (30pt,10pt) -- (30pt,-40pt);

\draw[-, loosely dashed] (120pt,-10pt) -- (120pt,10pt) node[at start, below] {$t_1 + a$};
\draw[-, loosely dashed] (120pt,-30pt) -- (120pt,-40pt);
\draw[-, loosely dashed] (180pt,-10pt) -- (180pt,10pt);
\draw[-, loosely dashed] (180pt,-30pt) -- (180pt,-40pt);

\draw[-, loosely dashed] (210pt,-30pt) -- (210pt,10pt) node[at start, below] {$t_1 + b$};

\draw[-, loosely dashed] (180pt,-10pt) -- (180pt,10pt) node[at start, below] {$t_2$};
\draw[-, loosely dashed] (230pt,-10pt) -- (230pt,10pt) node[at start, below] {$t_1$};
\draw[-, loosely dashed] (230pt,-30pt) -- (230pt,-40pt);
\draw[-, loosely dashed] (290pt,-40pt) -- (290pt,10pt);
\draw[-, loosely dashed] (260pt,-40pt) -- (260pt,10pt);

\draw[draw=black, fill=white] (1pt, -4pt) rectangle (-1pt, 4pt);
\draw[draw=black, fill=red] (14pt, -4pt) rectangle (12pt, 4pt);
\draw[draw=black, fill=red] (20pt, -4pt) rectangle (22pt, 4pt);

\draw[draw=black, fill=red] (36pt, -4pt) rectangle (38pt, 4pt);

\draw[draw=black, fill=white] (70pt, -4pt) rectangle (68pt, 4pt);
\draw[draw=black, fill=white] (86pt, -4pt) rectangle (84pt, 4pt);
\draw[draw=black, fill=white] (114pt, -4pt) rectangle (116pt, 4pt);

\draw[draw=black, fill=white] (131pt, -4pt) rectangle (133pt, 4pt);

\draw[draw=black, fill=white] (140pt, -4pt) rectangle (142pt, 4pt);

\draw[draw=black, fill=white] (150pt, -4pt) rectangle (152pt, 4pt);

\draw[draw=black, fill=white] (166pt, -4pt) rectangle (168pt, 4pt);

\draw[draw=black, fill=blue] (181pt, -4pt) rectangle (179pt, 4pt);

\draw[draw=black, fill=white] (191pt, -4pt) rectangle (189pt, 4pt);

\draw[draw=black, fill=white] (199pt, -4pt) rectangle (201pt, 4pt);

\draw[draw=black, fill=white] (231pt, -4pt) rectangle (229pt, 4pt);

\draw[draw=black, fill=white] (238pt, -4pt) rectangle (240pt, 4pt);

\draw[draw=black, fill=white] (252pt, -4pt) rectangle (254pt, 4pt);

\draw[draw=black, fill=white] (259pt, -4pt) rectangle (261pt, 4pt);

\draw[draw=black, fill=white] (304pt, -4pt) rectangle (306pt, 4pt);

\end{scope}

\begin{scope}[>=latex]
\draw[|<->|] (0pt,20pt)  -- (30pt,20pt) node[midway] {\circled{\footnotesize 1}};
\draw[|<->|] (30pt,20pt)  -- (60pt,20pt) node[midway] {\circled{\footnotesize 2}};
\end{scope}

\begin{scope}[>=to]

\draw[|<->|][dotted] (0pt,-40pt)  -- (30pt,-40pt) node[midway,below] {{\scriptsize $d_1$}};
\draw[|<->|][dotted] (30pt,-40pt)  -- (60pt,-40pt) node[midway,below] {{\scriptsize $d_2$}};
\draw[|<->|][dotted] (230pt,-40pt)  -- (260pt,-40pt) node[midway,below] {{\scriptsize $d_1$}};
\draw[|<->|][dotted] (260pt,-40pt)  -- (290pt,-40pt) node[midway,below] {{\scriptsize $d_2$}};
\end{scope}

\end{tikzpicture}
\end{figure}

\end{itemize}
The desired formula is the disjunction of these two formulas.

\item \emph{$a < 0$ and $a < c < 0$}:
Without loss of generality we assume $a < b < 0$.
The desired formula is identical to the formula in the case $a < 0$ and $c \geq 0$ above. \qedhere
\end{itemize}
\end{proof}

\noindent
We can now give an \mtlpastg{} formula that distinguishes, in the pointwise semantics,
the models $\mathcal{G}_{m}$ and $\mathcal{H}_{m}$ in the proof of
Proposition~\ref{prop:foone-strict}:
\[
\eventually_{(1, 2)} \bigl( \mathbf{true} \wedge (\mathbf{false} \guntil_{> 0}^{-1} \mathbf{true}) \bigr) \,.
\]
This formula is `equivalent' to the formula $\varphi_{\mathit{cont3}}$
which distinguishes $\mathcal{G}_{m}$ and $\mathcal{H}_{m}$ in the continuous semantics.

\subsection{The translation}\label{subsec:translation}

We now give a translation from an arbitrary \foone{} formula with a single free variable into an equivalent \mtlpastg{} formula over $\ropen{0, N}$-timed words.
Our proof strategy closely follows~\cite{Ouaknine2009}:
first convert the formula into a non-metric formula, then
translate this formula into \ltlpast{}, and finally construct
an \mtlpastg{} formula equivalent to the original formula.
The crux of the translation is a \emph{`stacking' bijection} between $\ropen{0, N}$-timed words over $\Sigma_\mathbf{P}$
and a set of $\ropen{0, 1}$-timed words over an extended alphabet.
Roughly speaking, since the time domain is bounded, we can encode the integer parts of timestamps
with a bounded number of new monadic predicates.
This enables us to work instead with `stacked' $\ropen{0, 1}$-timed words, in which only the ordering of events are relevant.

\paragraph{\emph{Stacking bounded timed words}}
For each monadic predicate $P \in \mathbf{P}$, we introduce
fresh monadic predicates $P_i$, $0 \leq i \leq N - 1$ and let
the set of all these new monadic predicates be $\overline{\mathbf{P}}$.
The intended meaning is that for $x \in \ropen{0, 1}$, $P_i(x)$ holds in a stacked $\ropen{0, 1}$-timed word iff
$P$ holds at time $i + x$ in the corresponding $\ropen{0, N}$-timed word.
We also introduce $\overline{\mathbf{Q}} = \{Q_i \mid 0 \leq i \leq N - 1\}$
such that for $x \in \ropen{0, 1}$, $Q_i(x)$ holds in a stacked $\ropen{0, 1}$-timed word
iff there is an event at time $i + x$ in the corresponding $\ropen{0, N}$-timed word,
regardless of whether any $P \in \mathbf{P}$ holds there.
Let
\[
    \vartheta_{\mathit{event}} = \forall x \, \left(\bigvee_{0 \leq i \leq N - 1} Q_i(x)\right) \wedge \forall x \, \left(\bigwedge_{0 \leq i \leq N - 1} \bigl( P_i(x) \implies Q_i(x) \bigr) \right)
\]
and $\vartheta_{\mathit{init}} = \exists x \, \bigl( \nexists x' \, (x' < x) \wedge Q_0(x) \bigr)$.\label{for:wellformed}
There is an obvious `stacking' bijection (indicated by overlining) between $\ropen{0, N}$-timed words over $\Sigma_\mathbf{P}$ and
$\ropen{0, 1}$-timed words over $\Sigma_{\overline{\mathbf{P}} \cup \overline{\mathbf{Q}}}$
satisfying $\vartheta_{\mathit{event}} \wedge \vartheta_{\mathit{init}}$. For a concrete example, the stacked counterpart
of the $\ropen{0, 2}$-timed word
\[
\rho = (\{A\}, 0)(\{A, C\}, 0.3)(\{B\}, 1)(\{B, C\}, 1.5)
\]
with $\mathbf{P} = \{A, B, C\}$ is the $\ropen{0, 1}$-timed word:
\[
\overline{\rho} = (\{Q_0, Q_1, A_0, B_1\}, 0)(\{Q_0, A_0, C_0\}, 0.3)(\{Q_1, B_1, C_1\}, 0.5) \,.
\]

\paragraph{\emph{Stacking \foone{} formulas}}
Let $\vartheta(x)$ be an \foone{} formula with a single free variable $x$ where each
quantifier uses a fresh new variable.
Without loss of generality, we assume that $\vartheta(x)$ contains only existential quantifiers (this can be achieved by syntactic rewriting).
Replace the formula by
\[
\bigl(Q_0(x) \wedge \vartheta[x/x]\bigr) \vee \bigl(Q_1(x) \wedge \vartheta[x+1/x]\bigr) \vee \ldots \vee \bigl(Q_{N-1}(x) \wedge \vartheta[x + (N - 1)/x] \bigr)
\]
where $\vartheta[e/x]$ denotes the formula obtained by substituting all free occurrences of $x$
in $\vartheta$ by (an expression) $e$.
Then, similarly, recursively replace every subformula $\exists x' \, \theta$ by
\[
\exists x' \, \Bigl( \bigl( Q_0(x') \wedge \theta[x'/x'] \bigr) \vee \ldots \vee \bigl( Q_{N-1}(x') \wedge \theta[x' + (N - 1)/x'] \bigr) \Bigr) \,.
\]
Note that we do not actually have the $+k$ functions in our pointwise version of \foone{};
they are only used as annotations here and will be removed later,
e.g., $x' + k$ means that $Q_k(x')$ holds. We then carry out the following syntactic substitutions:
\begin{itemize}
\item For each inequality of the form $x_1 + k_1 < x_2 + k_2$, replace it with
	\begin{itemize}
	\item $x_1 < x_2$ if $k_1 = k_2$
	\item $\mathbf{true}$ if $k_1 < k_2$
	\item $\mathbf{false}$ if $k_1 > k_2$
	\end{itemize}
\item For each distance formula, e.g., $d(x_1 + k_1, x_2 + k_2) < 2$, replace it with
	\begin{itemize}
	\item $\mathbf{true}$ if $|k_1 - k_2| \leq 1$
	\item $x_2 < x_1$ if $k_2 - k_1 = 2$
	\item $x_1 < x_2$ if $k_1 - k_2 = 2$
	\item $\mathbf{false}$ if $|k_1 - k_2| > 2$
	\end{itemize}
\item Replace terms of the form $P(x_1 + k)$ with $P_k(x_1)$.
\end{itemize}

\noindent
This gives a non-metric first-order formula $\overline{\vartheta}(x)$ over $\overline{\mathbf{P}} \cup \overline{\mathbf{Q}}$.
Denote by $\mathit{frac}(t)$ the fractional part of a non-negative real $t$.
It is not hard to see that for each $\ropen{0, N}$-timed word $\rho = (\sigma, \tau)$ over $\Sigma_{\mathbf{P}}$
and its stacked counterpart $\overline{\rho}$, the following holds:
\begin{itemize}
\item $\rho, t \models \vartheta(x)$ implies $\overline{\rho}, \overline{t} \models \overline{\vartheta}(x)$ where $\overline{t} = \mathit{frac}(t)$
\item $\overline{\rho}, \overline{t} \models \overline{\vartheta}(x)$ implies there
	exists $t \in \rho$ with $\mathit{frac}(t) = \overline{t}$ such that $\rho, t \models \vartheta(x)$.
\end{itemize}
Moreover,
if $\rho, t \models \vartheta(x)$, then the integer part of $t$
indicates which disjunct in $\overline{\vartheta}(x)$ is satisfied
when $x$ is substituted with $\overline{t} = \mathit{frac}(t)$, and vice versa.
By Kamp's theorem~\cite{Kamp1968} (applied individually on each $\vartheta[x + i/x]$), $\overline{\vartheta}(x)$ is equivalent to an \ltlpast{} formula
$\overline{\varphi}$ of the following form:
\[
(Q_0 \wedge \overline{\varphi}_0) \vee (Q_1 \wedge \overline{\varphi}_1) \vee \ldots \vee (Q_{N-1} \wedge \overline{\varphi}_{N-1}) \,.
\]

\paragraph{\emph{Unstacking}}

We construct inductively an \mtlpastg{} formula $\psi$
for each subformula $\overline{\psi}$ of $\overline{\varphi}_i$ (for some $i \in \{0, \ldots, N - 1\}$).
Again, we make use of the formulas in $\Phi_{\mathit{int}}$ defined earlier.
\begin{itemize}
\item $\overline{\psi} = P_j$: Let
\[
\psi = (\varphi_{0, 1} \wedge \eventually_{=j} P) \vee \ldots \vee (\varphi_{j, j+1} \wedge P)
\vee \ldots \vee (\varphi_{N-1, N} \wedge \once_{=\left( (N-1) - j \right)} P) \,.
\]
\item $\overline{\psi} = Q_j$: Similarly, let
\[
\psi = (\varphi_{0, 1} \wedge \eventually_{=j} \mathbf{true}) \vee \ldots \vee (\varphi_{j, j+1} \wedge \mathbf{true})
\vee \ldots \vee (\varphi_{N-1, N} \wedge \once_{=\left( (N-1) - j \right)} \mathbf{true}) \,.
\]
\item $\overline{\psi} = \overline{\psi}_1 \until \overline{\psi}_2$:
Let $\psi^{j, k, l} = \psi_1 \guntil_{(j, j + 1)}^k (\psi_2 \wedge \varphi_{l, l + 1})$.
The desired formula is
\[
\psi = \bigvee_{0 \leq i \leq N - 1} \left(\varphi_{i, i + 1} \wedge
	\bigvee_{\substack{-i \leq j \leq (N-1) - i \\ l = i + j}} \left(
		\bigwedge_{-i \leq k \leq (N-1) - i} \psi^{j, k, l}
	\right)
\right) \,.
\]
\item $\overline{\psi} = \overline{\psi}_1 \since \overline{\psi}_2$: This is symmetric to the case of $\overline{\psi}_1 \until \overline{\psi}_2$.
\end{itemize}
The construction for the other cases are trivial and therefore omitted.
\begin{prop}\label{prop:translation}
Let $\overline{\psi}$ be a subformula of $\overline{\varphi}_i$ for some $i \in \{0, \ldots, N - 1\}$.
There is an \emph{\mtlpastg{}} formula $\psi$ such that for any $\ropen{0, N}$-timed word $\rho$,
$t \in \rho$ and $\mathit{frac}(t) = \overline{t} \in \overline{\rho}$, we have
\[
\overline{\rho}, \overline{t} \models \overline{\psi} \iff \rho, t \models \psi \,.
\]
\end{prop}
\begin{proof}
Induction on the structure of $\overline{\psi}$ and $\psi$, where the latter is constructed as described above.
\begin{itemize}
\item $\overline{\psi} = P_j$: Assume $\overline{\rho}, \overline{t} \models \overline{\psi}$.
If $t = j + \overline{t}$, the disjunct $(\varphi_{j, j+1} \wedge P)$ of $\psi$
clearly holds at $t$ in $\rho$. If $t = j' + \overline{t}$ where $j' \neq j$, since there is a $P$-event
at time $j + \overline{t}$ in $\rho$, the $j'$-th disjunct of $\psi$ must hold at $t$ in $\rho$.
The proof for the other direction is similar.

\item $\overline{\psi} = \overline{\psi}_1 \until \overline{\psi}_2$:
Assume $\overline{\rho}, \overline{t} \models \overline{\psi}$
and let the `witness' (i.e.,~where $\overline{\psi}_2$ holds) be at $\overline{t'}$. By construction and the induction hypothesis, there is an event at $t' = l + \overline{t}$ in $\rho$ for some $l \in \{0, \ldots, N - 1\}$
such that $\rho, t' \models \psi_2$. Moreover, since we have $\overline{\rho}, \overline{t''} \models \overline{\psi_1}$
for all $\overline{t''}$, $\overline{t} < \overline{t''} < \overline{t'}$, we must have
$\rho, t'' \models \psi_1$
for all $t'' \in \rho$ with
$t'' = k' + \overline{t''}$ for some $\overline{t} < \overline{t''} < \overline{t'}$
and $0 \leq k' \leq N - 1$.
Now let $t = i + \overline{t}$ for some $i \in \{0, \ldots, N - 1\}$ be a timestamp in $\rho$
and let $j = l - i$. It is clear that $\rho, t \models \varphi_{i, i + 1}$
and
\[
    \rho, t \models \bigwedge_{\substack{0 \leq k' \leq N-1 \\ k = k' - i}} \psi^{j, k, l} \,,
\]
as required. For the other direction, let $t = i + \overline{t}$ for
some $i \in \{0, \ldots, N - 1\}$ and let
\[
    \rho, t \models \bigwedge_{-i \leq k \leq (N-1) - i} \psi^{j, k, l}
\]
for some $j \in \{-i, \ldots, (N - 1) - i\}$ and $l = i + j$. It follows that there is a
(minimal) $\overline{t'} > \overline{t}$ such that $\rho, l + \overline{t'} \models \psi_2$
and $\rho, k' + \overline{t''} \models \psi_1$
for all $t'' \in \rho$ with
$t'' = k' + \overline{t''}$ for some $\overline{t} < \overline{t''} < \overline{t'}$
and $0 \leq k' \leq N - 1$. The claim follows by construction and the induction hypothesis.

\end{itemize}
The other cases are trivial or symmetric.
\end{proof}

Using the construction above, we obtain an \mtlpastg{} formula $\varphi_i$ for each $\overline{\varphi}_i$.
Substitute them into $\overline{\varphi}$ and replace all remaining $Q_i$ by $\varphi_{i, i + 1}$
to obtain our final formula $\varphi$,
which is equivalent to the original \foone{} formula $\vartheta(x)$ over $\ropen{0, N}$-timed words.

\begin{prop}
For any $\ropen{0, N}$-timed word $\rho$ and $t \in \rho$, we have
\[
\rho, t \models \varphi(x) \iff \rho, t \models \vartheta(x) \,.
\]
\end{prop}

We are now ready to state the main result of this section.

\begin{thm}\label{thm:boundedexpcomp}
\emph{\mtlpastg{}} is expressively complete for \emph{\foone{}} over $\ropen{0, N}$-timed words.
\end{thm}

\subsection{Time-bounded verification}

We claim that the \emph{timed-bounded satisfiability} and \emph{time-bounded model-checking}
problems for \mtlpastg{} are $\mathrm{EXPSPACE}$-complete in both the pointwise and continuous semantics.

\begin{thm}
The time-bounded satisfiability problem for \emph{\mtlpastg{}} (in both the pointwise and continuous semantics)
is $\mathrm{EXPSPACE}$-complete.
\end{thm}
\begin{proof}
First note that for each \mtlpastg{} formula over timed words one can construct,
in linear time, an `equivalent' \mtlpastg{} formula over signals of the form $f^\rho$.
Then, in the continuous semantics, one can replace all subformulas of the form $\varphi_1 \guntil_{(a, b)}^c \varphi_2$
in an \mtlpastg{} formula by
\[
(\eventually_{= c} \varphi_1) \until_{< b} (\eventually_{= a} \varphi_2)
\]
(this can incur at most a linear blow-up). The claim therefore follows from~\cite{Ouaknine2009}.
However, for the sake of completeness, we give a direct proof (for the case of pointwise semantics)
along the lines of~\cite{Ouaknine2009} here; see Section~\ref{sec:conclusion} for a discussion on the practical implication.

For each subformula $\psi$ of a given \mtlpastg{} formula $\varphi$ and every $i \in \{0, \ldots, N\}$,
we introduce a monadic predicate $F^{\psi}_i$. We then add suitable
subformulas into $\overline{\varphi}$ to ensure that $F^{\psi}_i$ holds
at $\overline{t}$ in $\overline{\rho}$ iff $\psi$ holds at $t = \overline{t} + i$ in $\rho$.
As an example, let $A \guntil_{(2, 3)}^1 B$ be a subformula of $\varphi$.
We require the following formula to hold at every point in time
(assume that $i \leq N - 4$):
\[
    \arraycolsep=0.3ex
    \begin{array}{rcl}
    F_i^{A \guntil_{(2, 3)}^1 B} & \iff & \bigl( (F_{i+1}^Q \implies F_{i+1}^A) \until F_{i+2}^B\bigr) \\
    & & \vee \Bigl( \globally (F_{i+1}^Q \implies F_{i+1}^A) \wedge \once \bigl(F_{i+3}^B \wedge \henceforth (F_{i+2}^Q \implies F_{i+2}^A)\bigr) \Bigr)\,.
    \end{array}
\]
We also add the \ltl{} equivalents of $\vartheta_{\mathit{event}}$ and $\vartheta_{\mathit{init}}$
into $\overline{\varphi}$ as conjuncts.
It is clear that $\overline{\varphi}$ is of size exponential in the size of $\varphi$.
$\mathrm{EXPSPACE}$-hardness follows from the corresponding result of
\textmd{\textsf{Bounded-MTL}} (in the pointwise semantics) in~\cite{Bouyer2007}.
\end{proof}

Since the time-bounded model-checking problem and satisfiability problem
are inter-reducible in both the pointwise and continuous semantics~\cite{Wilke1994, Henzinger1998},
we have the following theorem.

\begin{thm}
The time-bounded model-checking problem for timed automata against \emph{\mtlpastg{}} (in both the pointwise
and continuous semantics) is $\mathrm{EXPSPACE}$-complete.
\end{thm}

\section{Expressive completeness of \mtlpastg{} over unbounded timed words}

Recall that the counting modality $C_2(x, X)$ asserts that $X$ holds at at least two points
in $(x, x+1)$. While the modality is not expressible in \mtlpast{}, it is equivalent to
the following \mtlpast{} formula with \emph{rational} constants:
\[
\eventually_{(0, \frac{1}{2})} (X \wedge \eventually_{(0, \frac{1}{2})} X)
\vee \eventually_{(\frac{1}{2}, 1)} (X \wedge \once_{(0, \frac{1}{2})} X)
\vee ( \eventually_{(0, \frac{1}{2})} X \wedge \eventually_{(\frac{1}{2}, 1)} X ) \,.
\]
Indeed, \mtlpast{} with rational constants
is expressively complete for \fo{} (the rational version of \foone{}) over signals~\cite{Hunter2012}.
Unfortunately, even with rational endpoints, \mtlpast{} is still less expressive than \foone{} in the pointwise semantics~\cite{Prabhakar2006}.
We show in this section that expressive completeness of \mtlpast{} over (infinite) timed words can be recovered by adding (the rational versions of) the modalities
generalised `Until' ($\guntil_I^c$) and generalised `Since' ($\gsince_I^c$)
we introduced in the last section.

Our presentation in this section essentially follows~\cite{Hunter2012}. We first give a set of rewriting
rules that `extract' unbounded temporal operators from the scopes of bounded temporal operators.
Then we invoke Gabbay's separation theorem~\cite{Gabbay1980} to obtain a syntactic separation result for \mtlpastg{}
in the pointwise semantics. Exploiting a normal form for \foone{} in~\cite{Gabbay1980},
we show that any bounded \fo{} formula can be rewritten into an \mtlpastg{} formula.
Finally, these ideas are combined to obtain the desired result.

\subsection{Syntactic separation of \mtlpastg{} formulas}\label{sec:separation}

We present a series of logical equivalence rules that can
be used to rewrite a given \mtlpastg{} formula into an equivalent
formula in which no unbounded temporal operators occurs within the
scope of a bounded temporal operator.  Only the rules for
open intervals are given, as the rules for other types of intervals are
straightforward variants.

\paragraph{\emph{A normal form for \mtlpastg{}}}
We say an \mtlpastg{} formula is in \emph{normal form} if it satisfies:
\begin{enumerate}[label=(\roman*).] % chktex 36
\item All occurrences of unbounded temporal operator are of the form $\until_{(0, \infty)}$, $\since_{(0, \infty)}$, $\globally_{(0, \infty)}$, $\henceforth_{(0, \infty)}$.
\item All other occurrences of temporal operators are of the form
		$\until_I$, $\since_I$, $\guntil^c_I$, $\gsince^c_I$ with \mbox{bounded $I$}.
\item Negation is only applied to monadic predicates or bounded temporal operators.
\item In any subformula of the form $\varphi_1 \until_I \varphi_2$, $\varphi_1 \since_I \varphi_2$,
$\eventually_I \varphi_2$, $\once_I \varphi_2$, $\varphi_1 \guntil_I^c \varphi_2$, $\varphi_1 \gsince_I^c \varphi_2$ where $I$ is bounded,
		$\varphi_1$ is a disjunction of subformulas and $\varphi_2$ is a conjunction of subformulas.
\end{enumerate}
We now describe how to rewrite a given formula into normal form.
To satisfy (i) and (ii), apply the usual rules (e.g.,~$\globally_I \varphi \iff \neg \eventually_I \neg \varphi$) and the rules:
\[
    \arraycolsep=0.3ex
    \begin{array}{rcl}
    \varphi_1 \until_{(a, \infty)} \varphi_2  & \iff &  \varphi_1 \until \varphi_2 \wedge
    \globally_{\lopen{0, a}} (\varphi_1 \wedge \varphi_1 \until \varphi_2) \\
    \varphi_1 \guntil^c_{(a, \infty)} \varphi_2 & \iff & \varphi_1 \guntil^c_{\lopen{a, 2a}} \varphi_2 \vee \Big( \eventually^w_{[0, c]} \big(\varphi_1 \until_{(a, \infty)} (\varphi_2 \vee \eventually_{\leq a - c} \varphi_2) \big) \Big) \,.
    \end{array}
\]
To satisfy (iii), use the usual rules and the rule:
\[
    \arraycolsep=0.3ex
    \begin{array}{rcl}
    \neg (\varphi_1 \until \varphi_2) & \iff & \globally \neg \varphi_2 \vee \big( \neg \varphi_2 \until (\neg \varphi_2 \wedge \neg \varphi_1) \big) \,.
    \end{array}
\]
For (iv), use the usual rules of Boolean algebra and the rules below:
\[
    \arraycolsep=0.3ex
    \begin{array}{rcl}
    \phi \until_I (\varphi_1 \vee \varphi_2) & \iff & (\phi \until_I \varphi_1) \vee (\phi \until_I \varphi_2) \\
    (\varphi_1 \wedge \varphi_2) \until_I \phi & \iff & (\varphi_1 \until_I \phi) \wedge (\varphi_2 \until_I \phi) \\
    \phi \guntil_I^c (\varphi_1 \vee \varphi_2) & \iff & (\phi \guntil_I^c \varphi_1) \vee (\phi \guntil_I^c \varphi_2) \\
    (\varphi_1 \wedge \varphi_2) \guntil_I^c \phi & \iff & (\varphi_1 \guntil_I^c \phi) \wedge (\varphi_2 \guntil_I^c \phi) \,.
    \end{array}
\]
The rules for past temporal operators are as symmetric.
We prove one of these rules as the others are simpler.
\begin{prop}\label{prop:removeunboundedguntil}
The following equivalence holds over infinite timed words:
\[
\varphi_1 \guntil^c_{(a, \infty)} \varphi_2 \iff \varphi_1 \guntil^c_{\lopen{a, 2a}} \varphi_2 \vee \Big( \eventually^w_{[0, c]} \big(\varphi_1 \until_{(a, \infty)} (\varphi_2 \vee \eventually_{\leq a - c} \varphi_2) \big) \Big) \,.
\]
\end{prop}
\begin{proof}
Let the current position be $i$ and the witness be at position $w$. Consider the following cases:
\begin{itemize}
\item $\tau_w \in \lopen{\tau_i + a, \tau_i + 2a}$: $\varphi_1 \guntil^c_{\lopen{a, 2a}} \varphi_2$ clearly holds.
\item $\tau_w \in (\tau_i + 2a, \infty)$: Consider the following subcases:
	\begin{itemize}
	\item $\varphi_1$ holds at all positions $j < w$ such that $\tau_j > \tau_i + c$:
		$\varphi_1 \until_{(a, \infty)} \varphi_2$ holds at the maximal position $j'$ such that $\tau_{j'} \in [\tau_i, \tau_i + c]$.
	\item $\varphi_1$ holds at all positions $j < w$ such that $\tau_j > \tau_i + c$ and $\tau_w - \tau_j > a - c$:
		By assumption, there is a position $j'$ at which $\varphi_1$ does not hold and $\tau_w - \tau_{j'} \leq a - c$.
		Since $\tau_w > \tau_i + 2a$, we have $\tau_{j'} > \tau_i + a + c$.
		It follows that $\varphi_1 \until_{(a, \infty)} (\eventually_{\leq a - c} \varphi_2)$ holds at the maximal position in $[\tau_i, \tau_i + c]$.
	\end{itemize}
\end{itemize}
The other direction is obvious.
\end{proof}

\paragraph{\emph{Extracting unbounded operators from bounded operators}}
We now provide a set of rewriting rules that extract unbounded temporal operators from the scopes of bounded temporal operators.
In what follows, let $\varphi_\textit{xlb} = \mathbf{false} \until_{(0, b)} \mathbf{true}$,
$\varphi_\textit{ylb} = \mathbf{false} \since_{(0, b)} \mathbf{true}$ and
\[
    \arraycolsep=0.3ex
    \begin{array}{rcl}
    \varphi_\textit{ugb} & = & \Big( (\varphi_\textit{xlb} \implies \globally_{(b, 2b)} \varphi_1) \wedge \big( \neg \varphi_\textit{ylb} \implies ( \varphi_1 \wedge \globally_{\lopen{0, b}} \varphi_1 ) \big) \Big) \\[0.5em]
    & & \until \Bigg( \big(\varphi_1 \wedge (\varphi_1 \until_{(b, 2b)} \varphi_2)\big) \vee \bigg( \neg \varphi_\textit{ylb} \wedge \Big( \varphi_2 \vee \big(\varphi_1 \wedge (\varphi_1 \until_{\lopen{0, b}} \varphi_2) \big) \Big) \bigg) \Bigg) \,, \\[1.5em]
    \varphi_\textit{ggb} & = & \globally \Big( (\varphi_\textit{xlb} \implies \globally_{(b, 2b)} \varphi_1) \wedge \big( \neg \varphi_\textit{ylb} \implies ( \varphi_1 \wedge \globally_{\lopen{0, b}} \varphi_1 ) \big) \Big) \,.
    \end{array}
\]
The intended meanings of the formulas $\varphi_\textit{ugb}$ and $\varphi_\textit{ggb}$ are similar (yet not identical) to $\varphi_1 \guntil^b_{>b} \varphi_2$ and
$\neg \big(\mathbf{true} \guntil^b_{>b} (\neg \varphi_1) \big)$, respectively.
Indeed, the equivalences in the following proposition
still hold if we replace all occurrences of $\varphi_\textit{ugb}$ and $\varphi_\textit{ggb}$
by these simpler formulas.
We, however, have to use these complicated formulas here as we aim to pull the
unbounded `Until' operator to the outermost level.
The subformulas $\globally_{(b, 2b)} \varphi_1$ and $\globally_{\lopen{0, b}} \varphi_1$
assert that $\varphi_1$ holds continuously in short `strips', and
we use the subformulas $\varphi_\textit{xlb}$ and $\varphi_\textit{ylb}$ to ensure that each event
before the point where $\varphi_2$ holds is covered by such a strip.

\begin{prop}\label{prop:extract}
The following equivalences hold over infinite timed words:
\small
\[
    \arraycolsep=0.3ex
    \begin{array}{rcl}
    \theta \until_{(a, b)} \big( (\varphi_1 \until \varphi_2) \wedge \chi \big) & \iff & \theta \until_{(a, b)} \big( (\varphi_1 \until_{(0, 2b)} \varphi_2) \wedge \chi \big) \vee \Big( \big( \theta \until_{(a, b)} (\globally_{(0, 2b)} \varphi_1 \wedge \chi) \big) \wedge \varphi_\textit{ugb} \Big) \\[0.5em]
    \theta \until_{(a, b)} (\globally \varphi \wedge \chi) & \iff & \big( \theta \until_{(a, b)} (\globally_{(0, 2b)} \varphi \wedge \chi) \big) \wedge \varphi_\textit{ggb} \\[0.5em]
    \theta \until_{(a, b)} \big( (\varphi_1 \since \varphi_2) \wedge \chi \big) & \iff & \theta \until_{(a, b)} \big( (\varphi_1 \since_{(0, b)} \varphi_2) \wedge \chi \big) \vee \Big( \big( \theta \until_{(a, b)} (\henceforth_{(0, b)} \varphi_1 \wedge \chi) \big) \wedge \varphi_1 \since \varphi_2 \Big) \\[0.5em]
    \theta \until_{(a, b)} (\henceforth \varphi \wedge \chi) & \iff & \big( \theta \until_{(a, b)} (\henceforth_{(0, b)} \varphi \wedge \chi) \big) \wedge \henceforth \varphi \\[0.5em]
    \big( (\varphi_1 \until \varphi_2) \vee \chi \big) \until_{(a, b)} \theta & \iff & \big( (\varphi_1 \until_{(0, 2b)} \varphi_2) \vee \chi \big) \until_{(a, b)} \theta \\
    & & {} \vee \bigg( \Big( \big( (\varphi_1 \until_{(0, 2b)} \varphi_2) \vee \chi \big) \until_{(0, b)} (\globally_{(0, 2b)} \varphi_1) \Big) \wedge \eventually_{(a, b)} \theta \wedge \varphi_\textit{ugb} \bigg) \\[0.5em]
    \big( (\globally \varphi) \vee \chi \big) \until_{(a, b)} \theta & \iff & \chi \until_{(a, b)} \theta \vee \big( \chi \until_{(0, b)} (\globally_{(0, 2b)} \varphi_1) \wedge \eventually_{(a, b)} \theta \wedge \varphi_\textit{ggb} \big) \\[0.5em]
    \big( (\varphi_1 \since \varphi_2) \vee \chi \big) \until_{(a, b)} \theta & \iff & \big( (\varphi_1 \since_{(0, b)} \varphi_2) \vee \chi \big) \until_{(a, b)} \theta \\ & & {} \vee \bigg( \Big( \big(\henceforth_{(0, b)} \varphi_1 \vee (\varphi_1 \since_{(0, b)} \varphi_2) \vee \chi \big) \until_{(a, b)} \theta \Big) \wedge \varphi_1 \since \varphi_2 \bigg) \\[0.5em]
    \big( (\henceforth \varphi) \vee \chi \big) \until_{(a, b)} \theta & \iff & \chi \until_{(a, b)} \theta \vee \Big( \big( (\henceforth_{(0, b)} \varphi \vee \chi) \until_{(a, b)} \theta \big) \wedge \henceforth \varphi \Big) \,.
    \end{array}
\]
\end{prop}
\begin{proof}
We sketch the proof for the first rule. In what follows, let the current position be $i$.

For the forward direction, let the witness be at position $w$. If $\tau_w < \tau_j + 2b$ for some $j$ such that $\tau_j \in (\tau_i + a, \tau_i + b)$,
the subformula $\varphi_1 \until_{(0, 2b)} \varphi_2$ clearly holds at $j$ and we are done.
Otherwise, let $j$ be the maximal position such that $\tau_j \in (\tau_i + a, \tau_i + b)$.
We know that $\globally_{(0, 2b)} \varphi_1$ must hold at $j$, so
$(\varphi_\textit{xlb} \implies \globally_{(b, 2b)} \varphi_1)$,
$\varphi_\textit{ylb}$, and hence $\big( \neg \varphi_\textit{ylb} \implies ( \varphi_1 \wedge \globally_{\lopen{0, b}} \varphi_1 ) \big)$
must hold at all positions $j'$, $i < j' < j$.
Let $l > j$ be the minimal position such that $\tau_w \in (\tau_l + b, \tau_l + 2b)$.
Consider the following cases:
\begin{itemize}
\item There exists such $l$:
It is clear that $\big(\varphi_1 \wedge (\varphi_1 \until_{(b, 2b)} \varphi_2)\big)$ holds at $l$.
Since $\globally_{(b, 2b)} \varphi_1$ holds at all positions $j''$, $j \leq j'' < l$ by
the minimality of $l$, $(\varphi_\textit{xlb} \implies \globally_{(b, 2b)} \varphi_1)$ also holds at these positions.
For the other conjunct, note that $\varphi_\textit{ylb}$ holds at $j$ and $\varphi_1 \wedge \globally_{\lopen{0, b}} \varphi_1$ holds at all positions $j'''$, $j < j''' < l$.
\item There is no such $l$:
Consider the following cases:
	\begin{itemize}
	\item $\neg \varphi_\textit{ylb}$ and $\neg \once_{= b} \mathbf{true}$ hold at $w$:
	By assumption, there is no event in $(\tau_w - 2b, \tau_w)$. The proof is similar to the case where $l$ exists.
	\item $\neg \varphi_\textit{ylb}$ and $\once_{= b} \mathbf{true}$ hold at $w$: Let $l'$ be
	the position such that $\tau_{l'} = \tau_w - b$. By assumption, there is no event in $(\tau_{l'} - b, \tau_{l'})$.
	It follows that $\neg \varphi_\textit{ylb}$ and $\big(\varphi_1 \wedge (\varphi_1 \until_{\lopen{0, b}} \varphi_2) \big)$
	hold at $l'$. The proof is similar.
	\item $\varphi_\textit{ylb}$ holds at $w$:
	By assumption, there is no event in $(\tau_w - 2b, \tau_w - b)$. It is easy to see that there
	is a position such that $\neg \varphi_\textit{ylb} \wedge \big(\varphi_1 \wedge (\varphi_1 \until_{\lopen{0, b}} \varphi_2) \big)$
	holds. The proof is again similar.
	\end{itemize}
\end{itemize}

\noindent
We prove the other direction by contraposition. Consider the interesting case where $\globally_{(0, 2b)} \varphi_1$ holds at the maximal position $j$
such that $j \in (\tau_i + a, \tau_i + b)$, yet $\varphi_1 \until \varphi_2$ does not hold at $j$.
By assumption, there is no $\varphi_2$-event in $(\tau_j, \tau_j + 2b)$.
If $\varphi_2$ never holds in $\ropen{\tau_j + 2b, \infty}$
then we are done. Otherwise, let $l > j$ be the minimal position such that both $\varphi_1$ and $\varphi_2$ do not hold at $l$
(note that $\tau_l \geq \tau_j + 2b$).
It is clear that
\[
    \Bigg( \big(\varphi_1 \wedge (\varphi_1 \until_{(b, 2b)} \varphi_2)\big) \vee \bigg( \neg \varphi_\textit{ylb} \wedge \Big( \varphi_2 \vee \big(\varphi_1 \wedge (\varphi_1 \until_{\lopen{0, b}} \varphi_2) \big) \Big) \bigg) \Bigg)
\]
does not hold at all positions $j'$, $i < j' \leq l$. Consider the following cases:
\begin{itemize}
\item $\neg \varphi_\textit{ylb}$ holds at $l$: $\varphi_1 \wedge \globally_{\lopen{0, b}} \varphi_1$
does not hold at $l$, and therefore $\varphi_\textit{ugb}$ fails to hold at $i$.
\item $\varphi_\textit{ylb}$ holds at $l$: Consider the following cases:
	\begin{itemize}
	\item There is an event in $(\tau_l - 2b, \tau_l - b)$: Let $j''$ be the maximal position of such an event. We have $j'' + 1 < l$, $\tau_{j'' + 1} - \tau_{j''} \geq b$
	and $\tau_l - \tau_{j'' + 1} < b$. However, it follows that $\varphi_\textit{ylb}$ does not hold at $j'' + 1$
	and $\varphi_1 \wedge \globally_{\lopen{0, b}} \varphi_1$ holds at $j'' + 1$, which is a contradiction.
	\item There is no event in $(\tau_l - 2b, \tau_l - b)$: Let $j''$ be the minimal position such that
	$\tau_{j''} \in \ropen{\tau_l - b, \tau_l}$. It is clear that $\varphi_\textit{ylb}$ does not hold at $j''$ and
	$\varphi_1 \wedge \globally_{\lopen{0, b}} \varphi_1$ must hold at $j''$, which is a contradiction. \qedhere
	\end{itemize}
\end{itemize}
\end{proof}

\begin{prop}
The following equivalences hold over infinite timed words:
\small
\[
    \arraycolsep=0.3ex
    \begin{array}{rcl}
    \theta \guntil^c_{(a, b)} \big( (\varphi_1 \until \varphi_2) \wedge \chi \big) & \iff & \theta \guntil^c_{(a, b)} \big( (\varphi_1 \until_{(0, 2b)} \varphi_2) \wedge \chi \big) \vee \Big( \big( \theta \guntil^c_{(a, b)} (\globally_{(0, 2b)} \varphi_1 \wedge \chi) \big) \wedge \varphi_\textit{ugb} \Big) \\[0.5em]
    \theta \guntil^c_{(a, b)} (\globally \varphi \wedge \chi) & \iff & \big( \theta \guntil^c_{(a, b)} (\globally_{(0, 2b)} \varphi \wedge \chi) \big) \wedge \varphi_\textit{ggb} \\[0.5em]
    \theta \guntil^c_{(a, b)} \big( (\varphi_1 \since \varphi_2) \wedge \chi \big) & \iff & \theta \guntil^c_{(a, b)} \big( (\varphi_1 \since_{(0, b)} \varphi_2) \wedge \chi \big) \vee \Big( \big( \theta \guntil^c_{(a, b)} (\henceforth_{(0, b)} \varphi_1 \wedge \chi) \big) \wedge \varphi_1 \since \varphi_2 \Big) \\[0.5em]
    \theta \guntil^c_{(a, b)} (\henceforth \varphi \wedge \chi) & \iff & \big( \theta \guntil^c_{(a, b)} (\henceforth_{(0, b)} \varphi \wedge \chi) \big) \wedge \henceforth \varphi \\[0.5em]
    \big( (\varphi_1 \until \varphi_2) \vee \chi \big) \guntil^c_{(a, b)} \theta & \iff & \big( (\varphi_1 \until_{(0, 2b)} \varphi_2) \vee \chi \big) \guntil^c_{(a, b)} \theta \\
    & & {} \vee \bigg( \Big( \big( (\varphi_1 \until_{(0, 2b)} \varphi_2) \vee \chi \big) \guntil^c_{\big(c, c + (b - a)\big)} (\globally_{(0, 2b)} \varphi_1) \Big) \wedge \eventually_{(a, b)} \theta \wedge \varphi_\textit{ugb} \bigg) \\[0.5em]
    \big( (\globally \varphi) \vee \chi \big) \guntil^c_{(a, b)} \theta & \iff & \chi \guntil^c_{(a, b)} \theta
    \vee \big( \chi \guntil^c_{\big(c, c + (b - a)\big)} (\globally_{(0, 2b)} \varphi_1) \wedge \eventually_{(a, b)} \theta \wedge \varphi_\textit{ggb} \big) \\[0.5em]
    \big( (\varphi_1 \since \varphi_2) \vee \chi \big) \guntil^c_{(a, b)} \theta & \iff & \big( (\varphi_1 \since_{(0, b)} \varphi_2) \vee \chi \big) \guntil^c_{(a, b)} \theta \\
    & & {} \vee \bigg( \Big( \big(\henceforth_{(0, b)} \varphi_1 \vee (\varphi_1 \since_{(0, b)} \varphi_2) \vee \chi \big) \guntil^c_{(a, b)} \theta \Big) \wedge \varphi_1 \since \varphi_2 \bigg) \\[0.5em]
    \big( (\henceforth \varphi) \vee \chi \big) \guntil^c_{(a, b)} \theta & \iff & \chi \guntil^c_{(a, b)} \theta \vee \Big( \big( (\henceforth_{(0, b)} \varphi \vee \chi) \guntil^c_{(a, b)} \theta \big) \wedge \henceforth \varphi \Big) \,.
    \end{array}
\]
\end{prop}

\begin{lem}\label{lem:termination}
For any \emph{\mtlpastg{}} formula $\varphi$, we can use the rules above to obtain an
equivalent \emph{\mtlpastg{}} formula $\hat{\varphi}$ in which no unbounded temporal
operator appears in the scope of a bounded temporal operator.
In particular, all occurrences of $\guntil_I^c, \gsince_I^c$ have $I$ bounded.
\end{lem}
\begin{proof}
Define the \emph{unbounding depth} $ud(\varphi)$ of an \mtlpastg{} formula
$\varphi$ to be the modal depth of $\varphi$ counting only unbounded
operators. We demonstrate a rewriting process on $\varphi$ which terminates in an
equivalent formula $\hat{\varphi}$ such that any subformula $\hat{\psi}$ of $\hat{\varphi}$
with outermost operator bounded has $ud(\hat{\psi}) = 0$.

Assume that the input formula $\varphi$ is in normal form.
Let $k$ be the largest unbounding depth among all subformulas of $\varphi$
with bounded outermost operators.
We pick all minimal (w.r.t.\ subformula order) such subformulas $\psi$ with $ud(\psi) = k$.
By applying the rules in Section~\ref{sec:separation},
we can rewrite $\psi$ into $\psi'$ where all subformulas of $\psi'$
with bounded outermost operators have unbounded depths strictly
less than $k$.
We then substitute these $\psi'$ back into $\varphi$
to obtain $\varphi'$.
We repeat this step until there remain no bounded temporal operators with unbounding depth $k$.
The rules that rewrite a formula into normal form are used whenever necessary on
relevant subformulas---this never affects their unbounding depths, and note that we never
introduce $\guntil_I^c$ or $\gsince_I^c$.
It is easy to see that we will eventually obtain such a formula $\varphi^*$.
Now rewrite $\varphi^*$ into normal form
and start over again. This is to be repeated until we reach $\hat{\varphi}$.
\end{proof}

\paragraph{\emph{Completing the separation}}
We now have an \mtlpastg{} formula $\hat{\varphi}$ in which no unbounded temporal
operator appears in the scope of a bounded temporal operator.
If we regard each bounded subformula as a new monadic predicate, the formula $\hat{\varphi}$ can
be seen as an \ltlpast{} formula $\Phi$, on which Gabbay's separation theorem is applicable.
\begin{thmC}[{\cite[Theorem $3$]{Gabbay1980}}]
Every \emph{\ltlpast{}} formula is equivalent (over discrete complete models)
to a Boolean combination of
\begin{itemize}
\item atomic formulas
\item formulas of the form $\varphi_1 \until \varphi_2$ such that $\varphi_1$ and $\varphi_2$ use only $\until$
\item formulas of the form $\varphi_1 \since \varphi_2$ such that $\varphi_1$ and $\varphi_2$ use only $\since$.
\end{itemize}
\end{thmC}
\begin{lem}
Every \emph{\mtlpastg{}} formula is equivalent to a Boolean combination of
\begin{itemize}
\item bounded \mtlpastg{} formulas
\item formulas that use arbitrary $\until_I$ but only bounded $\since_I$, $\guntil_I^c$, $\gsince_I^c$
\item formulas that use arbitrary $\since_I$ but only bounded $\until_I$, $\guntil_I^c$, $\gsince_I^c$.
\end{itemize}
\end{lem}

\noindent
We now prove the main theorem of this subsection: each \mtlpastg{} formula
is equivalent to a \emph{syntactically separated} \mtlpastg{} formula.
\begin{thm}\label{thm:separation}
Every \emph{\mtlpastg{}} formula can be written as a
Boolean combination of
\begin{itemize}
\item bounded \mtlpastg{} formulas
\item formulas of the form $\mathbf{false} \guntil_{\geq M}^M \varphi$ where $M \in \mathbb{N}$
\item formulas of the form $\mathbf{false} \gsince_{\geq M}^M \varphi$ where $M \in \mathbb{N}$.
\end{itemize}
\end{thm}
\begin{proof}
Suppose that we have an \mtlpastg{} formula $\varphi$ with no unbounded $\since$.
If $\varphi$ is bounded then we are done. Otherwise we can apply Lemma~\ref{lem:termination}
(note in particular that it does not introduce new unbounded $\until$ operators)
and further assume that $\varphi = \varphi_1 \until \varphi_2$.
Then, for any $M \in \mathbb{N}$, we can rewrite $\varphi$ into
\[
\varphi_1 \until_{< M} \varphi_2 \vee \Bigg( \globally_{<M} \varphi_1 \wedge \bigg( \mathbf{false} \guntil_{\geq M}^M \Big( \varphi_2 \vee \big( \varphi_1 \wedge (\varphi_1 \until \varphi_2) \big) \Big) \bigg) \Bigg) \,.
\]
It is clear that $\varphi_1$ and $\varphi_2$, and therefore $\varphi_1 \until_{< M} \varphi_2$ and $\globally_{<M} \varphi_1$,
have strictly fewer unbounded $\until$ operators than $\varphi$. By the induction hypothesis, $\varphi$
is equivalent to a syntactically separated \mtlpastg{} formula. The case of
formulas with no unbounded $\until$ is symmetric.
\end{proof}

\subsection{Expressing bounded \foone{} formulas}

In this section, we describe how to express bounded \foone{} formulas with a
single free variable in \mtlpastg{}. The use of rational constants
is crucial here; by~\cite{Hirshfeld2007}, \mtlpastg{} cannot express all counting modalities (which can be written as bounded \foone{} formulas) if only integer constants are allowed.
As some techniques here are exactly similar to that of Section~\ref{subsec:translation}, we omit certain explanations.

Suppose that we are given such a formula $\vartheta(x)$.
As before, we assume that each quantifier in $\vartheta(x)$ uses a fresh new variable
and $\vartheta(x)$ contains only existential quantifiers.
We say that $\vartheta(x)$ is \emph{$N$-bounded} if each subformula $\exists x' \, \psi$ of $\vartheta(x)$
is of the form
\[
\exists x' \, \big( (x' > x \implies d(x, x') < N) \wedge (x' < x \implies d(x, x') \leq N) \wedge \ldots \big) \,.
\]
Namely, $\vartheta(x)$ only refers to the events in the half-open interval $\ropen{x - N, x + n}$.
Similarly, we say that $\vartheta(x)$ is a \emph{unit formula} if each subformula $\exists x' \, \psi$ of $\vartheta(x)$
is of the form
\[
\exists x' \, \big(x' \geq x \wedge d(x, x') < 1 \wedge \ldots \big) \,.
\]
In this case, $\vartheta(x)$ only refers to the events in $\ropen{x, x+1}$.

\paragraph{\emph{Stacking events around a point}}

Let $\rho$ be an infinite timed word over $\Sigma_\mathbf{P}$,
$\overline{\mathbf{P}} = \{P_i \mid P \in \mathbf{P}, -N \leq i < N\}$ and $\overline{\mathbf{Q}} = \{Q_i \mid N \leq i < N \}$.
For each $t \in \rho$, we can construct a (finite) $\ropen{0, 1}$-timed word $\overline{\rho_t}$ over $\Sigma_{\overline{\mathbf{P}} \cup \overline{\mathbf{Q}}}$
that satisfies the following:
\begin{itemize}
\item For all $\overline{t} \in \ropen{0, 1}$ and $-N \leq i < N$, $P_i$ holds at $\overline{t} \in \overline{\rho_t}$ iff $P$ holds at $i + \overline{t} \in \rho$.
\item For all $\overline{t} \in \ropen{0, 1}$ and $-N \leq i < N$, $Q_i$ holds at $\overline{t} \in \overline{\rho_t}$ iff $i + \overline{t} \in \rho$.
\end{itemize}

\paragraph{\emph{Stacking $N$-bounded \foone{} formulas}}

Now let $\vartheta(x)$ be an $N$-bounded \foone{} formula.
Recursively replace every subformula $\exists x' \, \theta$ by
\[
\exists x' \, \Big( \big( Q_{-N}(x') \wedge \theta[x' + (-N)/x'] \big) \vee \ldots \vee \big( Q_{N-1}(x') \wedge \theta[x' + (N - 1)/x'] \big) \Big)
\]
where $\vartheta[e/x]$ denotes the formula obtained by substituting all free occurrences of $x$
in $\vartheta$ by $e$. We then carry out the following syntactic substitutions:
\begin{itemize}
\item For each inequality of the form $x_1 + k_1 < x_2 + k_2$, replace it with
	\begin{itemize}
	\item $x_1 < x_2$ if $k_1 = k_2$
	\item $\mathbf{true}$ if $k_1 < k_2$
	\item $\mathbf{false}$ if $k_1 > k_2$
	\end{itemize}
\item For each distance formula, e.g., $d(x_1 + k_1, x_2 + k_2) < 2$, replace it with
	\begin{itemize}
	\item $\mathbf{true}$ if $|k_1 - k_2| \leq 1$
	\item $x_2 < x_1$ if $k_2 - k_1 = 2$
	\item $x_1 < x_2$ if $k_1 - k_2 = 2$
	\item $\mathbf{false}$ if $|k_1 - k_2| > 2$
	\end{itemize}
\item Replace terms of the form $P(x_1 + k)$ with $P_k(x_1)$.
\end{itemize}
Finally, recursively replace every subformula $\exists x' \, \theta$ by $\exists x' \, \big( x' \geq x \wedge d(x, x') < 1 \wedge \theta \big)$.
This gives a unit formula $\overline{\vartheta}(x)$ such that for each $t \in \rho$,
\[
\rho, t \models \vartheta(x) \iff \overline{\rho_t}, 0 \models \overline{\vartheta}(x) \,.
\]

\paragraph{\emph{Unstacking}}

For each $\overline{\rho_t}$, we add an event at time $1$ (at which no monadic predicate holds) and call the resulting
$\ropen{0, 1}$-timed word $\overline{\rho_t}'$. It is clear that
\[
\overline{\rho_t}, 0 \models \overline{\vartheta}(x) \iff \overline{\rho_t}', 0, 1 \models \overline{\vartheta}'(x, y)
\]
where $\overline{\vartheta}'(x, y)$ is a non-metric \textmd{\textsf{FO[$<$]}} formula
obtained by replacing all distance formulas of the form $d(x, x') < 1$ with $x' < y$ in $\overline{\vartheta}(x)$.
We now invoke a normal form lemma from~\cite{Gabbay1980} to rewrite $\overline{\vartheta}'(x, y)$ into
a disjunction of \emph{decomposition formulas}.
\begin{lemC}[\cite{Gabbay1980}]
Every \emph{\textmd{\textsf{FO[$<$]}}} formula $\theta(x, y)$ in which all quantifications are of the form
$\exists x' \, (x' \geq x \wedge x' < y \wedge \ldots )$
is equivalent to a disjunction of decomposition formulas, i.e.,~\emph{\textmd{\textsf{FO[$<$]}}} formulas of the form
\[
    \arraycolsep=0.3ex
    \begin{array}{ll}
    x < y & {} \wedge \exists z_0 \ldots \exists z_n \, (x = z_0 < \cdots < z_n = y) \\
                    & {} \wedge \bigwedge \{ \Phi_i(z_i): 0 \leq i < n \} \\
                    & {} \wedge \bigwedge \{ \forall u \, \big(z_i < u < z_{i+1} \implies \Psi_i(u)\big): 0 \leq i < n \}
    \end{array}
\]
where $\Phi_i$ and $\Psi_i$ are \emph{\ltl{}} formulas.\footnote{This version of the lemma
follows from Lemma $4$ and Main Lemma in~\cite{Gabbay1980}.}
\end{lemC}

In fact, when the underlying order is discrete (as is the case here), we can further
postulate that $\Phi_i$ and $\Psi_i$ are simply Boolean combinations of atomic formulas~\cite{Dam1994}.
It follows that $\overline{\vartheta}(x)$ is equivalent to a disjunction of unit formulas $\overline{\delta}(x)$ of the form
\[
    \arraycolsep=0.3ex
    \begin{array}{ll}
    \exists z_0 \ldots \exists z_{n-1} \, (x = z_0 < \cdots < z_{n-1}) & {} \wedge d(x, z_{n - 1}) < 1  \\
                    & {} \wedge \bigwedge \{ \Phi_i(z_i): 0 \leq i < n \} \\
                    & {} \wedge \bigwedge \{ \forall u \, \big(z_i < u < z_{i+1} \implies \Psi_i(u)\big): 0 \leq i < n - 1 \} \\
                    & {} \wedge \forall u \, \big(z_{n-1} < u \wedge d(x, u) < 1 \implies \Psi_{n - 1}(u) \big)
    \end{array}
\]
where $\Phi_i$ and $\Psi_i$ are Boolean combinations of atomic formulas.

It remains to show that for each such unit formula $\overline{\delta}(x)$ and each $t \in \rho$, we can construct an \mtlpastg{} formula
$\varphi$ such that
\[
\overline{\rho_t}, 0 \models \overline{\delta}({x}) \iff \rho, t \models \varphi \,.
\]
For later convenience, we prove a stronger claim, i.e.,~we can handle
\fo{} formulas of the following form for any rational number $r$, $0 \leq r < 1$:
\[
    \arraycolsep=0.3ex
    \begin{array}{ll}
    \exists z_0 \ldots \exists z_{n-1} \, (x = z_0 < \cdots < z_{n-1}) & {} \wedge d(x, z_1) > r \wedge d(x, z_{n - 1}) < 1  \\
                    & {} \wedge \bigwedge \{ \Phi_i(z_i): 1 \leq i < n \} \\
                    & {} \wedge \forall u \, \big(x < u \wedge u < z_1 \wedge d(x, u) > r \implies \Psi_{0}(u) \big) \\
                    & {} \wedge \bigwedge \{ \forall u \, \big(z_i < u < z_{i+1} \implies \Psi_i(u)\big): 1 \leq i < n - 1 \} \\
                    & {} \wedge \forall u \, \big(z_{n-1} < u \wedge d(x, u) < 1 \implies \Psi_{n - 1}(u) \big) \,.
    \end{array}
\]
The proof is by induction on the number of existential quantifiers in $\overline{\delta}(x)$.
Before we proceed with the proof, we define a function $f$ that maps
a Boolean combination $\Omega$ of atomic formulas over $\overline{\mathbf{P}} \cup \overline{\mathbf{Q}}$
and $i$, $-N \leq i < N$ to an \mtlpastg{} formula $f(\Omega, i)$ over $\mathbf{P}$:
\begin{itemize}
\item $f(P_j, i) = \begin{cases}
						\once_{= (i - j)} P & \text{if } i > j\\
						P & \text{if } i = j \\
						\eventually_{= (j - i)} P & \text{if } i < j \\
						\end{cases}$
\item $f(Q_j, i) = \begin{cases}
						\once_{= (i - j)} \mathbf{true} & \text{if } i > j\\
						\mathbf{true} & \text{if } i = j \\
						\eventually_{= (j - i)} \mathbf{true} & \text{if } i < j \\
						\end{cases}$
\item $f(\mathbf{true}, i) = \mathbf{true}$
\item $f(\Omega_1 \wedge \Omega_2, i) = f(\Omega_1, i) \wedge f(\Omega_2, i)$
\item $f(\neg \Omega, i) = \neg f(\Omega, i)$.
\end{itemize}

\noindent
Now first consider the base step. We have
\[
\overline{\delta}(x) = \forall u \, \big(x < u \wedge d(x, u) > r \wedge d(x, u) < 1 \implies \Psi(u) \big)
\]
where $\Psi$ is a Boolean combination of atomic formulas. It is clear that
\[
\varphi = \bigwedge_{0 \leq i < N} \big(\globally_{(i + r, i + 1)} f(\Psi, i)\big) \wedge \bigwedge_{-N \leq i < 0} \big(\henceforth_{(|i + 1|, |i + r|)} f(\Psi, i)\big) \,.
\]

For the induction step we need to consider how $z_1$, \ldots, $z_{n-1}$ are scattered in $(r, 1)$.
Let us split $(r, 1)$ into an open interval $(r, r + \frac{1 - r}{2n})$ and $2n - 1$
half-open intervals $\ropen{r + \frac{1 - r}{2n}, r + \frac{2(1 - r)}{2n}}$, $\ropen{r + \frac{2(1-r)}{2n}, r + \frac{3(1-r)}{2n}}$, \ldots, $\ropen{r + \frac{(2n-1)(1-r)}{2n}, 1}$.
Consider the following cases:
\begin{enumerate}[label=(\roman*).] % chktex 36
\item $\{ z_1, \ldots, z_{n - 1} \} \subseteq (r, r + \frac{1 - r}{2n})$ or $\{ z_1, \ldots, z_{n - 1} \} \subseteq \ropen{r + \frac{k(1 - r)}{2n}, r + \frac{(k + 1)(1 - r)}{2n}}$ for some $k$, $1 \leq k < n$.
\item $\{ z_1, \ldots, z_{n - 1} \} \subseteq \ropen{r + \frac{k(1 - r)}{2n}, r + \frac{(k + 1)(1 - r)}{2n}}$ for some $k$, $n \leq k < 2n$.
\item There exists $k$, $1 \leq k < 2n$ and $l$, $1 \leq l < n - 1$ such that $z_l < r + \frac{k(1-r)}{2n} \leq z_{l + 1}$ (i.e.,~$z_1$, \ldots, $z_{n-1}$ are not in a single interval).
\end{enumerate}
We detail the construction of a formula $\psi$ in each case; the desired formula $\varphi$ is the disjunction of these $\psi$.
The correctness proofs are omitted as they are similar to the proof of Proposition~\ref{prop:translation}.
\begin{itemize}
\item Case (i):
Consider the subcase $z_1 > r + \frac{k(1-r)}{2n}$.  Let
\[
\overrightarrow{\varphi}^i_{n-1} = \bigwedge_{0 \leq j < N - i} \big( \globally_{(j, j + \frac{1-r}{2n})} f(\Psi_{n-1}, i + j) \big) \wedge \bigwedge_{-N - i \leq j < 0} \big( \henceforth_{(|j + \frac{1-r}{2n}|, |j|)} f(\Psi_{n-1}, i + j) \big)
\]
for all $i$, $-N \leq i < N$ and recursively define
\[
\overrightarrow{\varphi}^i_{m} = \bigvee_{-N - i \leq j < N - i} \bigg( \bigwedge_{-N - i \leq h < N - i} \Big( \big( f(\Psi_m, i + h) \big) \guntil_{(j, j + \frac{1-r}{2n})}^h \big( f(\Phi_{m+1}, i + j) \wedge \overrightarrow{\varphi}^{i+j}_{m+1} \big) \Big)  \bigg)
\]
for all $i$, $-N \leq i < N$ and $m$, $1 \leq m < n - 1$.  Let $\alpha_k$ be the conjunction of
\[
\bigwedge_{0 \leq i < N} \big( \globally_{\lopen{i + r, i + r + \frac{k(1-r)}{2n}}} f(\Psi_0, i) \big) \wedge \bigwedge_{-N \leq i < 0} \big( \henceforth_{\ropen{|i + r + \frac{k(1-r)}{2n}|, |i + r|}} f(\Psi_0, i) \big)
\]
and
\[
\bigvee_{-N \leq j < N} \bigg( \bigwedge_{-N \leq h < N} \Big( \big( f(\Psi_0, h) \big) \guntil_{(j + r + \frac{k(1-r)}{2n}, j + r + \frac{(k+1)(1-r)}{2n})}^{h + r + \frac{k(1-r)}{2n}} \big( f(\Phi_{1}, j) \wedge \overrightarrow{\varphi}^{j}_{1} \big) \Big)  \bigg)
\]
and
\[
\bigwedge_{0 \leq i < N} \big( \globally_{\ropen{i + r + \frac{(k + 1)(1-r)}{2n}, i + 1}} f(\Psi_{n-1}, i) \big) \wedge \bigwedge_{-N \leq i < 0} \big( \henceforth_{\lopen{|i + 1|, |i + r + \frac{(k + 1)(1-r)}{2n}|}} f(\Psi_{n - 1}, i) \big) \,.
\]
Similarly, we construct $\alpha_k'$ to handle the subcase $z_1 = r + \frac{k(1-r)}{2n}$.
The formula $\psi$ is the disjunction of formulas $\{ \alpha_k \mid 0 \leq k < n \}$ and $\{ \alpha_k' \mid 0 < k < n \}$.
\item Case (ii):
Let
\[
\overleftarrow{\varphi}^i_{1} = \bigwedge_{0 < j < N - i} \big( \globally_{(j - \frac{1-r}{2n}, j)} f(\Psi_{0}, i + j) \big) \wedge \bigwedge_{-N - i \leq j \leq 0} \big( \henceforth_{(|j|, |j -  \frac{1-r}{2n}|)} f(\Psi_{0}, i + j) \big)
\]
for all $i$, $-N \leq i < N$ and recursively define
\[
\overleftarrow{\varphi}^i_{m} = \bigvee_{-N - i \leq j < N - i} \bigg( \bigwedge_{-N - i \leq h < N - i} \Big( \big( f(\Psi_{m-1}, i + h) \big) \gsince_{(j, j + \frac{1-r}{2n})}^h \big( f(\Phi_{m-1}, i + j) \wedge \overleftarrow{\varphi}^{i+j}_{m-1} \big) \Big)  \bigg)
\]
for all $i$, $-N \leq i < N$ and $m$, $1 < m \leq n - 1$.  Let $\beta_k$ be the conjunction of
\[
\bigwedge_{0 \leq i < N} \big( \globally_{\ropen{i + r + \frac{(k + 1)(1-r)}{2n}, i + 1}} f(\Psi_{n-1}, i) \big) \wedge \bigwedge_{-N \leq i < 0} \big( \henceforth_{\lopen{|i + 1|, |i + r + \frac{(k+1)(1-r)}{2n}|}} f(\Psi_{n-1}, i) \big)
\]
and
\[
\bigvee_{-N \leq j < N} \bigg( \bigwedge_{-N \leq h < N} \Big( \big( f(\Psi_{n - 1}, h) \big) \gsince_{\lopen{-(j + r + \frac{(k + 1)(1-r)}{2n}), -(j + r + \frac{k(1-r)}{2n})}}^{-(h + r + \frac{(k + 1)(1-r)}{2n})} \big( f(\Phi_{n - 1}, j) \wedge \overleftarrow{\varphi}^{j}_{n - 1} \big) \Big)  \bigg)
\]
and
\[
\bigwedge_{0 \leq i < N} \big( \globally_{(i + r, i + r + \frac{k(1-r)}{2n})} f(\Psi_{0}, i) \big) \wedge \bigwedge_{-N \leq i < 0} \big( \henceforth_{(|i + r + \frac{k(1-r)}{2n}|, |i + r|)} f(\Psi_{0}, i) \big) \,.
\]
The formula $\psi$ is the disjunction of $\beta_k$, $n \leq k < 2n$.
\item Case (iii):
Suppose that $z_l < r + \frac{k(1-r)}{2n} \leq z_{l+1}$ for some $k$, $1 \leq k < 2n$
and $l$, $1 \leq l < n - 1$. Consider the following subcases:
	\begin{itemize}
	\item $r + \frac{k(1-r)}{2n} < z_{l+1}$:
	This can be handled by the conjunction of the formulas below:
	\begin{itemize}
	\item $\{z_1, \ldots, z_l\} \subseteq (r, r + \frac{k(1-r)}{2n})$: We can scale the corresponding \fo{} formula by
	$\frac{1}{r + \frac{k(1-r)}{2n}}$, apply the induction hypothesis (with $r' = \frac{r}{r + \frac{k(1-r)}{2n}}$)
	and scale the resulting \mtlpastg{} formula by $r + \frac{k(1-r)}{2n}$.
	\item $\{z_{l+1}, \ldots, z_{n-1}\} \subseteq (r + \frac{k(1-r)}{2n}, 1)$: We can set $r' = r + \frac{k(1-r)}{2n}$ and apply the induction hypothesis.
	\end{itemize}

	\item $r + \frac{k(1-r)}{2n} = z_{l+1}$: Exactly similar except that we also use the following formula as a conjunct:
	\[
		\bigvee_{0 \leq i < N} \big( \eventually_{= i + r + \frac{k(1-r)}{2n}} f(\Phi_{l + 1}, i) \big) \vee \bigvee_{-N \leq i < 0} \big( \once_{= |i + r + \frac{k(1-r)}{2n}|} f(\Phi_{l+1}, i) \big) \,.
	\]
	\end{itemize}
The formula $\psi$ is the disjunction of these formulas for all $k$, $1 \leq k < 2n$ and $l$, $1 \leq l < n - 1$.
\end{itemize}
Finally, observe that the original claim can be achieved by setting $r = 0$ and using the conjunct $f(\Phi_0, 0)$.
We can now state the following theorem.
\begin{thm}\label{thm:expcompforbounded}
For every $N$-bounded \emph{\foone{}} formula $\vartheta(x)$, there exists an equivalent \emph{\mtlpastg{}} formula $\varphi$ (with rational constants).
\end{thm}

\subsection{Expressive completeness of \mtlpastg{}}\label{sec:unboundedexpcomp}

In this section, we show that any \fo{} formula with one free variable
can be expressed as an \mtlpastg{} formula. The crucial idea here is that
we can separate formulas `far enough' that all references to a certain
variable become vacuous.
To this end, we define $\mathit{fr}(\varphi)$ and $\mathit{pr}(\varphi)$
(\emph{future-reach} and \emph{past-reach}) for an \mtlpastg{} formula $\varphi$ as follows:
\begin{itemize}
\item $\mathit{fr}(\mathbf{true}) = \mathit{pr}(\mathbf{true}) = \mathit{fr}(P) = \mathit{pr}(P) = 0$ for all $P \in \mathbf{P}$
\item $\mathit{fr}(\varphi_1 \until_I \varphi_2) = \sup(I) + \max\big(\mathit{fr}(\varphi_1), \mathit{fr}(\varphi_2)\big)$
\item $\mathit{pr}(\varphi_1 \since_I \varphi_2) = \sup(I) + \max\big(\mathit{pr}(\varphi_1), \mathit{pr}(\varphi_2)\big)$
\item $\mathit{fr}(\varphi_1 \since_I \varphi_2) = \max\big(\mathit{fr}(\varphi_1), \mathit{fr}(\varphi_2) - \inf(I)\big)$
\item $\mathit{pr}(\varphi_1 \until_I \varphi_2) = \max\big(\mathit{pr}(\varphi_1), \mathit{pr}(\varphi_2) - \inf(I)\big)$
\item $\mathit{fr}(\varphi_1 \guntil^c_I \varphi_2) = \max\big( c + |I| + \mathit{fr}(\varphi_1), \sup(I) + \mathit{fr}(\varphi_2) \big)$
\item $\mathit{pr}(\varphi_1 \gsince^c_I \varphi_2) = \max\big( c + |I| + \mathit{pr}(\varphi_1), \sup(I) + \mathit{pr}(\varphi_2) \big)$
\item $\mathit{fr}(\varphi_1 \gsince^c_I \varphi_2) = \max\big( \mathit{fr}(\varphi_1) - c, \mathit{fr}(\varphi_2) - \inf(I) \big)$
\item $\mathit{pr}(\varphi_1 \guntil^c_I \varphi_2) = \max\big( \mathit{pr}(\varphi_1) - c, \mathit{pr}(\varphi_2) - \inf(I) \big)$.
\end{itemize}
The cases for Boolean connectives are defined in the expected way.
Intuitively, these are meant as over-approximations of the lengths of the time horizons needed to
determine the truth value of $\varphi$.

\begin{thm}
For every \emph{\foone{}} formula $\vartheta(x)$, there exists an equivalent \emph{\mtlpastg{}} formula $\varphi$ (with rational constants).
\end{thm}
\begin{proof}
The proof is by induction on the quantifier depth of $\vartheta(x)$.
In what follows, let the set of monadic predicates be $\mathbf{P}$.
As before, we assume that each quantifier in $\vartheta(x)$ uses a fresh new variable.

\begin{itemize}
\item \emph{Base step.} $\vartheta(x)$ is a Boolean combination of atomic formulas $P(x)$, $x < x$, $d(x, x) \sim c$, $\mathbf{true}$.
		We can replace them by $P$, $\mathbf{false}$, $0 \sim c$ and $\mathbf{true}$ respectively to obtain $\varphi$.
\item \emph{Induction step.} Without loss of generality assume that $\vartheta(x) = \exists y \, \theta(x, y)$.
		Our goal here is to remove $x$ from $\theta(x, y)$. For this purpose, we can remove $x < x$ and $d(x, x) \sim c$
		as before and use a mapping $\gamma: \mathbf{P} \mapsto \{ \mathbf{true}, \mathbf{false} \}$ to determine
		the truth value of $P(x)$ for each $P \in \mathbf{P}$. Thus we can rewrite $\exists y \, \theta(x, y)$ as
		\begin{equation}\label{eqn:gamma}
			\bigvee_{\gamma} \big( \eta_\gamma(x) \wedge \exists y \, \theta_\gamma(x, y) \big)
		\end{equation}
		where
        \[
            \eta_\gamma = \bigwedge_{P \in \mathbf{P}} \big(P(x) \iff \gamma(P)\big)
        \]
		and $\theta_\gamma(x, y)$ is obtained by replacing $P(x)$ with $\gamma(P)$ for all $P \in \mathbf{P}$ in $\theta(x, y)$.
Observe that in each $\theta_\gamma(x, y)$, $x$ only appears in atomic formulas of the form $x < z$, $z < x$, $d(x, z) \sim c$
		and $d(z, x) \sim c$ where $\sim \; \in \{<, >\}$. We now introduce new monadic predicates $P_{<}$, $P_{>}$, and $P_{\sim c}$ for each $d(x, z) \sim c$ or $d(z, x) \sim c$
		that correspond to these atomic formulas. Namely, we write $x < z$ as $P_{<}(z)$, $z < x$ as $P_{>}(z)$, and
		$d(x, z) \sim c$ or $d(z, x) \sim c$ as $P_{\sim c}(z)$. Apply these substitutions to (\ref{eqn:gamma})
		yields
		\begin{equation}\label{eqn:nox}
			\bigvee_{\gamma} \big( \eta_\gamma(x) \wedge \exists y \, \theta_\gamma'(y) \big)
		\end{equation}
		where $x$ does not occur in each $\theta_\gamma'(y)$. In particular, (\ref{eqn:gamma}) and (\ref{eqn:nox})
		have the same truth value at any given point if $P_{<}$, $P_{>}$ and all $P_{\sim c}$ are interpreted correctly
		with respect to that point.
Each $\eta_\gamma(x)$ is clearly equivalent to an \mtlpastg{} formula $\psi_\gamma$.
		By the induction hypothesis, each $\theta_\gamma'(y)$ is also equivalent to an \mtlpastg{} formula $\varphi_\gamma$.
		It follows that (\ref{eqn:nox}) is equivalent to the following \mtlpastg{} formula:
		\[
			\varphi' = \bigvee_{\gamma} \big( \psi_\gamma \wedge (\once \varphi_\gamma \vee \varphi_\gamma \vee \eventually \varphi_\gamma) \big) \,.
		\]
		By Theorem~\ref{thm:separation} and noting that $M \in \mathbb{N}$
		can be chosen arbitrarily, $\varphi'$ is equivalent to a Boolean combination $\varphi''$ of
		\begin{itemize}
		\item bounded formulas
		\item formulas of the form $\mathbf{false} \guntil_{\geq M}^M \psi$ such that $M > c_\maxit + \mathit{pr}(\psi)$
		\item formulas of the form $\mathbf{false} \gsince_{\geq M}^M \psi$ such that $M > c_\maxit + \mathit{fr}(\psi)$.
		\end{itemize}
		where $c_\maxit$ are the largest constants appearing in monadic predicates of the form
		$P_{\sim c}$ in respective formulas $\psi$.
Now suppose that $\varphi''$ is evaluated at $t_1$.
		For the formulas of the second type, since all references to $P_{<}$, $P_{>}$ and all $P_{\sim c}$
		must happen at time strictly greater than $t_1 + c_\maxit$, we can simply replace them with
		$\mathbf{true}$, $\mathbf{false}$ and $c_\maxit + 1 \sim c$ to obtain equivalent \mtlpastg{}
		formulas over $\mathbf{P}$. The formulas of the third type
		can be dealt with similarly. Finally, for the formulas of the first type, we replace
		$P_{<}$, $P_{>}$ and all $P_{\sim c}$ with $x < z$, $z < x$ and $d(x, z) \sim c$.
		The resulting formulas are clearly bounded \fo{} formulas. We can scale them to bounded \foone{} formulas,
		apply Theorem~\ref{thm:expcompforbounded} and then scale back
		to obtain equivalent \mtlpastg{} formulas over $\mathbf{P}$. \qedhere
\end{itemize}
\end{proof}

\noindent
The main result of this chapter now follows from a simple scaling argument.

\begin{thm}\label{thm:unboundedexpcomp}
\emph{\mtlpastg{}} with rational constants is expressively complete for \emph{\fo{}} over infinite timed words.
\end{thm}

\section{Monitoring of \mtlpastg{} properties}

While the expressive completeness result in the last section
may be interesting from a theoretical point of view,
it is unclear how it can benefit
practical verification tasks
as the model-checking problem for \mtl{} is
already undecidable~\cite{Ouaknine2006a}.
Nonetheless, we show that those results can be very useful in
\emph{monitoring}, a core element of \emph{runtime verification}.
We first define some basic notions used throughout
this section.
Then we give a bi-linear offline trace-checking algorithm for \mtlpastg{},
which is later modified to work in an online fashion (under a
bounded-variability assumption) and used as the basis of a
monitoring procedure for an `easy-to-monitor' fragment of \mtlpastg{}.\footnote{In this section
we assume that all timestamps are rational, can be finitely represented (e.g., with a built-in data type), and additions and subtractions on them can be done in constant time.}
The main advantage of the proposed procedure is that it is \emph{trace-length independent}, i.e.,~the space usage
is independent of the length of the (growing) trace.
Finally, we show that our approach extends to arbitrary \mtlpastg{} formulas
via the syntactic rewriting rules in Section~\ref{sec:separation}.

\subsection{Satisfaction over finite prefixes}

\paragraph{\emph{Truncated semantics}}

As one is naturally concerned with truncated traces in monitoring,
it is useful to define satisfaction relations of \mtlpastg{} formulas
over finite timed words.  To this end, we adopt a timed version of the
\emph{truncated semantics}~\cite{Eisner2003} which offers
\emph{strong} and \emph{weak} views on satisfaction over finite timed words.
Intuitively, these views indicate whether the satisfaction of
the formula on the whole (infinite) trace is `clearly' confirmed/refuted
by the finite prefix read so far.
In the strong view, one is \emph{pessimistic} on satisfaction---for example,
$\globally P$ can never be strongly satisfied by any finite timed word, as any such finite timed word
can be extended into an infinite timed word that violates the formula.
Conversely, in the weak view one is \emph{optimistic} on satisfaction---for example, $\eventually_{< 5} P$ is weakly satisfied by any finite
timed word whose timestamps are all strictly less than $5$, since
one can always extend such into an infinite timed word that satisfies the formula.
We also consider the \emph{neutral} view, which extends the
traditional \ltlpast{} semantics over finite words~\cite{Manna1995}
to \mtlpast{}. In what follows, we denote the strong, neutral and weak
satisfaction relations by $\modelsfs$, $\modelsfn$ and $\modelsfw$ respectively.
We write $\rho \modelsfs \varphi$ ($\rho \modelsfn \varphi$, $\rho \modelsfw \varphi$)
if $\rho, 0 \modelsfs \varphi$ ($\rho, 0 \modelsfn \varphi$, $\rho, 0 \modelsfw \varphi$).
\begin{defi}
The satisfaction relation $\rho, i \modelsfs \varphi$ for an
\emph{\mtlpastg{}} formula $\varphi$, a finite timed word $\rho =
(\sigma, \tau)$ and a position $i$, $0 \leq i < |\rho|$ is defined
as follows:
\begin{itemize}

\item $\rho, i \modelsfs P$ iff $P \in \sigma_i$

\item $\rho, i \modelsfs \mathbf{true}$

\item $\rho, i \modelsfs \varphi_1 \wedge \varphi_2$ iff $\rho, i \modelsfs \varphi_1$ and $\rho, i \modelsfs \varphi_2$

\item $\rho, i \modelsfs \neg \varphi$ iff $\rho, i \centernot \modelsfw \varphi$

\item $\rho, i \modelsfs \varphi_1 \until_I \varphi_2$ iff there
exists $j$, $i < j < |\rho|$ such that $\rho, j \modelsfs
\varphi_2$, $\tau_j - \tau_i \in I$, and $\rho, k \modelsfs
\varphi_1$ for all $k$ with $i < k < j$

\item $\rho, i \modelsfs \varphi_1 \since_I \varphi_2$ iff there exists $j$, $0 \leq j < i$ such
that $\rho, j \modelsfs \varphi_2$, $\tau_i - \tau_j \in I$,
and $\rho, k \modelsfs \varphi_1$ for all $k$ with $j < k < i$

\item $\rho, i \modelsfs \varphi_1 \guntil^c_I \varphi_2$ iff there exists $j$, $i < j < |\rho|$
such that $\rho, j \modelsfs \varphi_2$, $\tau_j - \tau_i \in I$,
and $\rho, k \modelsfs \varphi_1$ for all $k$, $i < k < j$ such that $\tau_{k} - \tau_i > c$ and
$\tau_j - \tau_{k} > a - c$ where $a = \inf(I)$

\item $\rho, i \modelsfs \varphi_1 \gsince^c_I \varphi_2$ iff there exists $j$, $0 \leq j < i$
such that $\rho, j \modelsfs \varphi_2$, $\tau_i - \tau_j \in I$,
and $\rho, k \modelsfs \varphi_1$ for all $k$, $j < k < i$ such that $\tau_i - \tau_{k} > c$
and $\tau_{k} - \tau_{j} > a - c$ where $a = \inf(I)$.

\end{itemize}
\end{defi}

\begin{defi}
The satisfaction relation $\rho, i \modelsfn \varphi$ for an
\emph{\mtlpastg{}} formula $\varphi$, a finite timed word $\rho =
(\sigma, \tau)$ and a position $i$, $0 \leq i < |\rho|$ is defined
as follows:
\begin{itemize}

\item $\rho, i \modelsfn P$ iff $P \in \sigma_i$

\item $\rho, i \modelsfn \mathbf{true}$

\item $\rho, i \modelsfn \varphi_1 \wedge \varphi_2$ iff $\rho, i \modelsfn \varphi_1$ and $\rho, i \modelsfn \varphi_2$

\item $\rho, i \modelsfn \neg \varphi$ iff $\rho, i \centernot \modelsfn \varphi$

\item $\rho, i \modelsfn \varphi_1 \until_I \varphi_2$ iff there
exists $j$, $i < j < |\rho|$ such that $\rho, j \modelsfn
\varphi_2$, $\tau_j - \tau_i \in I$, and $\rho, k \modelsfn
\varphi_1$ for all $k$ with $i < k < j$

\item $\rho, i \modelsfn \varphi_1 \since_I \varphi_2$ iff there exists $j$, $0 \leq j < i$ such
that $\rho, j \modelsfn \varphi_2$, $\tau_i - \tau_j \in I$,
and $\rho, k \modelsfn \varphi_1$ for all $k$ with $j < k < i$

\item $\rho, i \modelsfn \varphi_1 \guntil^c_I \varphi_2$ iff there exists $j$, $i < j < |\rho|$
such that $\rho, j \modelsfn \varphi_2$, $\tau_j - \tau_i \in I$,
and $\rho, k \modelsfn \varphi_1$ for all $k$, $i < k < j$ such that $\tau_{k} - \tau_i > c$ and
$\tau_j - \tau_{k} > a - c$ where $a = \inf(I)$

\item $\rho, i \modelsfn \varphi_1 \gsince^c_I \varphi_2$ iff there exists $j$, $0 \leq j < i$
such that $\rho, j \modelsfn \varphi_2$, $\tau_i - \tau_j \in I$,
and $\rho, k \modelsfn \varphi_1$ for all $k$, $j < k < i$ such that $\tau_i - \tau_{k} > c$
and $\tau_{k} - \tau_{j} > a - c$ where $a = \inf(I)$.

\end{itemize}
\end{defi}

\begin{defi}
The satisfaction relation $\rho, i \modelsfw \varphi$ for an
\emph{\mtlpastg{}} formula $\varphi$, a finite timed word $\rho =
(\sigma, \tau)$ and a position $i$, $0 \leq i < |\rho|$ is defined
as follows:
\begin{itemize}

\item $\rho, i \modelsfw P$ iff $P \in \sigma_i$

\item $\rho, i \modelsfw \mathbf{true}$

\item $\rho, i \modelsfw \varphi_1 \wedge \varphi_2$ iff $\rho, i
\modelsfw \varphi_1$ and $\rho, i \modelsfw \varphi_2$

\item $\rho, i \modelsfw \neg \varphi$ iff $\rho, i \centernot \modelsfs \varphi$

\item $\rho, i \modelsfw \varphi_1 \until_I \varphi_2$ iff either of
the following holds:
	\begin{itemize}
	\item there exists $j$, $i < j < |\rho|$ such that $\rho, j \modelsfw \varphi_2$, $\tau_j -
	\tau_i \in I$, and $\rho, k \modelsfw \varphi_1$ for all $k$ with $i < k < j$
	\item $\tau_{|\rho| - 1} - \tau_i < \sup(I)$ and $\rho, k \modelsfw \varphi_1$ for all $k$ with $i < k < |\rho|$
	\end{itemize}

\item $\rho, i \modelsfw \varphi_1 \since_I \varphi_2$ iff there exists
$j$, $0 \leq j < i$ such that $\rho, j \modelsfw \varphi_2$, $\tau_i
- \tau_j \in I$, and $\rho, k \modelsfw \varphi_1$ for all $k$ with $j < k < i$

\item $\rho, i \modelsfw \varphi_1 \guntil^c_I \varphi_2$ iff
either of the following holds:
	\begin{itemize}
	\item there exists $j$, $i < j < |\rho|$ such that $\rho, j \modelsfw \varphi_2,
	\tau_j - \tau_i \in I$, and $\rho, k \modelsfw \varphi_1$ for all $k$, $i < k < j$
	such that $\tau_{k} - \tau_i > c$ and $\tau_{j} - \tau_{k} > a - c$ where $a = \inf(I)$

	\item $\tau_{|\rho| - 1} - \tau_i < \sup(I)$ and $\rho, k \modelsfw \varphi_1$ for all $k$, $i < k < |\rho|$
	such that $\tau_{k} - \tau_i > c$ and $\tau_{|\rho| - 1} - \tau_{k} \geq a - c$ where $a = \inf(I)$
	\end{itemize}

\item $\rho, i \modelsfw \varphi_1 \gsince^c_I \varphi_2$ iff there exists $j$, $0 \leq j < i$ such that
$\rho, j \modelsfw \varphi_2$, $\tau_i - \tau_j \in I$,
and $\rho, k \modelsfw \varphi_1$ for all $k$, $j < k < i$
such that $\tau_i - \tau_{k} > c$ and $\tau_{k} - \tau_{j} > (a - c)$ where $a = \inf(I)$.
\end{itemize} \end{defi}

\begin{prop}\label{prop:stontow}
For a finite timed word $\rho$, a position $i$ in $\rho$ and an \emph{\mtlpastg{}} formula
$\varphi$,
\[
\rho, i \modelsfs \varphi \implies \rho, i \modelsfn \varphi \text{
     and } \rho, i \modelsfn \varphi \implies \rho, i \modelsfw
\varphi \,.
\]
\end{prop}

\paragraph{\emph{Informative prefixes}}
We say that $\rho$ is \emph{informative} for $\varphi$ if either
of the following holds:

\begin{itemize}
\item $\rho$ strongly satisfies $\varphi$, i.e.,~$\rho \modelsfs \varphi$. In this case we say that $\rho$ is an \emph{informative good prefix} for $\varphi$; or
\item $\rho$ fails to weakly satisfy $\varphi$, i.e.,~$\rho \centernot \modelsfw \varphi$. In this case we say that $\rho$ is an \emph{informative bad prefix} for $\varphi$.\footnote{Note that informative
good/bad prefixes are under-approximations of good/bad prefixes; see Section~\ref{sec:conclusion} for a discussion.}
\end{itemize}
The following proposition follows immediately from the definition of
informative prefixes. In words, negating (syntactically) a formula swaps its
set of informative good prefixes and informative bad prefixes.
\begin{prop}\label{prop:negxchggoodbad}
For an \emph{\mtlpastg{}} formula, a finite timed word $\rho$ is an informative good prefix for $\varphi$
if and only if $\rho$ is an informative bad prefix for $\neg \varphi$.
\end{prop}

     \begin{exa}\label{ex:pathologicallysafe} Consider the
     following formula over $\{P\}$: \[ \varphi = \eventually
     \globally (\neg P) \wedge \globally(P \implies
     \eventually_{< 3} P) \,.  \] We say that the finite timed word $\rho =
     (\{P\}, 0)(\{P\}, 2)(\emptyset, 5.5)$ is an informative bad
     prefix for $\varphi$ as the second conjunct has been `clearly' violated,
		i.e.,~there is a $P$-event with no $P$-event in the following three time units
     ($\rho \modelsfs \neg \varphi$, or equivalently $\rho \centernot \modelsfw \varphi$).
	  On the other hand, while $\rho' = (\{P\}, 0)(\{P\},
     2)(\{P\}, 4)$ is indeed a bad prefix for $\varphi$, it is not informative
	 as both the first and second conjuncts are not yet `clearly' violated
	($\rho' \centernot \modelsfs \neg \varphi$, or equivalently $\rho' \modelsfw \varphi$).
\end{exa}

     \begin{exa}\label{ex:intentionallysafe} Consider the
     following formula over $\{P\}$: \[ \varphi' = \globally (\neg P)
\wedge
\globally(P \implies \eventually_{<3} P) \,.  \]
 	 This formula is equivalent to the formula $\varphi$ in the
previous example. However, all the bad prefixes $\rho$ for $\varphi'$ are
informative ($\rho \modelsfs \neg \varphi$, or equivalently $\rho \centernot \modelsfw \varphi$).
\end{exa}

\subsection{Offline trace-checking algorithm}

\emph{Trace checking} can be seen as a much more restricted case of
model checking where one is only concerned with a single finite trace.
Formally, the trace-checking problem for \mtlpastg{} asks the following:
given a finite trace $\rho$ and an \emph{\mtlpastg{}} formula $\varphi$, is $\rho \modelsfn \varphi$?  An offline algorithm
for the problem is shown as
Algorithms~\ref{alg:pathcheckinguntil}~and~\ref{alg:pathcheckingguntil}.
For given $\rho$ and $\varphi$, the algorithm maintains a two-dimensional Boolean array \texttt{table} of $|\psi|$ rows and $|\rho|$ columns.
Each row is used to store the truth values of a subformula at all positions.
The algorithm proceeds by filling up the array \texttt{table} in a bottom-up manner, starting from minimal subformulas.
We only detail the cases for $\varphi_1 \until_I \varphi_2$ and $\varphi_1 \guntil_I^c \varphi_2$
as other cases are either symmetric or trivial.
In what follows, we write $x \leq I$ for $x < \sup(I)$ if $I$ is right-open and for $x \leq \sup(I)$ otherwise.
To ease the presentation we omit the array-bounds checks, e.g., the algorithm
should stop when $\mathit{ptr}$ ($\mathit{ptr1}$) reaches $-1$.

\begin{algorithm}[ht]
\caption{$\textsc{FillTable}(\texttt{table}, \varphi_1 \until_I \varphi_2)$}%
\label{alg:pathcheckinguntil}
\begin{algorithmic}[1]
	\State $\mathit{ptr} \gets |\rho| - 1$
	\For{$j = |\rho| - 1$ to $0$}

		\If{$\mathit{ptr} = j$}
			\State $\texttt{table}[\varphi_1 \until_I \varphi_2][\mathit{ptr}] \gets \bot$
			\State $\mathit{ptr} \gets \mathit{ptr} - 1$
		\EndIf

		\If{$\texttt{table}[\varphi_2][j]$}
			\If{$\mathit{ptr} = j - 1$}
				\State $\texttt{table}[\varphi_1 \until_I \varphi_2][\mathit{ptr}] \gets (\tau_j - \tau_{\mathit{ptr}} \in I)$
				\State $\mathit{ptr} \gets \mathit{ptr} - 1$
			\EndIf
			\While{$\texttt{table}[\varphi_1][\mathit{ptr} + 1] \wedge \tau_j - \tau_{\mathit{ptr}} \leq I$}
				\If{$\tau_j - \tau_{\mathit{ptr}} \in I$}
					\State $\texttt{table}[\varphi_1 \until_I \varphi_2][\mathit{ptr}] \gets \top$
				\Else
					\State $\texttt{table}[\varphi_1 \until_I \varphi_2][\mathit{ptr}] \gets \bot$
				\EndIf
				\State $\mathit{ptr} \gets \mathit{ptr} - 1$
			\EndWhile
		\EndIf

	\EndFor
\end{algorithmic}
\end{algorithm}

\begin{prop}
After executing $\textsc{FillTable}(\emph{\texttt{table}}, \varphi_1 \until_I \varphi_2)$, we have
\[
\emph{\texttt{table}}[\varphi_1 \until_I \varphi_2][i] \iff \rho, i \modelsfn \varphi_1 \until_I \varphi_2
\]
for all $0 \leq i < |\rho|$ if $\emph{\texttt{table}}[\varphi_1]$ and $\emph{\texttt{table}}[\varphi_2]$
were both correct.
\end{prop}
\begin{proof}
Suppose that $\texttt{table}[\varphi_1 \until_I \varphi_2][i] = \top$.
Since each entry in $\texttt{table}[\varphi_1 \until_I \varphi_2]$ is filled exactly once,
it must be filled at either line $8$ or line $12$. In the former case it is clear that
$\rho, i \modelsfn \varphi_1 \until_I \varphi_2$.
In the latter case we must have
$\mathit{ptr} \leq j - 2$. If $\mathit{ptr} = j - 2$ then we are done, so we assume $\mathit{ptr} < j - 2$.
If there is a maximal position $\mathit{ptr}'$, $\mathit{ptr} + 1 < \mathit{ptr}' < j$ such that
$\texttt{table}[\varphi_1][\mathit{ptr}'] = \bot$, we must have $\mathit{ptr} + 1 = \mathit{ptr}'$, which
is a contradiction. We therefore conclude that $\rho, i \modelsfn \varphi_1 \until_I \varphi_2$.

For the other direction, assume $\rho, i \modelsfn \varphi_1 \until_I \varphi_2$ and let $j' > i$ be the witness
position, i.e.,~$\tau_{j'} - \tau_i \in I$, $\texttt{table}[\varphi_2][j'] = \top$ and $\texttt{table}[\varphi_1][j''] = \top$
for all $j''$, $i < j'' < j'$. Now consider $j = j'$. If $\mathit{ptr} \geq i$ then we are done.
So we assume $\mathit{ptr} < i$. If we already have $\texttt{table}[\varphi_1 \until_I \varphi_2][i] = \bot$,
then it must be the case that $\tau_{j'} - \tau_i \notin I$, which is a contradiction.
Therefore we must have $\texttt{table}[\varphi_1 \until_I \varphi_2][i] = \top$.
\end{proof}

\begin{algorithm}[ht]
\caption{$\textsc{FillTable}(\texttt{table}, \varphi_1 \guntil_I^c \varphi_2)$}%
\label{alg:pathcheckingguntil}
\begin{algorithmic}[1]
	\State $\mathit{ptr1}$, $\mathit{ptr2} \gets |\rho| - 1$
	\For{$j = |\rho| - 1$ to $0$}

		\While{$\tau_j - \tau_{\mathit{ptr2}} \leq \inf(I) - c \vee \texttt{table}[\varphi_1][\mathit{ptr2}]$}
			\State $\mathit{ptr2} \gets \mathit{ptr2} - 1$
		\EndWhile

		\If{$\mathit{ptr1} = j$}
			\State $\texttt{table}[\varphi_1 \guntil_I^c \varphi_2][\mathit{ptr1}] \gets \bot$
			\State $\mathit{ptr1} \gets \mathit{ptr1} - 1$
		\EndIf

		\If{$\texttt{table}[\varphi_2][j]$}
			\If{$\mathit{ptr1} = j - 1$}
				\State $\texttt{table}[\varphi_1 \guntil_I^c \varphi_2][\mathit{ptr1}] \gets (\tau_j - \tau_{\mathit{ptr1}} \in I)$
				\State $\mathit{ptr1} \gets \mathit{ptr1} - 1$
			\EndIf
			\While{$\tau_j - \tau_{\mathit{ptr1}} \leq I \wedge \tau_{\mathit{ptr2}} - \tau_{\mathit{ptr1}} \leq c$}
				\If{$\tau_j - \tau_{\mathit{ptr1}} \in I$}
					\State $\texttt{table}[\varphi_1 \guntil_I^c \varphi_2][\mathit{ptr1}] \gets \top$
				\Else
					\State $\texttt{table}[\varphi_1 \guntil_I^c \varphi_2][\mathit{ptr1}] \gets \bot$
				\EndIf
				\State $\mathit{ptr1} \gets \mathit{ptr1} - 1$
			\EndWhile
		\EndIf

	\EndFor
\end{algorithmic}
\end{algorithm}

\begin{prop}
After executing $\textsc{FillTable}(\emph{\texttt{table}}, \varphi_1 \guntil_I^c \varphi_2)$, we have
\[
\emph{\texttt{table}}[\varphi_1 \guntil_I^c \varphi_2][i] \iff \rho, i \modelsfn \varphi_1 \guntil_I^c \varphi_2
\]
for all $0 \leq i < |\rho|$ if $\emph{\texttt{table}}[\varphi_1]$ and $\emph{\texttt{table}}[\varphi_2]$
were both correct.
\end{prop}
\begin{proof}
Observe that after line $5$, $\mathit{ptr2}$ is equal to the maximal position
such that $\tau_j - \tau_{\mathit{ptr2}} > \inf(I) - c$ and $\texttt{table}[\varphi_1][\mathit{ptr2}] = \bot$.
The proof is very similar to the case of $\varphi_1 \until_I \varphi_2$.
\end{proof}

\subsection{Monitoring procedure}\label{subsec:monitoring}

Conceptually, we can regard a monitor as a \emph{deterministic} automaton
over finite traces.
The monitoring process, then, can be carried out by simply moving a token as directed by
the prefix.
However, it is well-known that in a dense real-time setting, such a monitor
(say, which accepts all the bad prefixes for $\varphi$)
needs an unbounded number of clocks and therefore cannot be realised
in practice~\cite{Alur1992, Maler2005, Reynolds2014}.
For this reason, we shall from now on assume that all input traces have
variability at most $k_{\mathit{var}}$, i.e.,~there are
at most $k_{\mathit{var}}$ events in any (open) unit time interval.
Based on this assumption, we give a monitoring procedure
for \mtlpastg{} formulas of the form
\[
\hat{\varphi} = \Phi(\psi_1, \ldots, \psi_m)
\]
where $\psi_1$, $\ldots$, $\psi_m$ are \emph{bounded} \mtlpastg{} formulas and $\Phi$ is an
\ltlpast{} formula.
The main idea is similar to the one used in the previous section:
we introduce new propositions $\mathbf{Q} = \{Q_1, \ldots, Q_m\}$
that correspond to $\psi_1$, $\ldots$, $\psi_m$.
In this way, we can monitor $\Phi$ as an \ltlpast{} property over $\mathbf{Q}$.\footnote{A similar idea is used in~\cite{Finkbeiner2009} to synthesise smaller monitor circuits for \ltl{} formulas.}
Since these propositions correspond to bounded formulas,
their truth values can be obtained by running the trace-checking algorithm on subtraces:
as the input trace has variability at most $k_{\mathit{var}}$,
we only need to store a `sliding window' of a certain size.

\paragraph{\emph{The untimed \ltlpast{} part}}

We recall briefly the standard methodology to construct finite automata that accept
exactly the informative good/bad prefixes for a given \ltl{} formula~\cite{Kupferman2001a}.
Given such a formula $\Psi$, first use a standard construction~\cite{Vardi1996} to
obtain a language-equivalent alternating B\"uchi automaton
$\mathcal{A}_\Psi$. Then redefine its accepting set to be
the empty set and treat it as an automaton over finite words; the
resulting automaton $\mathcal{A}_\Psi^\mathit{true}$ accepts exactly all
informative good prefixes for $\Psi$. In particular, one can determinise
$\mathcal{A}_\Psi^\mathit{true}$ with the usual subset
construction. The same can be done for $\neg \Psi$ to obtain a
deterministic automaton that accepts exactly the informative bad prefixes for $\Psi$.

In our case, we first translate the \ltlpast{} formulas $\Phi$ and $\neg
\Phi$ into a pair of \emph{two-way} alternating B\"uchi automata~\cite{Gastin2003}.
With the same modifications, we obtain two automata
that accept informative good prefixes and informative bad prefixes for $\Phi$.
We then apply existing procedures that translate two-way alternating automata over
finite words into deterministic automata (e.g.,~\cite{Birget1993}) and obtain
$\mathcal{D}_{\mathit{good}}$ and $\mathcal{D}_\mathit{bad}$, respectively.
To detect both types of prefixes simultaneously,
we will execute $\mathcal{D}_\mathit{good}$ and $\mathcal{D}_\mathit{bad}$ in parallel.

\begin{prop}
For an \emph{\mtlpastg{}} formula $\hat{\varphi}$ of the form described above, the automata $\mathcal{D}_\mathit{good}$
and $\mathcal{D}_\mathit{bad}$ are of size $2^{2^{O(|\Phi|)}}$ where $\Phi$ is
the `backbone' \emph{\ltlpast{}} formula.
\end{prop}

\paragraph{\emph{Na\"{\i}ve evaluation of the bounded metric parts}}

In what follows, let $l_{\mathit{fr}}(\psi) = k_{\mathit{var}} \cdot \lceil \mathit{fr}(\psi) \rceil$
and $l_{\mathit{pr}}(\psi) = k_{\mathit{var}} \cdot \lceil \mathit{pr}(\psi) \rceil$
(the functions $\mathit{fr}$ and $\mathit{pr}$ are defined in Section~\ref{sec:unboundedexpcomp}).
Suppose that we want to obtain the truth value of $\psi_i$ at
position $j$ in the input trace $\rho = (\sigma, \tau)$. Since $\psi_i$ is bounded, only the events occurring between $\tau_j -
\mathit{pr}(\psi_i)$ and $\tau_j + \mathit{fr}(\psi_i)$ can affect the truth value of
$\psi_i$ at $j$. This implies that $\rho, j \models \psi_i
\iff \rho', j \modelsfn \psi_i$ where $\rho'$ is a
prefix of $\rho$ that contains all the events between $\tau_j -
\mathit{pr}(\psi_i)$ and $\tau_j + \mathit{fr}(\psi_i)$ in $\rho$. Since $\rho$ is of bounded
variability $k_{\mathit{var}}$, there can be at most $l_\mathit{pr}(\psi_i) + 1 + l_\mathit{fr}(\psi_i)$
events between $\tau_j - pr(\psi_i)$ and $\tau_j + fr(\psi_i)$. It
follows that we can simply `record' all events in this interval
with a two-dimensional array of $l_\mathit{pr}(\psi_i) + 1 + l_\mathit{fr}(\psi_i)$ columns and
$1 + |\psi_i|$ rows: the first row is used to store the timestamps whereas the other rows are used to store the truth values.
Intuitively, the two-dimensional array acts as a sliding window around position $j$ in $\rho$.
Now consider all the propositions in $\mathbf{Q}$: their truth values
at position $j$ can be evaluated using a two-dimensional array
of $l_\mathit{pr}^\mathbf{Q} + 1 + l_\mathit{fr}^\mathbf{Q}$ columns and $1 + |\psi_1| + \cdots + |\psi_m|$ rows where $l_\mathit{pr}^\mathbf{Q} =
\displaystyle{\max_{1 \leq i \leq m} {l_\mathit{pr}(\psi_i)}}$ and $l_\mathit{fr}^\mathbf{Q} = \displaystyle{\max_{1 \leq i \leq m}
{l_\mathit{fr}(\psi_i)}}$. Each row can be filled in
time $O(l_\mathit{pr}^\mathbf{Q} + 1 + l_\mathit{fr}^\mathbf{Q})$ with the trace-checking algorithm.
Overall, we need a two-dimensional array of size $O(k_{\mathit{var}} \cdot c_{\mathit{sum}} \cdot |\hat{\varphi}|)$ where $c_{\mathit{sum}}$ is the sum of the constants in
$\hat{\varphi}$; for each position $j$, we need time $O(k_{\mathit{var}} \cdot
c_{\mathit{sum}} \cdot |\hat{\varphi}|)$ to obtain the truth values of all
propositions in $\mathbf{Q}$, which are then used as input to $\mathcal{D}_\mathit{good}$ and $\mathcal{D}_\mathit{bad}$.

\paragraph{\emph{Incremental evaluation of the bounded metric parts}}
While the procedure above uses only bounded space, it is clearly inefficient as for each $j$
we have to fill the whole two-dimensional array from scratch.
This is because some of the filled entries (other than those for position $j$)
may depend on the events outside of the sliding window, and thus can be incorrect.
We now describe an optimisation which enables the reuse of
previously filled entries.

We first deal with the simpler case of past subformulas.
Observe that as the trace-checking algorithm is filling a row for
$\varphi_1 \since_I \varphi_2$ or $\varphi_1 \gsince_I^c \varphi_2$,
the variables $\mathit{ptr}$, $\mathit{ptr1}$ and $\mathit{ptr2}$ all increases \emph{monotonically}.
This implies that for past subformulas, the trace-checking algorithm
can be used in an online manner: simply suspend the algorithm
when we have filled all entries using the truth values of $\varphi_1$ and $\varphi_2$ (at
various positions) that are currently known, and resume the algorithm when the
truth values of $\varphi_1$ and $\varphi_2$ (at some other positions)
that are previously unknown become available.

The case of future subformulas is more involved.
Suppose that we want to evaluate the truth value of a subformula $P_1 \until_{(a, b)} P_2$ at position $j$
in the input trace $\rho = (\sigma, \tau)$.
It is clear that the value may depend on future events
if $\tau_j + b$ is greater than the timestamp of the last acquired event.
However, observe that even when this is the case, we may still do the evaluation if any of the following holds:
\begin{itemize}
\item $P_1$ fails to hold at some position $j'$ such that
$\tau_{j'}$ is less or equal than the timestamp of the last acquired event.
In this case, we know that all the truth values of $P_1 \until_{(a, b)} P_2$ at positions $< j'$
cannot depend on the events at positions $> j'$.
\item $P_2$ holds at some position $j' > j$ and $P_1$ holds at all positions $j''$, $j < j'' < j'$.
In this case, the truth values of $P_1 \until_{(a, b)} P_2$ at positions $k < j'$
where $\tau_{j'} - \tau_{k} \in (a, b)$ are $\top$ and
do not depend on the events at positions $> j'$.
\end{itemize}

\noindent
We generalise this observation to handle the general case of
updating the row for $\varphi_1 \until_{(a, b)} \varphi_2$.
First of all, we maintain indices $j_{\varphi_1}$, $j_{\varphi_2}$, $j_{\varphi_1 \until_I \varphi_2}$
which point to the first unknown entries in the rows for
$\varphi_1$, $\varphi_2$ and $\varphi_1 \until_I \varphi_2$.
Let $t_{\maxit} = \min \{ \tau_{(j_{\varphi_1} - 1)}, \tau_{(j_{\varphi_2} - 1)} \}$ and update
its value when either $j_{\varphi_1}$ or $j_{\varphi_2}$ changes.
Whenever $t_{\maxit}$ is updated to a new value,
we also update the following indices:
\begin{itemize}
\item $j_1$ is the maximal position such that $\tau_{j_1} + b \leq t_\maxit$
\item $j_2$ is the maximal position such that $\tau_{j_2} \leq t_\maxit$ and $\varphi_2$ holds at $j_2$
\item $j_3$ is the maximal position such that $\tau_{j_3} + a < \tau_{j_2}$
\item $j_4$ is the maximal position such that $\tau_{j_4} \leq t_\maxit$ and $\varphi_1$ does not hold at $j_4$.
\end{itemize}
Now, after both the rows for $\varphi_1$ and $\varphi_2$ have been updated, if any of $j_1$, $j_3$, $j_4 - 1$ is greater or equal than $j_{\varphi_1 \until_I \varphi_2}$,
we let $j_5 = \max\{j_1, j_3, j_4 - 1\}$ and start Algorithm~\ref{alg:pathcheckinguntil}
from line $3$ with $\mathit{ptr} = j_5$ and $j = j_2$. We run the algorithm
till all the entries at positions $\leq j_5$ in the row for $\varphi_1 \until_I \varphi_2$
have been filled.
The crucial observation here is that $j_1$, $j_2$, $j_3$, $j_4$ all increase monotonically,
and therefore can be maintained in amortised linear time.
Also, the truth value of any subformula at any position will be filled only once.
The case of $\varphi_1 \guntil_{(a, b)}^c \varphi_2$ is similar (but slightly more involved).
These observations imply that each entry in the two-dimensional array can be filled in
amortised constant time.
Assuming that moving a token on a deterministic finite automaton takes constant time,
we can state the following theorem.

\begin{thm}
For an \emph{\mtlpastg{}} formula $\hat{\varphi}$ of the form described earlier and an infinite
trace of variability $k_\mathit{var}$, our monitoring procedure
uses two DFAs of size $2^{2^{O(|\Phi|)}}$, a two-dimensional array of size $O(k_\mathit{var} \cdot c_\mathit{sum} \cdot
|\hat{\varphi}|)$ where $c_\mathit{sum}$ is the sum of the constants in
$\hat{\varphi}$, and amortised time $O(|\hat{\varphi}|)$ per event.
\end{thm}

\paragraph{\emph{Correctness}}

We now show that our procedure is sound and complete for detecting informative prefixes.
\begin{prop}\label{prop:boundedpartinformative} For a bounded
\emph{\mtlpastg{}} formula $\psi$, a finite trace $\rho = (\sigma, \tau)$ and a
position $0 \leq i < |\rho|$ such that $\tau_i + \mathit{fr}(\psi) \leq
\tau_{|\rho| - 1}$, we have \[ \rho, i
\modelsfs \psi \iff \rho, i \modelsfn \psi \iff
\rho, i \modelsfw \psi \,. \] \end{prop}

\begin{prop}\label{prop:lengthen}
For an \emph{\mtlpastg{}} formula $\varphi$, a finite trace $\rho$ and a position $i$ in $\rho$,
if $\rho$ is a prefix of a longer finite trace $\rho'$, then
\[
\rho, i \modelsfs \varphi \implies \rho', i \modelsfs \varphi
\text{ and } \rho, i \centernot \modelsfw \varphi \implies \rho', i
\centernot \modelsfw \varphi \,.
\]
\end{prop}

\begin{thm}[Soundness]\label{thm:soundness}
In our procedure, if we ever reach an accepting state of
$\mathcal{D}_\mathit{good}$ ($\mathcal{D}_\mathit{bad}$) via a finite word $u \in
\Sigma_Q^*$, then we must eventually read an
informative good (bad) prefix for $\hat{\varphi}$. \end{thm}
\begin{proof} For such $u$ and a corresponding $\rho = (\sigma, \tau)$
such that $\tau_{|u| - 1} + l_\mathit{fr}^Q \leq \tau_{|\rho| - 1}$, we have
\[ \forall i \in \ropen{0, |u|} \, \big( (u, i \modelsfs \Psi \implies \rho, i \modelsfs \psi)
\wedge
(u, i \centernot
\modelsfw \Psi \implies \rho, i \centernot \modelsfw \psi) \big)
 \] where $\Psi$ is a subformula of $\Phi$ and $\psi =
\Psi(\psi_1, \ldots, \psi_m)$. This can easily be proved by
structural induction and Proposition~\ref{prop:boundedpartinformative}. If $u$ is accepted by $\mathcal{D}_\mathit{good}$, we
have $u, 0 \modelsfs \Phi$ by construction. By the above we have $\rho, 0
\modelsfs \Phi(\psi_1, \ldots, \psi_m)$, as desired. The case of
$\mathcal{D}_\mathit{bad}$ is symmetric.
\end{proof}

\begin{thm}[Completeness]\label{thm:completeness}
Whenever we read an informative good (bad) prefix $\rho = (\sigma,
\tau)$ for $\hat{\varphi}$, $\mathcal{D}_\mathit{good}$ ($\mathcal{D}_\mathit{bad}$) must
eventually reach an accepting state. \end{thm}
\begin{proof}
For the finite word $u' \in \Sigma_Q^*$ obtained a bit later with $|u'| = |\rho|$,
\[
\forall i \in \ropen{0, |u'|} \, \big( (\rho, i \modelsfs \psi
\implies u', i \modelsfs \Psi) \wedge (\rho, i \centernot
\modelsfw \psi \implies u', i \centernot \modelsfw \Psi)
\big) \] where $\Psi$ is a subformula of $\Phi$ and $\psi =
\Psi(\psi_1, \ldots, \psi_m)$. This can be proved by
structural induction and Proposition~\ref{prop:lengthen}. The theorem follows.
\end{proof}

\subsection{Preservation of informative prefixes}

As we have seen earlier in Example~\ref{ex:pathologicallysafe} and~\ref{ex:intentionallysafe},
it is possible for two equivalent \mtlpastg{} formulas to have different sets of informative good/bad prefixes.
In this section, we show that this is cannot be the case when the two formulas
are related by one of the rewriting rules in Section~\ref{sec:separation}.
In other words, the rewriting rules in
Section~\ref{sec:separation} preserves the `informativeness' of
formulas.

\begin{lem}\label{lem:informativenesspreservation}
For an \emph{\mtlpastg{}} formula $\varphi$, let $\varphi'$ be the formula obtained from $\varphi$
by applying one of the rules in Section~\ref{sec:separation} on
some of its subformula. We have \[
\rho \modelsfs \varphi \iff \rho \modelsfs \varphi'
\text{ and } \rho \modelsfw \varphi \iff \rho \modelsfw
\varphi' \,. \]
\end{lem}
Given the lemma above, we can state the following theorem on any \mtlpast{} formula
$\varphi$ and the equivalent formula $\hat{\varphi}$ (of our desired form)
obtained from $\varphi$ by applying the rewriting rules in Section~\ref{sec:separation}.
\begin{thm}\label{thm:coincide}
The set of informative good prefixes of $\varphi$ coincides with the
set of informative good prefixes of $\hat{\varphi}$. The same holds for
informative bad prefixes. \end{thm}

We now have a way to detect the informative good/bad prefixes for an arbitrary
\mtlpastg{} formula $\varphi$: use the rewriting rules to obtain $\hat{\varphi}$, and apply
the monitoring procedure we described in the last subsection.
The monitor only needs a bounded amount of memory,
even for complicated \mtlpastg{} formulas with arbitrary nestings of (bounded and unbounded) past and future operators.
\begin{proof}[Proof of Lemma~\ref{lem:informativenesspreservation}]

Since the satisfaction relations are defined inductively, we can work directly on the relevant subformulas.
We would like to prove that for a finite timed word $\rho$ and a position $i$ in $\rho$,
\[
\rho, i \modelsfs \phi \iff \rho, i \modelsfs \phi' \text{ and } \rho, i \modelsfw \phi \iff \rho, i \modelsfw \phi'
\]
where $\phi \iff \phi'$ matches one of the rules in Section~\ref{sec:separation}.
For a group of similar rules we will only prove a representative one, as the proof for others follow similarly. In the following let the LHS be $\phi$ and RHS be $\phi'$.

\begin{itemize}[itemsep=0.7em,topsep=0.7em]
\item $\varphi_1 \until_{(a, \infty)} \varphi_2  \iff  \varphi_1 \until \varphi_2 \wedge \globally_{\lopen{0, a}} (\varphi_1 \wedge \varphi_1 \until \varphi_2)$:
	\begin{itemize}[itemsep=0.5em,topsep=0.5em]
	\item $\rho, i \modelsfs \phi \iff \rho, i \modelsfs \phi'$:

	Assume $\rho, i \modelsfs \phi$. By definition we have $\rho, i \modelsfs \varphi_1 \until \varphi_2$.
	If there is no event in $\lopen{\tau_i, \tau_i + a}$, since there must be an event in $\lopen{\tau_i + a, \tau_{|\rho| - 1}}$, we are done.
	If there are events in $\lopen{\tau_i, \tau_i + a}$, then for all $j$ such that $\tau_j - \tau_i \in \lopen{0, a}$ we have $\rho, j \centernot \modelsfw \neg \varphi_1$.
	Also for all such $j$ we have $\rho, j \centernot \modelsfw \neg \varphi_1 \until \varphi_2$ since it is obvious that $\rho, j \modelsfs \varphi_1 \until \varphi_2$.
	For the other direction, if the witness (for $\rho, i \modelsfs \varphi_1 \until \varphi_2$) is in $(\tau_i + a, \tau_{|\rho| - 1})$ then we are done. If this is
	not the case, since $\rho, i \centernot \modelsfw \eventually_{\lopen{0, a}} \big( \neg \varphi_1 \vee \neg (\varphi_1 \until \varphi_2) \big)$, we must have $\tau_{|\rho| - 1} \geq a$.
	Now for all $j$ such that $\tau_j - \tau_i \in \lopen{0, a}$ we have $\rho, j \modelsfs \varphi_1$ and $\rho, j \modelsfs \varphi_1 \until \varphi_2$, which
	imply $\rho, i \modelsfs \phi$.

	\item $\rho, i \modelsfw \phi \iff \rho, i \modelsfw \phi'$:

	Assume $\rho, i \modelsfw \phi$. This holds if there is a witness in $(a, \infty)$ or $\rho, i \modelsfw \globally \varphi_1$.
	In both cases we have $\rho, i \modelsfw \varphi_1 \until \varphi_2$. If there is no event in $\lopen{\tau_i, \tau_i + a}$ then we are done.
	If there is a witness, then for all such $j$ that $\tau_j - \tau_i \in \lopen{0, a}$ we have $\rho, j \modelsfw \varphi_1$ and $\rho, j \modelsfw \varphi_1 \until \varphi_2$.
	If there is no witness then for all such $j$ we again have $\rho, j \modelsfw \varphi_1$ and $\rho, j \modelsfw \varphi_1 \until \varphi_2$.
	For the other direction, if there is no event in $\lopen{\tau_i, \tau_i + a}$ we are done. If there are events in $\lopen{\tau_i, \tau_i + a}$, all $j$
	such that $\tau_j - \tau_i \in \lopen{0, a}$ will satisfy $\rho, j \modelsfw \varphi_1$ and $\rho, j \modelsfw \varphi_1 \until \varphi_2$.
	This clearly gives $\rho, i \modelsfw \phi$.
	\end{itemize}

\item $\varphi_1 \guntil^c_{(a, \infty)} \varphi_2  \iff  \varphi_1 \guntil^c_{\lopen{a, 2a}} \varphi_2 \vee \Big( \eventually^w_{[0, c]} \big(\varphi_1 \until_{(a, \infty)} (\varphi_2 \vee \eventually_{\leq a - c} \varphi_2) \big) \Big)$:

    \noindent
	The proof is very similar to the proof of Proposition~\ref{prop:removeunboundedguntil}.

\item $\neg (\varphi_1 \until \varphi_2) \iff \globally \neg \varphi_2 \vee \big( \neg \varphi_2 \until (\neg \varphi_2 \wedge \neg \varphi_1) \big)$:
	\begin{itemize}[itemsep=0.5em,topsep=0.5em]
	\item $\rho, i \modelsfs \phi \iff \rho, i \modelsfs \phi'$:

	Assume $\rho, i \modelsfs \phi \iff \rho, i \centernot \modelsfw \varphi_1 \until \varphi_2$.
	This implies that $\varphi_1$ fails to hold before $\varphi_2$ holds, and we have $\rho, i \modelsfs \neg \varphi_2 \until (\neg \varphi_2 \wedge \neg \varphi_1)$.
	For the other direction note that $\rho, i \centernot \modelsfs \globally \neg \varphi_2$, the second disjunct must be satisfied,
	and it is easy to see that $\rho, i \modelsfs \phi$.

	\item $\rho, i \modelsfw \phi \iff \rho, i \modelsfw \phi'$:

	Assume $\rho, i \modelsfw \neg (\varphi_1 \until \varphi_2) \iff \rho, i \centernot \modelsfs \varphi_1 \until \varphi_2$.
	This implies either $\rho, j \centernot \modelsfs \varphi_2 \iff \rho, j \modelsfw \neg \varphi_2$ for all $j > i$ in $\rho$
	(this gives $\rho, i \modelsfw \globally \neg \varphi_2$) or $\varphi_1$ fails to hold before $\varphi_2$ holds---$\rho, i \modelsfw \neg \varphi_2 \until (\neg \varphi_2 \wedge \neg \varphi_1)$.
	For the other direction, if $\rho, i \modelsfw \globally \neg \varphi_2 \iff \rho, i \centernot \modelsfs \eventually \varphi_2$
	then $\rho, i \modelsfs \varphi_1 \until \varphi_2$ cannot hold. If $\rho, i \modelsfw \neg \varphi_2 \until (\neg \varphi_2 \wedge \neg \varphi_1)$ then
	either $\rho, i \modelsfw \globally \neg \varphi_2$ or there is a witness, and it is easy to see that $\rho, i \modelsfs \varphi_1 \until \varphi_2$ cannot hold.
	\end{itemize}

\item $\theta \until_{(a, b)} \big( (\varphi_1 \until \varphi_2) \wedge \chi \big) \iff \theta \until_{(a, b)} \big( (\varphi_1 \until_{(0, 2b)} \varphi_2) \wedge \chi \big) \\
\vee \Big( \big( \theta \until_{(a, b)} (\globally_{(0, 2b)} \varphi_1 \wedge \chi) \big) \wedge \varphi_{ugb} \Big)$:

	The proof is very similar to the proof of Proposition~\ref{prop:extract}.

\item $\big( (\varphi_1 \until \varphi_2) \vee \chi \big) \until_{(a, b)} \theta \iff \big( (\varphi_1 \until_{(0, 2b)} \varphi_2) \vee \chi \big) \until_{(a, b)} \theta \\
\vee \bigg( \Big( \big( (\varphi_1 \until_{(0, 2b)} \varphi_2) \vee \chi \big) \until_{(0, b)} (\globally_{(0, 2b)} \varphi_1) \Big) \wedge \eventually_{(a, b)} \theta \wedge \varphi_{ugb} \bigg)$

	\begin{itemize}[itemsep=0.5em,topsep=0.5em]
	\item $\rho, i \modelsfs \phi \iff \rho, i \modelsfs \phi'$:

	Assume $\rho, i \modelsfs \phi$. It is obvious that $\rho, i \modelsfs \eventually_{(a, b)} \theta$ holds. If the first disjunct of $\phi'$ does not hold,
	then $\rho, i \modelsfs \big( (\varphi_1 \until_{(0, 2b)} \varphi_2) \vee \chi \big) \until_{(0, b)} (\globally_{(0, 2b)} \varphi_1)$ must hold.
	The last conjunct holds by an argument similar to the proof of Proposition~\ref{prop:extract}.
	For the other direction, if the first disjunct of $\phi'$ holds then we are done.
	If it does not hold, then there must be a witness (at which $\varphi_2$ holds)
	in $[\tau_i + 2b, \tau_{|\rho| - 1}]$, and it is easy to see that $\rho, i \modelsfs \phi$.

	\item $\rho, i \modelsfw \phi \iff \rho, i \modelsfw \phi'$:

	Assume $\rho, i \modelsfw \phi$. If the first disjunct of $\phi'$ does not hold then there must be events in $[\tau_i + 2b, \tau_{|\rho| - 1}]$.
	It follows that $\rho, i \modelsfw \big( (\varphi_1 \until_{(0, 2b)} \varphi_2) \vee \chi \big) \until_{(0, b)} (\globally_{(0, 2b)} \varphi_1)$ and
	$\rho, i \modelsfw \eventually_{(a, b)} \theta$ must hold. The rest is similar to the proof to Proposition~\ref{prop:extract}.
	For the other direction, if the first disjunct of $\phi'$ holds then we are done. Otherwise if $\tau_{|\rho| - 1} < b$, it is easy to see that
	$\rho, i \modelsfw \phi$. If this is not the case then the proof again follows Proposition~\ref{prop:extract}. \qedhere
	\end{itemize}

\end{itemize}
\end{proof}

\section{Conclusion and future work}\label{sec:conclusion}

\paragraph{\emph{Expressive completeness over bounded timed words}}
We showed that \mtlpast{} extended with our new modalities
`\emph{generalised Until}' and `\emph{generalised Since}' (\mtlpastg{}) is expressively complete for \foone{} over bounded timed words.
Moreover, the time-bounded satisfiability and model-checking problems for \mtlpastg{} remain $\mathrm{EXPSPACE}$-complete, same as that of \mtlpast{}.
The situation here is similar to \ltlpast{} over general,
possibly non-Dedekind complete, linear orders (e.g., the rationals):
in this case, \ltlpast{} can be made expressively complete (for \textmd{\textsf{FO[$<$]}}) by adding the Stavi modalities~\cite{Gabbay1994},
yet the complexity of the satisfiability problem remains $\mathrm{PSPACE}$-complete~\cite{Rabinovich2010}.
Along the way, we also obtained a strict hierarchy of metric temporal logics
based on their expressiveness over bounded timed words.

One drawback of the modalities $\guntil_I^c$ and $\gsince_I^c$ is that they are not very intuitive.
However, as we proved that simpler versions
of these modalities ($\first_I$ and $\pfirst_I$) are strictly less expressive,
we believe it is unlikely that any other expressively complete extension of \mtlpast{} could be much simpler than ours.

The satisfiability and model-checking procedures for \mtlpast{} over time-bounded signals in~\cite{Ouaknine2009}
are based on the satisfiability procedure for \ltlpast{} over signals in~\cite{Reynolds2010}.
While the satisfiability problem for \ltlpast{} remains $\mathrm{PSPACE}$-complete when interpreted over signals,
very few implementations are currently available~\cite{French2013}.
This is in sharp contrast with the discrete case where a number of mature, industrial-strength tools (e.g., SPIN~\cite{Holzmann1997}) are readily available.
Our results enable the direct application of these tools to
time-bounded verification.
Whether this yields efficiency gains in practice, however, can only be
evaluated empirically, which we leave as future work.

\paragraph{\emph{Expressive completeness over unbounded timed words}}
Building upon a previous work of Hunter, Ouaknine and Worrell~\cite{Hunter2012}, we showed that the \emph{rational} version of \mtlpastg{} is expressively complete for \fo{}
over infinite timed words. The result answers an implicit open question in a long line of research started in~\cite{Alur1990}
and further developed in~\cite{Bouyer2005, Prabhakar2006, DSouza2006a, Pandya2011}.

It is known that the \emph{integer} version of \mtlpast{} extended with counting modalities (and their past counterparts)
is expressively complete for \foone{} over the reals~\cite{Hunter2013}.\footnote{This result is stronger than~\cite{Hunter2012} as counting modalities (and their past counterparts) can be expressed in \mtlpast{} with rational endpoints.}
We conjecture that the analogous result holds in the pointwise semantics, i.e.,~the integer version of \mtlpastg{}
becomes expressively complete for \foone{} when we add counting modalities.
Adapting the proof in~\cite{Hunter2013} to the pointwise case, however, is not a straight-forward task.
In particular, the proof relies on the expressive completeness of \mitlpast{} with counting modalities
for \textmd{\textsf{Q2MLO}}~\cite{Hirshfeld2004},
a result that itself requires a highly non-trivial proof~\cite{Hirshfeld2006}
and seems to hold only in the continuous semantics.

Besides expressiveness, another major concern in the study of metric logics
is \emph{decidability}.
We intend to investigate whether the expressiveness of
\mitl{} or \mitlpast{} can be enhanced with the new modalities
while retaining decidability.
Specifically, we would like to answer the following question:
what is the complexity of the satisfiability problem for the logic obtained by
adding $\first_I$ (with non-singular $I$) into \mitl{}?
Since $\first_I$ can be expressed in one-clock alternating timed automata,
it can possibly be handled in the framework of~\cite{Brihaye2014}.
More generally, we may consider \mitlpast{} extended with $\first_I$ and $\pfirst_I$ (with non-singular $I$);
it is not clear whether allowing these modalities simultaneously
leads to undecidability.

\paragraph{\emph{Monitoring}}

We identified an `easy-to-monitor' fragment of \mtlpastg{}, for which we proposed
an efficient trace-length independent monitoring procedure.
This fragment is much more expressive than the fragments previously considered
in the literature.
Moreover, we showed that informative good/bad prefixes
are preserved by the syntactic rewriting rules in Section~\ref{sec:separation}.
It follows that the informative good/bad prefixes for an arbitrary \mtlpastg{} formula
can be monitored in a trace-length independent fashion,
thus overcoming a long-standing barrier to the runtime verification
of real-time properties.

For an arbitrary \mtlpastg{} formula, the syntactic rewriting process could potentially induce a non-elementary blow-up.
In practice, however, the resulting formula is often of
comparable size to the original one, which itself is typically small.
For example, consider the following formula:
\[
\globally \big( \texttt{ChangeGear} \implies \eventually_{(0, 30)} (\texttt{InjectFuel} \wedge \once \texttt{InjectAir}) \big) \,.
\]
The resulting formula after rewriting is
\[
    \arraycolsep=0.3ex
    \begin{array}{rcl}
    \globally \big( \texttt{ChangeGear} & \implies & \eventually_{(0, 30)} (\texttt{InjectFuel} \wedge \once_{(0, 30)} \texttt{InjectAir}) \\
    & & {} \vee ( \eventually_{(0, 30)} \texttt{InjectFuel} \wedge \once \texttt{InjectAir} ) \big) \,.
    \end{array}
\]
In fact, it can be argued that most common real-time specification patterns~\cite{Konrad2005} belong syntactically to our `easy-to-monitor' fragment
and thus need no rewriting.
Another way to alleviate the issue is to allow more liberal syntax
(or more derived operators).
For example, the procedure described in Section~\ref{subsec:monitoring}
can handle subformulas with unbounded past without modification.

To detect informative bad prefixes, our monitoring procedure uses a deterministic finite automaton
doubly-exponential in the size of the input formula.
While such a blow-up is rarely a problem in practice (see~\cite[Section $2.5$]{Bauer2011}),
it would be better if it could be avoided altogether.
In the untimed setting, it is known that if a safety property can be written as an \ltlpast{} formula,
then it is equivalent to a formula of the form $\globally \psi$
where $\psi$ is a past-only \ltlpast{} formula~\cite{Lichtenstein1985}.
So, if we restrict our attention to safety properties,
it suffices to consider formulas of this form,
for which there is an efficient monitoring procedure that
uses $O(| \psi |)$ time (per event) and $O(| \psi |)$ space~\cite{Havelund2001}.
Unfortunately, the question of whether a corresponding result holds for \mtlpast{} (or similar metric temporal logics) is still open.

Our procedure detects only \emph{informative} good/bad prefixes, which themselves can
be regarded as easily-checkable certificates for the fulfilment/violation of the property.
While we believe this limitation is in no way severe---in fact, the limitation
is implicit in almost all current approaches to monitoring real-time properties---there are certain
practical scenarios where detecting \emph{all} good/bad prefixes is preferred.
We could have used two deterministic finite automata that detect all good/bad prefixes
for the backbone \ltlpast{} formula, but still they cannot detect all good/bad prefixes
for the whole formula (consider Example~\ref{ex:pathologicallysafe}).
We leave as future work a procedure that detects all good/bad prefixes.

Finally, we remark that the offline trace-checking problem is of independent theoretical interest~\cite{Markey2003}.
It is known that the trace-checking problem for \ltlpast{}~\cite{Kuhtz2009} and \mtlpast{}~\cite{Bundala2014}
are both in $\mathrm{AC^1[\log \mathrm{DCFL}]}$, yet their precise complexity is still open.
It would be interesting to see whether the construction for \mtlpast{} in~\cite{Bundala2014} carries over to \mtlpastg{}.

\bibliographystyle{alpha}
\bibliography{refs}

\newcommand{\etalchar}[1]{$^{#1}$}
\begin{thebibliography}{BMOW07}

\bibitem[ABLS05]{Arafat2005}
Oliver Arafat, Andreas Bauer, Martin Leucker, and Christian Schallhart.
\newblock Runtime verification revisited.
\newblock Technical Report TUM-I0518, Technische Universit{\"a}t M{\"u}nchen,
  2005.

\bibitem[AD94]{Alur1994}
Rajeev Alur and David Dill.
\newblock A theory of timed automata.
\newblock {\em Theoretical Computer Science}, 126(2):183--235, 1994.

\bibitem[AFH96]{Alur1996}
Rajeev Alur, Tom\'{a}s Feder, and Thomas~A. Henzinger.
\newblock The benefits of relaxing punctuality.
\newblock {\em Journal of the ACM}, 43(1):116--146, 1996.

\bibitem[AH90]{Alur1990}
Rajeev Alur and Thomas~A. Henzinger.
\newblock Real-time logics: Complexity and expressiveness.
\newblock In {\em Proceedings of LICS 1990}, pages 390--401. IEEE Computer
  Society Press, 1990.

\bibitem[AH92]{Alur1992}
Rajeev Alur and Thomas~A. Henzinger.
\newblock Back to the future: towards a theory of timed regular languages.
\newblock In {\em Proceedings of FOCS 1992}, pages 177--186. IEEE Computer
  Society Press, 1992.

\bibitem[AM04]{Alur2004}
Rajeev Alur and Parthasarathy Madhusudan.
\newblock Decision problems for timed automata: {A} survey.
\newblock In {\em Formal Methods for the Design of Real-Time Systems}, volume
  3185 of {\em LNCS}, pages 1--24. Springer, 2004.

\bibitem[AS87]{Alpern1987}
Bowen Alpern and Fred~B. Schneider.
\newblock Recognizing safety and liveness.
\newblock {\em Distributed Computing}, 2(3):117--126, 1987.

\bibitem[BBBB09]{Baier2009}
Christel Baier, Nathalie Bertrand, Patricia Bouyer, and Thomas Brihaye.
\newblock When are timed automata determinizable?
\newblock In {\em Proceedings of ICALP 2009}, volume 5556 of {\em LNCS}, pages
  43--54. Springer, 2009.

\bibitem[BCE{\etalchar{+}}14]{Basin2014}
David~A. Basin, Germano Caronni, Sarah Ereth, Mat{\'u}s Harvan, Felix Klaedtke,
  and Heiko Mantel.
\newblock Scalable offline monitoring.
\newblock In {\em Proceedings of RV 2014}, volume 8734 of {\em LNCS}, pages
  31--47. Springer, 2014.

\bibitem[BCM05]{Bouyer2005}
Patricia Bouyer, Fabrice Chevalier, and Nicolas Markey.
\newblock {On the expressiveness of TPTL and MTL}.
\newblock In {\em Proceedings of FSTTCS 2005}, volume 3821 of {\em LNCS}, pages
  432--443. Springer, 2005.

\bibitem[BEG14]{Brihaye2014}
Thomas Brihaye, Morgane Esti{\'e}venart, and Gilles Geeraerts.
\newblock On {MITL} and alternating timed automata over infinite words.
\newblock In {\em Proceedings of FORMATS 2014}, volume 8711 of {\em LNCS},
  pages 69--84. Springer, 2014.

\bibitem[BGHM17]{Brihaye2017}
Thomas Brihaye, Gilles Geeraerts, Hsi-Ming Ho, and Benjamin Monmege.
\newblock {MightyL}: A compositional translation from {MITL} to timed automata.
\newblock In {\em Proceedings of CAV 2017}, volume 10426 of {\em LNCS}, pages
  421--440. Springer, 2017.

\bibitem[Bir93]{Birget1993}
Jean-Camille Birget.
\newblock State-complexity of finite-state devices, state compressibility and
  incompressibility.
\newblock {\em Mathematical Systems Theory}, 26(3):237--269, 1993.

\bibitem[BKV13]{Bauer2013}
A~Bauer, JC~K\"{u}ster, and Gil Vegliach.
\newblock From propositional to first-order monitoring.
\newblock In {\em Proceedings of RV 2013}, volume 8174 of {\em LNCS}, pages
  59--75. Springer, 2013.

\bibitem[BKZ11]{Basin2011}
David Basin, Felix Klaedtke, and Eugene Z\u{a}linescu.
\newblock Algorithms for monitoring real-time properties.
\newblock In {\em Proceedings of RV 2011}, volume 7186 of {\em LNCS}, pages
  260--275. Springer, 2011.

\bibitem[BLS11]{Bauer2011}
Andreas Bauer, Martin Leucker, and Christian Schallhart.
\newblock Runtime verification for {LTL} and {TLTL}.
\newblock {\em ACM Transactions on Software Engineering and Methodology},
  20(4):14, 2011.

\bibitem[BMOW07]{Bouyer2007}
Patricia Bouyer, Nicolas Markey, Jo{\"e}l Ouaknine, and James Worrell.
\newblock The cost of punctuality.
\newblock In {\em Proceedings of LICS 2007}, pages 109--120. IEEE Computer
  Society Press, 2007.

\bibitem[BN12]{Baldor2012}
Kevin Baldor and Jianwei Niu.
\newblock Monitoring dense-time, continuous-semantics, metric temporal logic.
\newblock In {\em Proceedings of RV 2012}, volume 7687 of {\em LNCS}, pages
  245--259. Springer, 2012.

\bibitem[BO14]{Bundala2014}
Daniel Bundala and Jo{\"e}l Ouaknine.
\newblock On the complexity of temporal-logic path checking.
\newblock In {\em Proceedings of ICALP 2014}, volume 8573 of {\em LNCS}, pages
  86--97. Springer, 2014.

\bibitem[CE81]{Clarke1981}
Edmund~M. Clarke and E.~Allen Emerson.
\newblock Design and synthesis of synchronization skeletons using
  branching-time temporal logic.
\newblock In {\em Proceedings of IBM Workshop on Logic of Programs}, volume 131
  of {\em LNCS}, pages 52--71. Springer-Verlag, 1981.

\bibitem[Dam94]{Dam1994}
Mads Dam.
\newblock Temporal logic, automata, and classical theories - an introduction,
  1994.

\bibitem[DHV07]{DSouza2007}
Deepak D'Souza, Raveendra Holla, and Deepak Vankadaru.
\newblock {On the expressiveness of TPTL in the pointwise and continuous
  semantics}.
\newblock Unpublished manuscript, 2007.

\bibitem[DKL10]{Dax2010a}
Christian Dax, Felix Klaedtke, and Martin Lange.
\newblock On regular temporal logics with past.
\newblock {\em Acta Informatica}, 47(4):251--277, 2010.

\bibitem[DM13]{DSouza2013}
Deepak D'Souza and Raj~Mohan Matteplackel.
\newblock {A clock-optimal hierarchical monitoring automaton construction for
  MITL}.
\newblock Technical Report 2013-1, Department of Computer Science and
  Automation, Indian Institute of Science, 2013.

\bibitem[DP06]{DSouza2006a}
Deepak D'Souza and Pavithra Prabhakar.
\newblock {On the expressiveness of MTL in the pointwise and continuous
  semantics}.
\newblock {\em International Journal on Software Tools for Technology
  Transfer}, 9(1):1--4, 2006.

\bibitem[DT04]{DSouza2004}
Deepak D'Souza and Nicolas Tabareau.
\newblock On timed automata with input-determined guards.
\newblock In {\em Proceedings of FORMATS/FTRTFT 2004}, volume 3253 of {\em
  LNCS}, pages 68--83. Springer, 2004.

\bibitem[EFHL03]{Eisner2003}
Cindy Eisner, Dana Fisman, John Havlicek, and Yoad Lustig.
\newblock Reasoning with temporal logic on truncated paths.
\newblock In {\em Proceedings of CAV 2003}, volume 2725 of {\em LNCS}, pages
  27--39. Springer, 2003.

\bibitem[EL86]{Emerson1986}
E.~Allen Emerson and Chin-Laung Lei.
\newblock Efficient model checking in fragments of the propositional
  mu-calculus.
\newblock In {\em Proceedings of LICS 1986}, pages 267--278. IEEE Computer
  Society Press, 1986.

\bibitem[EW96]{Etessami1996}
Kousha Etessami and Thomas Wilke.
\newblock An until hierarchy for temporal logic.
\newblock In {\em Proceedings of LICS 1996}, pages 108--117. IEEE Computer
  Society Press, 1996.

\bibitem[FK09]{Finkbeiner2009}
Bernd Finkbeiner and Lars Kuhtz.
\newblock {Monitor circuits for LTL with bounded and unbounded future}.
\newblock In {\em Proceedings of RV 2009}, volume 5779 of {\em LNCS}, pages
  60--75. Springer, 2009.

\bibitem[FMDR13]{French2013}
Tim French, John~Christopher McCabe-Dansted, and Mark Reynolds.
\newblock Verifying temporal properties in real models.
\newblock In {\em Proceedings of LPAR 2013}, volume 8312 of {\em LNCS}, pages
  309--323. Springer, 2013.

\bibitem[GHR94]{Gabbay1994}
Dov~M. Gabbay, Ian Hodkinson, and Mark Reynolds.
\newblock {\em Temporal Logics: Mathematical Foundations and Computational
  Aspects, Volume 1}.
\newblock Oxford University Press, 1994.

\bibitem[GO03]{Gastin2003}
Paul Gastin and Denis Oddoux.
\newblock {LTL with past and two-way very-weak alternating automata}.
\newblock In {\em Proceedings of MFCS 2003}, volume 2747 of {\em LNCS}, pages
  439--448. Springer, 2003.

\bibitem[GPSS80]{Gabbay1980}
Dov Gabbay, Amir Pnueli, Sharanon Shelah, and J.~Stavi.
\newblock On the temporal analysis of fairness.
\newblock In {\em Proceedings of POPL 1980}, pages 163--173. ACM Press, 1980.

\bibitem[HMP92]{Henzinger1992}
Thomas~A. Henzinger, Zohar Manna, and Amir Pnueli.
\newblock What good are digital clocks?
\newblock In {\em Proceedings of ICALP 1992}, volume 623 of {\em LNCS}, pages
  545--558. Springer, 1992.

\bibitem[Hol97]{Holzmann1997}
Gerard~J. Holzmann.
\newblock The model checker {SPIN}.
\newblock {\em IEEE Transactions on Software Engineering}, 23(5):279--295,
  1997.

\bibitem[HOW13]{Hunter2012}
Paul Hunter, Jo\"{e}l Ouaknine, and James Worrell.
\newblock Expressive completeness of metric temporal logic.
\newblock In {\em Proceedings of LICS 2013}, pages 349--357. IEEE Computer
  Society Press, 2013.

\bibitem[HR01]{Havelund2001}
Klaus Havelund and Grigore Ro{\c{s}}u.
\newblock Testing linear temporal logic formulae on finite execution traces.
\newblock Technical Report {RIACS 01.08}, Research Institute for Advanced
  Computer Science, 2001.

\bibitem[HR04]{Hirshfeld2004}
Yoram Hirshfeld and Alexander~Moshe Rabinovich.
\newblock Logics for real time: Decidability and complexity.
\newblock {\em Fundamenta Informaticae}, 62(1):1--28, 2004.

\bibitem[HR06]{Hirshfeld2006}
Yoram Hirshfeld and Alexander~Moshe Rabinovich.
\newblock An expressive temporal logic for real time.
\newblock In {\em Proceedings of MFCS 2006}, volume 4162 of {\em LNCS}, pages
  492--504. Springer, 2006.

\bibitem[HR07]{Hirshfeld2007}
Yoram Hirshfeld and Alexander Rabinovich.
\newblock Expressiveness of metric modalities for continuous time.
\newblock {\em Logical Methods in Computer Science}, 3(1), 2007.

\bibitem[HRS98]{Henzinger1998}
Thomas~A. Henzinger, Jean-Fran{\c c}ois Raskin, and Pierre-Yves Schobbens.
\newblock The regular real-time languages.
\newblock In {\em Proceedings of ICALP 1998}, volume 1443 of {\em LNCS}, pages
  580--591. Springer, 1998.

\bibitem[Hun13]{Hunter2013}
Paul Hunter.
\newblock When is metric temporal logic expressively complete?
\newblock In {\em Proceedings of CSL 2013}, volume~23 of {\em LIPIcs}, pages
  380--394. Schloss Dagstuhl - Leibniz-Zentrum fuer Informatik, 2013.

\bibitem[Kam68]{Kamp1968}
Johan~A. Kamp.
\newblock {\em Tense logic and the theory of linear order}.
\newblock PhD thesis, University of California, Los Angeles, 1968.

\bibitem[KC05]{Konrad2005}
Sascha Konrad and Betty H.~C. Cheng.
\newblock Real-time specification patterns.
\newblock In {\em Proceedings of ICSE 2005}, pages 372--381. ACM Press, 2005.

\bibitem[KF09]{Kuhtz2009}
Lars Kuhtz and Bernd Finkbeiner.
\newblock {LTL} path checking is efficiently parallelizable.
\newblock In {\em Proceedings of ICALP 2009}, volume 5556 of {\em LNCS}, pages
  235--246. Springer, 2009.

\bibitem[KKP11]{Kini2011}
Dileep Kini, Shankara~N. Krishna, and Paritosh Pandya.
\newblock {On construction of safety signal automata for MITL[U,S] using
  temporal projections}.
\newblock In {\em Proceedings of FORMATS 2011}, volume 6919 of {\em LNCS},
  pages 225--239. Springer, 2011.

\bibitem[Koy90]{Koymans1990}
Ron Koymans.
\newblock Specifying real-time properties with metric temporal logic.
\newblock {\em Real-Time Systems}, 2(4):255--299, 1990.

\bibitem[KV01]{Kupferman2001a}
Orna Kupferman and Moshe~Y. Vardi.
\newblock Model checking of safety properties.
\newblock {\em Formal Methods in System Design}, 19(3):291--314, 2001.

\bibitem[LPZ85]{Lichtenstein1985}
Orna Lichtenstein, Amir Pnueli, and Lenore~D. Zuck.
\newblock The glory of the past.
\newblock In {\em Proceedings of Logics of Programs 1985}, volume 193 of {\em
  LNCS}, pages 196--218. Springer, 1985.

\bibitem[LS09]{Leucker2009}
Martin Leucker and Christian Schallhart.
\newblock A brief account of runtime verification.
\newblock {\em Journal of Logic and Algebraic Programming}, 78(5):293--303,
  2009.

\bibitem[LW08]{Lasota2008}
Slawomir Lasota and Igor Walukiewicz.
\newblock Alternating timed automata.
\newblock {\em ACM Transactions on Computational Logic}, 9(2), 2008.

\bibitem[MN04]{Maler2004}
Oded Maler and Dejan Nickovic.
\newblock Monitoring temporal properties of continuous signals.
\newblock In {\em Proceedings of FORMATS/FTRTFT 2004}, volume 3253 of {\em
  LNCS}, pages 152--166. Springer, 2004.

\bibitem[MNP05]{Maler2005}
Oded Maler, Dejan Nickovic, and Amir Pnueli.
\newblock Real time temporal logic: Past, present, future.
\newblock In {\em Proceedings of FORMATS 2005}, volume 3829 of {\em LNCS},
  pages 2--16. Springer, 2005.

\bibitem[MNP06]{Maler2006}
Oded Maler, Dejan Nickovic, and Amir Pnueli.
\newblock {From MITL to timed automata}.
\newblock In {\em Proceedings of FORMATS 2006}, volume 4202 of {\em LNCS},
  pages 274--289. Springer, 2006.

\bibitem[MP95]{Manna1995}
Zohar Manna and Amir Pnueli.
\newblock {\em {Temporal verification of reactive systems: safety}}, volume~2.
\newblock Springer, 1995.

\bibitem[MS03]{Markey2003}
Nicolas Markey and Philippe Schnoebelen.
\newblock Model checking a path (preliminary report).
\newblock In {\em Proceedings of CONCUR 2003}, volume 2761 of {\em LNCS}, pages
  251--265. Springer, 2003.

\bibitem[NP10]{Nickovic2010a}
Dejan Nickovic and Nir Piterman.
\newblock {From MTL to deterministic timed automata}.
\newblock In {\em Proceedings of FORMATS 2010}, volume 6246 of {\em LNCS},
  pages 152--167. Springer, 2010.

\bibitem[ORW09]{Ouaknine2009}
Jo\"{e}l Ouaknine, Alexander Rabinovich, and James Worrell.
\newblock Time-bounded verification.
\newblock In {\em Proceedings of CONCUR 2009}, volume 5710 of {\em LNCS}, pages
  496--510. Springer, 2009.

\bibitem[OW06]{Ouaknine2006a}
Jo{\"e}l Ouaknine and James Worrell.
\newblock On metric temporal logic and faulty turing machines.
\newblock In {\em Proceedings of FoSSaCS 2006}, volume 3921 of {\em LNCS},
  pages 217--230. Springer, 2006.

\bibitem[OW08]{Ouaknine2008}
Jo\"{e}l Ouaknine and James Worrell.
\newblock Some recent results in metric temporal logic.
\newblock In {\em Proceedings of FORMATS 2008}, volume 5215 of {\em LNCS},
  pages 1--13. Springer, 2008.

\bibitem[OW10]{Ouaknine2010}
Jo{\"e}l Ouaknine and James Worrell.
\newblock Towards a theory of time-bounded verification.
\newblock In {\em Proceedings of ICALP 2010}, volume 6199 of {\em LNCS}, pages
  22--37. Springer, 2010.

\bibitem[PD06]{Prabhakar2006}
Pavithra Prabhakar and Deepak D'Souza.
\newblock {On the expressiveness of MTL with past operators}.
\newblock In {\em Proceedings of FORMATS 2006}, volume 4202 of {\em LNCS},
  pages 322--336. Springer, 2006.

\bibitem[PS11]{Pandya2011}
Paritosh~K. Pandya and Simoni~S. Shah.
\newblock On expressive powers of timed logics: Comparing boundedness,
  non-punctuality and deterministic freezing.
\newblock In {\em Proceedings of CONCUR 2011}, volume 6901 of {\em LNCS}, pages
  60--75. Springer, 2011.

\bibitem[QS82]{Queille1982}
Jean-Pierre Queille and Joseph Sifakis.
\newblock Specification and verification of concurrent systems in {CESAR}.
\newblock In {\em Proceedings of Symposium on Programming 1982}, volume 137 of
  {\em LNCS}, pages 337--351. Springer, 1982.

\bibitem[Rab10]{Rabinovich2010}
Alexander Rabinovich.
\newblock Temporal logics over linear time domains are in {PSPACE}.
\newblock In {\em Proceedings of RP 2010}, volume 6227 of {\em LNCS}, pages
  29--50. Springer, 2010.

\bibitem[Rey10]{Reynolds2010}
Mark Reynolds.
\newblock The complexity of temporal logic over the reals.
\newblock {\em Annals of Pure and Applied Logic}, 161(8):1063--1096, 2010.

\bibitem[Rey14]{Reynolds2014}
Mark Reynolds.
\newblock Metric temporal logics and deterministic timed automata (long report
  version).
\newblock Technical report, University of West Australia, 2014.

\bibitem[Ro{\c{s}}12]{Rosu2012}
Grigore Ro{\c{s}}u.
\newblock On safety properties and their monitoring.
\newblock {\em Scientific Annals of Computer Science}, 22(2):327--365, 2012.

\bibitem[SC85]{Sistla1985}
A.~Prasad Sistla and Edmund~M. Clarke.
\newblock The complexity of propositional linear temporal logics.
\newblock {\em Journal of the ACM}, 32(3):733--749, 1985.

\bibitem[SHL11]{Sokolsky2011}
Oleg Sokolsky, Klaus Havelund, and Insup Lee.
\newblock Introduction to the special section on runtime verification.
\newblock {\em International Journal on Software Tools for Technology
  Transfer}, 14(3):243--247, 2011.

\bibitem[Sto74]{Stockmeyer1974}
Larry Stockmeyer.
\newblock The complexity of decision problems in automata theory and logic.
\newblock {PhD} thesis, {TR} 133, M.I.T., Cambridge, 1974.

\bibitem[TR05]{Thati2005}
Prasanna Thati and Grigore Ro{\c{s}}u.
\newblock Monitoring algorithms for metric temporal logic specifications.
\newblock {\em Electronic Notes in Theoretical Computer Science}, 113:145--162,
  2005.

\bibitem[Tri02]{Tripakis2002}
Stavros Tripakis.
\newblock Fault diagnosis for timed automata.
\newblock In {\em Proceedings of FTRTFT 2002}, volume 2469 of {\em LNCS}, pages
  205--224. Springer, 2002.

\bibitem[Var96]{Vardi1996}
Moshe~Y. Vardi.
\newblock An automata-theoretic approach to linear temporal logic.
\newblock In {\em Logics for Concurrency -- Structure versus Automata (8th
  Banff Higher Order Workshop'95)}, volume 1043 of {\em LNCS}, pages 238--266.
  Springer, 1996.

\bibitem[Wil94]{Wilke1994}
Thomas Wilke.
\newblock Specifying timed state sequences in powerful decidable logics and
  timed automata.
\newblock In {\em Proceedings of FTRTFT 1994}, volume 863 of {\em LNCS}, pages
  694--715. Springer, 1994.

\bibitem[WVS83]{Wolper1983}
Pierre Wolper, Moshe~Y. Vardi, and A.~Prasad Sistla.
\newblock Reasoning about infinite computation paths.
\newblock In {\em Proceedings of FOCS 1983}, pages 185--194. IEEE Computer
  Society Press, 1983.

\end{thebibliography}

\end{document}